\documentclass[10pt,prd,aps,amsfonts,eqsecnum,superscriptaddress,nofootinbib,notitlepage,longbibliography,draft]{revtex4-1}

\usepackage[a4paper, left=2cm, right=2cm, top=2cm, bottom=2cm]{geometry}

\usepackage[english]{babel}
\addto\extrasenglish{%
	\def\appendixautorefname{Appendix}%
}
\usepackage[utf8]{inputenc}
\usepackage[T1]{fontenc}
\usepackage{lmodern}
\usepackage{amsmath,amsthm,mathdots,amssymb,bm,bbm,mathrsfs,mathtools}
\usepackage[final]{graphicx} 
\usepackage{tikz,pgfplots}\pgfplotsset{compat=1.16}

\usepackage{thm-restate}
\newtheorem{thm}{Theorem}[section]
\newtheorem{cor}{Corollary}[section]
\newtheorem{lem}{Lemma}[section]

\newtheorem{fact}{Fact}[section]
\newtheorem{prop}{Proposition}[section]

\newtheorem{defn}{Definition}[section]

\usepackage{comment}

\usepackage{mleftright}\mleftright 

\usepackage{enumitem}
\mathtoolsset{showonlyrefs=true}
\usepackage{microtype,color} 
\usepackage{xcolor}
\usepackage[colorlinks=true, urlcolor=violet, linkcolor=blue, citecolor=red, hyperindex=true, linktocpage=true, pagebackref=true, draft=false]{hyperref}

\renewcommand{\vec}{\bm}
\newcommand{\CA}{\mathcal{A}}
\newcommand{\CB}{\mathcal{B}}

\newcommand{\BC}{\mathbb{C}}
\newcommand{\CD}{\mathcal{D}}

\newcommand{\CF}{\mathcal{F}}

\newcommand{\CH}{\mathcal{H}}
\newcommand{\CI}{\mathcal{I}}

\newcommand{\CL}{\mathcal{L}}
\newcommand{\CM}{\mathcal{M}}
\newcommand{\CN}{\mathcal{N}}

\newcommand{\CO}{\mathcal{O}}

\newcommand{\BR}{\mathbb{R}}
\newcommand{\CS}{\mathcal{S}}
\newcommand{\CT}{\mathcal{T}}

\newcommand{\BZ}{\mathbb{Z}}

\newcommand{\vA}{\bm{A}}

\newcommand{\vB}{\bm{B}}

\newcommand{\vF}{\bm{F}}

\newcommand{\vH}{\bm{H}}
\newcommand{\vI}{\bm{I}}

\newcommand{\vK}{\bm{K}}

\newcommand{\vL}{\bm{L}}

\newcommand{\vM}{\bm{M}}
\newcommand{\vN}{\bm{N}}
\newcommand{\vO}{\bm{O}}
\newcommand{\vP}{\bm{P}}

\newcommand{\vQ}{\bm{Q}}
\newcommand{\vR}{\bm{R}}

\newcommand{\bt}{\bar{t}}
\newcommand{\vU}{\bm{U}}
\newcommand{\vV}{\bm{V}}
\newcommand{\vW}{\bm{W}}
\newcommand{\vX}{\bm{X}}
\newcommand{\vY }{\bm{Y }}
\newcommand{\vZ}{\bm{Z}}

\newcommand{\bbeta}{\bar{ \beta}}
\newcommand{\vrho}{\bm{ \rho}}
\renewcommand{\L}{\left}
\newcommand{\R}{\right}

\newcommand{\bnu}{\bar{\nu}}

\newcommand{\bomega}{\bar{\omega}}
\newcommand{\tOmega}{\tilde{\Omega}}
\newcommand{\tCO}{\tilde{\CO}}

\newcommand{\dagg}{\dagger}

\newcommand{\vertiii}[1]{{\left\vert\kern-0.25ex\left\vert\kern-0.25ex\left\vert #1 \right\vert\kern-0.25ex\right\vert\kern-0.25ex\right\vert}}
\newcommand{\norm}[1]{\Vert {#1} \Vert}

\newcommand{\normp}[2]{\norm{#1}_{#2}}
\newcommand{\lnormp}[2]{\lnorm{#1}_{#2}}

\newcommand{\labs}[1]{\left\vert {#1} \right\vert}
\newcommand{\lnorm}[1]{\left\Vert {#1} \right\Vert}
\newcommand{\e}{\mathrm{e}}

\newcommand{\ri}{\mathrm{i}}
\newcommand{\rd}{\mathrm{d}}
\newcommand*{\tr}{\mathrm{Tr}}
\newcommand*{\poly}{\mathrm{Poly}}

\newcommand{\indicator}{\mathbbm{1}}
\newcommand{\Lword}[1]{\text{Lindbladian}}

\newcommand{\nrm}[1]{\left\| #1 \right\|}

\newcommand{\undersetbrace}[2]{ \underset{#1}{\underbrace{#2}}}

\DeclareMathOperator{\sinc}{sinc}
\DeclareMathOperator{\sgn}{sgn}
\DeclareMathOperator{\erf}{erf}
\DeclareMathOperator{\erfc}{erfc}
\DeclarePairedDelimiterX{\braket}[1]{\langle}{\rangle}{#1}
\DeclarePairedDelimiterX\ketbra[2]{| }{|}{#1 \delimsize\rangle\!\delimsize\langle #2}	
\DeclarePairedDelimiterX\dotp[2]{\langle}{\rangle}{#1, #2}
\newcommand{\bigO}[1]{\mathcal{O}\left( #1 \right)}
\newcommand{\bigOt}[1]{\widetilde{\mathcal{O}}\left( #1 \right)}

\DeclareMathAlphabet{\dutchcal}{U}{dutchcal}{m}{n}
\SetMathAlphabet{\dutchcal}{bold}{U}{dutchcal}{b}{n}
\DeclareMathAlphabet{\dutchbcal} {U}{dutchcal}{b}{n}

\makeatletter
\DeclareRobustCommand*{\pmzerodot}{%
	\nfss@text{%
		\sbox0{$\vcenter{}$}
		\sbox2{0}%
		\sbox4{0\/}%
		\ooalign{%
			0\cr
			\hidewidth
			\kern\dimexpr\wd4-\wd2\relax 
			\raise\dimexpr(\ht2-\dp2)/2-\ht0\relax\hbox{%
				\if b\expandafter\@car\f@series\@nil\relax
				\mathversion{bold}%
				\fi
				$\cdot\m@th$%
			}%
			\hidewidth
			\cr
			\vphantom{0}
		}%
	}%
}

\usepackage{mmacells}

\mmaSet{
	morefv={gobble=2},
	linklocaluri=mma/symbol/definition:#1,
	morecellgraphics={yoffset=0mm},
	leftmargin=11.5mm,
	labelsep=1mm,
}

\usepackage{ifdraft}
\ifdraft{
	\newcommand{\authnote}[3]{{\color{#3} {\bf  #1:} #2}}	
}{
	\newcommand{\authnote}[3]{}
}

\usetikzlibrary{quantikz2}
\makeatletter
\def\l@subsection#1#2{}
\def\l@subsubsection#1#2{}
\makeatother
\begin{document}
\renewcommand{\appendixautorefname}{Appendix}
\renewcommand{\chapterautorefname}{Chapter}
\renewcommand{\sectionautorefname}{Section}
\renewcommand{\subsubsectionautorefname}{Section}	

\title{An efficient and exact noncommutative quantum Gibbs sampler}
\author{Chi-Fang Chen}
\email{achifchen@gmail.com}
\affiliation{Institute for Quantum Information and Matter,
California Institute of Technology, Pasadena, CA, USA}
    \affiliation{AWS Center for Quantum Computing, Pasadena, CA}
    \author{Michael J. Kastoryano}
    \affiliation{AWS Center for Quantum Computing, Pasadena, CA}
    \affiliation{IT University of Copenhagen, Denmark}
    \author{Andr\'as Gily\'en}
    \affiliation{Alfr\'ed R\'enyi Institute of Mathematics, HUN-REN, Budapest, Hungary}
 
\begin{abstract}
Preparing thermal and ground states is an essential quantum algorithmic task for quantum simulation. In this work, we construct the first efficiently implementable and exactly detailed-balanced Lindbladian for Gibbs states of arbitrary noncommutative Hamiltonians. Our construction can also be regarded as a continuous-time quantum analog of the Metropolis-Hastings algorithm. To prepare the quantum Gibbs state, our algorithm invokes Hamiltonian simulation for a time proportional to the mixing time and the inverse temperature $\beta$, up to polylogarithmic factors. Moreover, the gate complexity reduces significantly for lattice Hamiltonians as the corresponding Lindblad operators are (quasi-) local (with radius $\sim\beta$) and only depend on local Hamiltonian patches. Meanwhile, purifying our Lindbladians yields a temperature-dependent family of frustration-free ``parent Hamiltonians''\kern-1mm, prescribing an adiabatic path for the canonical purified Gibbs state (i.e., the Thermal Field Double state). These favorable features suggest that our construction serves as a quantum algorithmic counterpart to classical Markov chain Monte Carlo sampling.
\end{abstract}
\maketitle

\section{Introduction}
One of the leading candidate applications of quantum computers~\cite{dalzell2023QuantumAlgSurvey} is to simulate quantum systems~\cite{feynman1982SimQPhysWithComputers}. In particular, the preparation of thermal states or ground states for materials and molecules has received significant attention~\cite{babbush2018low,Chamberland2020BuildingAF,THC_google,2021_Microsoft_catalysis}. Surprisingly, there has not been a consensus on the ``go-to'' quantum algorithm for this task due to a lack of provable guarantees or empirical evidence~\cite{lee2022there}. Recently, several Monte Carlo-style, nonunitary quantum algorithms have been proposed~\cite{temme2009QuantumMetropolis, yung2010QuantumQuantumMetropolis, wocjan2021szegedy, Rall_thermal_22, chen2023QThermalStatePrep,ding2023single}. While their efficacy has only been validated using small-scale numerics and under strong theoretical assumptions~\cite{ETH_thermalization_Chen21,Shtanko2021AlgorithmsforGibbs,ding2023single}, there are reasons for optimism. On physical grounds, these algorithms resemble naturally occurring system-bath dynamics~\cite{Mozgunov2020completelypositive}; if a system rapidly cools in a refrigerator, the same plausibly applies to a ``cooling algorithm'' that emulates this process. Alternatively, from a computer science perspective, these algorithms are cousins of classical Markov chain Monte Carlo (MCMC) methods but with quantum mechanical effects and complications. This work sets out to complete this line of thought and construct an ideal quantum MCMC algorithm where the robustness, simplicity, and empirical success of the classical case may be transferable. 

The cornerstone of classical Markov chain Monte Carlo methods is \textit{detailed balance} (see, e.g.,~\cite{Markovchain_mixing}): given a target state $\pi$, we impose a certain symmetry of the Markov chain $\vM$ 
\begin{align}
    \vM_{s's} \pi_{s} = \pi_{s'} \vM_{s's} \quad \text{for each configuration} \quad s,s', \quad \text{ensuring stationarity} \quad \vM[\pi] = \pi.
\end{align}
This simple recipe for the stationary state has been crucial in constructing and analyzing the Metropolis-Hastings algorithm and related Markov chains. 
Notably, detailed balance gives a conceptually simple picture of convergence via the spectral gap of $\vM$, a quantity amenable to numerical and analytic bounds. If the problem at hand has a local structure, detailed balance can often be imposed locally and efficiently, relegating the algorithm's complexity to the \textit{mixing time}, the time scale of convergence towards stationarity. While the mixing time may be challenging to analyze, MCMC methods can often be employed heuristically. 
In particular, we are often interested in sampling the \textit{Gibbs distribution} $\pi_\beta \propto e^{-\beta H}$ of a certain energy functional $H$ at temperature $1/\beta$. Analogously, the central idea of \textit{quantum Gibbs sampling} is to construct a detailed-balanced quantum process where the quantum Gibbs state is stationary. In this work, we focus on designing a \Lword{} $\CL_{\beta}$ (the quantum analog of a continuous-time Markov chain generator) such that
\begin{align}
\e^{\CL_{\beta}t}[\vrho_{\beta}]=\vrho_{\beta} \quad\text{where}\quad \vrho_\beta := \e^{-\beta \vH}/\tr(\e^{-\beta \vH})\label{eq:main_Gibbs_fixed_point},  
\end{align}
for \textit{any} target quantum Hamiltonian $\vH$. 
As in the classical case, we can prepare samples of quantum Gibbs states if the Lindbladian evolution can be \textit{efficiently} implemented and the state converges \textit{rapidly} to the Gibbs state.

The main issue with existing quantum Gibbs sampling algorithms is that \textit{quantum detailed balance} (Figure~\ref{fig:QuantumDB})  only holds \textit{approximately} unless we can distinguish individual energy eigenstates exactly, which is generally intractable except for fast-forwardable Hamiltonians (e.g., Hamiltonians with commuting terms). Consequently, we either lose accuracy guarantees for the stationary state or the efficiency for the individual steps of the Gibbs sampling algorithm, leading to significant aggregated complexity plaguing various constructions; see Ref.~\cite{chen2023QThermalStatePrep} for a comprehensive catalog. The algorithmic challenge in enforcing quantum detailed balance is the energy-time uncertainty principle rooted in \textit{metrology}: for each energy estimate, the uncertainty scales inversely proportional to the Hamiltonian simulation time. Indeed, all existing quantum MCMC algorithms attempt to attain detailed balance via an ``energy estimation'' subroutine (quantum phase estimation~\cite{temme2009QuantumMetropolis,wocjan2021szegedy, Rall_thermal_22} or operator Fourier Transform~\cite[Appendix A]{chen2023QThermalStatePrep}). 
Consequently, this error propagates to the desired Gibbs state and impacts the implementation cost. 

To our knowledge, the best general lower-bound on the Hamiltonian simulation time is $\Omega(\beta)$ per Gibbs sample~\cite[Proposition G.5]{chen2023QThermalStatePrep}. This comes from a sensitivity argument that the Gibbs state is a smooth matrix function of $\vH$ with derivatives bounded by $\CO( \beta)$. This conceptual gap motivates our guiding question: 
\begin{align}
    \textit{Can we design an efficiently implementable yet exactly detailed-balanced quantum Gibbs sampler?}
\end{align}
If so, we may recover both the simplicity and versatility of classical MCMC algorithms. In this work, we answer this question in the affirmative by explicitly constructing an exactly detailed-balanced \Lword{} at a moderate cost: $\tCO(\beta)$-Hamiltonian simulation cost per \textit{unit time} of Linbladian evolution $\e^{\CL}$. (The unit-time evolution for a continuous-time generator should be regarded as ``one step,'' corresponding to one discrete Markov chain update.) Furthermore, for lattice Hamiltonians (with local jumps), our Lindbladian is (quasi-)local with locality scaling as $\tCO(\beta)$. Thus, one step of the algorithm only needs to simulate localized Hamiltonian patches; this starkly contrasts with previous works, whose cost per unit time step generally scales with the system size due to simulating the \textit{global} Hamiltonian. The mathematical and conceptual simplicity of our result immediately initiates a list of new directions, which we discuss in detail in~\autoref{sec:discussion}.

The key revelation behind our construction is that quantum detailed balance can be enforced \textit{smoothly} without ever knowing the energy. Indeed, the standard \textit{metrology} lower bound $\sim\Omega(\frac{1}{\epsilon})$ is not an obstruction because having access to a detailed-balanced \Lword{} (or the Gibbs state) does not give energy estimates. We have seen that Quantum Signal Processing~\cite{low2016HamSimQSignProc} or Quantum Singular Value Transform (QSVT)~\cite{gilyen2018QSingValTransf} allows one to directly access smooth (polynomial) functions of Hamiltonians without a phase-estimation subroutine; the costs often scale linearly with the largest derivatives and only \textit{logarithmically} with the precision. However, Lindbladians, as superoperators, are more restrictive to manipulate than Hermitian matrices. A key design ingredient is a carefully chosen \textit{coherent} term in our \Lword{}
\begin{align}
    \CL_{\beta}[\vrho] = \underset{\text{``coherent''}}{\underbrace{-\ri [\vB, \vrho]}} + (\text{``dissipative''}),
\end{align}
which appears necessary to \textit{coherently} and \textit{exactly} cancel out certain unwanted errors from the dissipative part.

As a by-product, purifying our Lindbladian yields a temperature-dependent family of ``parent Hamiltonians'' whose zero-eigenstate is a canonical purification of the Gibbs state.\footnote{This purification coincides with the \textit{Thermal Field Double} state featured in recent quantum gravity discussions. See, e.g.,~\cite{Maldacena2013CoolHF}.} 
Similarly to how classical Markov chains can be ``quantized'' to prepare the purified stationary state, here we prepare the purified Gibbs state by following a prescribed adiabatic path (called \textit{quantum simulated annealing}~\cite{szegedy2004QMarkovChainSearch,yung2010QuantumQuantumMetropolis}), drawing a surprisingly simple connection between thermal dissipation and adiabatic evolution. In particular, for lattice Hamiltonians, the parent Hamiltonian inherits the (quasi-)locality, which curiously connects the purified Gibbs state to the ground state of (quasi-)local Hamiltonians. 

\begin{figure}[t]
\includegraphics[width=0.8\textwidth]{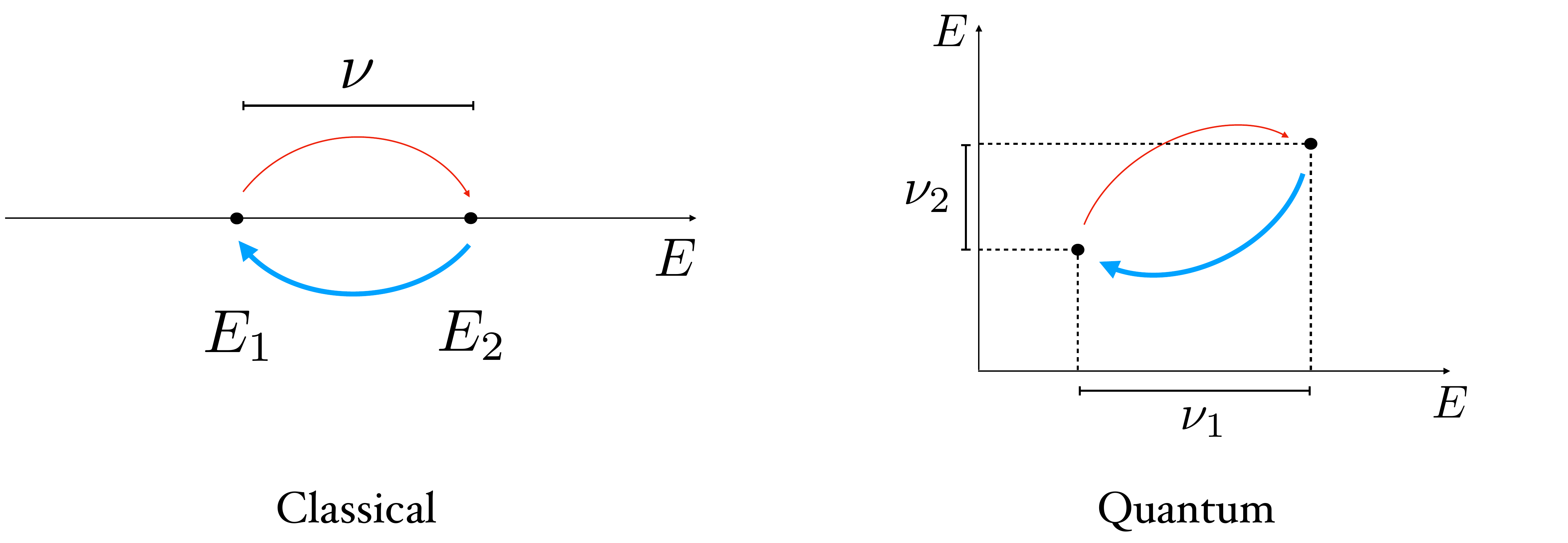}
\caption{
(Left) For the classical Gibbs distribution, the detailed balance condition is a pairwise relation between heating (red) and cooling (blue) transition rates, depending on the energy difference $\nu$ of states. (Right) For the quantum Gibbs state, the detailed balance condition refers to pairs of matrix elements of the density operator (expanded in the energy basis), where each matrix element is described by a pair of energies (of the basis elements in the ket and bra respectively) therefore the relation depends on both of the respective energy differences $\nu_1$ and $\nu_2$.
}\label{fig:QuantumDB}
\end{figure}

\subsection{Main results}

Our main results, based on the algorithmic framework of~\cite[Section III]{chen2023QThermalStatePrep}, consider the following \Lword{} in the Schrodinger picture
\begin{align}
		\CL_{\beta}[\cdot] := \underset{\text{``coherent''}}{\underbrace{-\ri [\vB, \cdot]}} + \sum_{a\in A} 
		\int_{-\infty}^{\infty} \gamma(\omega) \left(\underset{\text{``transition''}}{\underbrace{\hat{\vA}^a(\omega)(\cdot)\hat{\vA}^{a}(\omega)^\dagg}} - \underset{\text{``decay''}}{\underbrace{\frac{1}{2}\{\hat{\vA}^{a}(\omega)^\dagg\hat{\vA}^a(\omega),\cdot\}}}\right)\rd\omega\label{eq:exact_DB_L}
\end{align}
which is parameterized by the following terms (with convenient normalization conditions~\cite[Section I.B]{chen2023QThermalStatePrep}): 
\begin{itemize}
    \item The distinct \textit{jump operators} $\vA^a$ ``drive'' the transitions, and can be chosen arbitrarily as long as their adjoints are included
\begin{align}
    \{\vA^a\colon a\in A\}=\{\vA^{a\dagg}\colon a\in A\}\quad \text{and}\quad \nrm{\sum_{a\in A}\vA^{a\dagg} \vA^a}\leq 1.\label{eq:AAdagger}
\end{align}
For better mixing and ergodicity, the jumps should be ``scrambling'' and not commute with the Hamiltonian (e.g., breaking the symmetries of the Hamiltonian). For lattice Hamiltonians, the jump operators may be chosen simply to be the single-site Pauli operators, but global jumps could also be helpful in some cases, as in the classical case (e.g., cluster updates).

    \item The \textit{Operator Fourier Transform} (OFT)\footnote{Note the sign convention, which might differ from that of other works in the literature. } (\autoref{sec:OFT}) weighted by a Gaussian \textit{filter} with a tunable width $\sim \sigma_E^{-1}$
    \begin{align}\label{eq:OpOFT}
	\hat{\vA}^a(\omega) :=  \frac{1}{\sqrt{2\pi}}\int_{-\infty}^{\infty} \e^{\ri \vH t} \vA^a \e^{-\ri \vH t} \e^{-\ri \omega t} f(t)\rd t \quad \text{where}\quad     f(t)&:=\e^{-\sigma_E^2 t^2}\sqrt{\sigma_E \sqrt{2/\pi}}\\
 &=\frac{\e^{-t^2/\beta^2}}{\sqrt{\beta \sqrt{\pi/2}}} \quad \text{if}\quad \sigma_E = \frac{1}{\beta}. 
\end{align}
In particular, the Gaussian is normalized $\int_{-\infty}^{\infty} \labs{f(t)}^2\rd t =1$. Naturally, the Heisenberg evolution $\e^{\ri \vH t} \vA^a \e^{-\ri \vH t}$ diagnoses the energy difference $\omega$ before and after the jump. Integrating over time $\frac{1}{\sqrt{2\pi}}\int_{-\infty}^{\infty} (\cdot)\e^{-\ri \omega t} f(t)\rd t $ yields the operator Fourier Transform $\hat{\vA}^a(\omega)$, which selects the transitions of $\vA^a$ that \textit{increase} the energy by roughly $\sim \omega \pm \CO(\sigma_E)$. At first glance, the Gaussian filter seems to merely ensure good concentration for both the frequency and time domain, but it turns out to have a more intimate connection~\cite{Moussa2019LowDepthQM} to quantum detailed balance (see also \autoref{sec:why_gaussian} for an alternative justification).

\item The \textit{transition weight} $\gamma(\omega)$ follows (yet another) Gaussian with a tunable variance $\sigma_{\gamma}>0$:
\begin{align}
	\gamma(\omega) &= \exp\L(- \frac{(\omega + \omega_{\gamma})^2}{2\sigma_{\gamma}^2}\R)\quad \text{with variance}\quad \sigma_{\gamma}^2 := \frac{2\omega_{\gamma}}{\beta}-\sigma_E^2 \label{eq:ExctDissip}\\
 &= \exp\L(- \frac{(\beta\omega + 1)^2}{2}\R)\quad  \text{if}\quad  \sigma_E = \sigma_{\gamma} = \omega_{\gamma} = \frac{1}{\beta}.\label{eq:gaussian_gamma_beta}
\end{align}
The normalization is such that $\norm{\gamma(\omega)}_{\infty} \le 1$, and the maximum is attained at $\omega = -\omega_{\gamma}$. 
\item A coherent (i.e., nondissipative) term generated by a fine-tuned Hermitian matrix $\vB$. The expression depends on $\omega_{\gamma}, \sigma_E, \beta$ in the general case (see \autoref{cor:GaussianLikeWeight}), but it simplifies to
  \begin{align}
		\kern-7mm\vB&\!:= \sum_{a\in A} 
		\int_{-\infty}^{\infty}\!b_1(t)\e^{-\ri\beta\vH t} \L(\int_{-\infty}^{\infty}b_2(t')\e^{\ri\beta\vH t'}\vA^{a\dagger}\e^{-2\ri\beta\vH t'}\vA^a\e^{\ri\beta\vH t'}\rd t' \R)\e^{\ri\beta\vH t}\rd t\quad \text{if}\quad \omega_{\gamma}= \sigma_E = \sigma_{\gamma} =\frac{1}{\beta}\quad\quad \label{eq:mainBDef}
	\end{align}
 for some carefully chosen smooth and rapidly decaying functions $b_1,b_2$ normalized by $\norm{b_1}_1,\norm{b_2}_1 \le 1$. The coherent term $\vB$ may appear intimidating but plays a crucial role in ensuring \textit{quantum detailed balance} for the Gibbs state $\vrho_\beta$, defined as:
\begin{align}
    \CL_{\beta}^{\dagger}[\cdot] = \sqrt{\vrho_{\beta}}^{-1}\CL_{\beta}[ \sqrt{\vrho_{\beta}} \cdot \sqrt{\vrho_{\beta}}]\sqrt{\vrho_{\beta}}^{-1}\quad \text{for fixed}\quad \beta, \vH.\label{eq:main_DB}
\end{align}
It implies the stationarity of Gibbs state exactly (see~\autoref{defn:DB}). 
\end{itemize}

The Gaussian transition weight~\eqref{eq:ExctDissip} is inspired by an observation of~\cite{Moussa2019LowDepthQM}:\footnote{Their algorithm~\cite{Moussa2019LowDepthQM} seems qualitatively different from Monte Carlo style quantum algorithms~\cite{temme2009QuantumMetropolis, yung2010QuantumQuantumMetropolis, wocjan2021szegedy, Rall_thermal_22, chen2023QThermalStatePrep} and closer to performing phase estimation on trial states; see the discussion in \cite[Page 5]{Moussa2019LowDepthQM}.} the functional form of Gaussians is naturally compatible with exact detailed balance\footnote{We thank Jonathan Moussa for pointing us to his paper and raising the question of whether detailed balance can hold exactly in the precursor of this work~\cite{chen2023QThermalStatePrep}.} if we make conscious choices of $\omega_{\gamma}, \sigma_E$
\begin{align}
    \exp\L( - \frac{(\omega+\omega_{\gamma})^2}{2\sigma^2}\R) = \exp(-\frac{2\omega_{\gamma}}{\sigma^2} \cdot \omega) \exp\L(- \frac{(-\omega+\omega_{\gamma})^2}{2\sigma^2}\R).
\end{align}
Compared with the usual step-function-like Metropolis weight $\min(1,\e^{-\beta \omega})$, the Gaussian weight is more selective, only allowing energy transitions $-\omega_\gamma\pm\CO(\sigma_\gamma)$; this narrower window could potentially freeze the dynamics, leading to a long mixing time.

Fortunately, quantum detailed balance~\eqref{eq:main_DB} is preserved under \textit{linear combination} of \Lword{}s; hence, choosing a linear combination of $\gamma$ covering a range of different widths $\sigma_\gamma$ can remove the heavy restriction on energy transitions. Surprisingly, a suitable linear combination recovers Metropolis-like transition weights, which we focus on as the representative. To obtain the corresponding exactly detailed-balanced \Lword{}, the only change compared to \eqref{eq:ExctDissip}-\eqref{eq:mainBDef} is the choice of transition weight
\begin{align}\label{eq:Metropolish2}
    \text{(Metropolis-Style)}\quad\gamma^M(\omega) := \exp\L(-\beta\max\left(\omega+\frac{1}{2\beta},0\right)\R) \quad \text{if}\quad \sigma_E = \frac{1}{\beta}
\end{align}
with the corresponding coherent term $\vB^M$ parameterized by another function $b_2^M(t)$ (the function $b^M_1(t)=b_1(t)$ remains the same as in \eqref{eq:mainBDef}).\footnote{The generalized function (distribution) $b_2^M(t)$ should be interpreted as the Cauchy principal value $\lim_{\eta\rightarrow 0+}\indicator(|t|>\eta)b_2^M(t)$.
	In case $[\vH,\sum_{a\in A} \vA^{a\dagger}\vA^a]\neq 0$ an additional correction term $\frac{1}{16\sqrt{2}\pi}\delta(t)$ should be added.} {In fact, it was later shown that in general it suffices if $\gamma(\omega) e^{\beta\omega + \sigma^2/2} =\gamma(-\omega-\sigma^2\beta^2)$, i.e., one may take any $\gamma_0(\nu)e^{\beta \nu} = \gamma_0(-\nu)$ and apply a shift $\gamma(\omega)=\gamma_0(\omega+\sigma^2\beta^2/2)$~\cite[Lemma 7.1]{ramkumar2024mixing}.}

Now, we present the first main result: the Gibbs state is an \textit{exact} stationary state of the advertised \Lword{} (see~\autoref{sec:prove_exactDB} for the proof). Although we have mainly focused on the Gibbs state, we can formally invoke Gibbs sampling for $\beta \vH = \log(\vrho_{fix})$ for an arbitrary target stationary state $\vrho_{fix}$, albeit with potential overhead from implementing the matrix logarithm.
\begin{thm}[Gibbs state is stationary]\label{thm:Exact_fixedpoint}
For any $\beta \ge 0$, the \Lword{}~\eqref{eq:exact_DB_L}-\eqref{eq:OpOFT} with $\sigma_E = \frac{1}{\beta}$, Gaussian transition wieght~\eqref{eq:gaussian_gamma_beta}, and the coherent term $\vB$~\eqref{eq:mainBDef} satisfies detailed balance~\eqref{eq:main_DB} exactly. Therefore, the Gibbs state is stationary
\begin{align}
    \CL_{\beta}[\vrho_{\beta}] =0.
\end{align}
The same applies for the Metropolis transition weight $\gamma^M(\omega)$~\eqref{eq:Metropolish2} with the corresponding coherent term $\vB^M$. 
\end{thm}
 Furthermore, building on the algorithmic machinery developed in Ref.~\cite[Section III]{chen2023QThermalStatePrep}, the \Lword{} can be efficiently simulated at a moderate cost (see~\autoref{sec:prove_L_cost} for the proof).
\begin{thm}[Efficient implementation]\label{thm:L_cost}
Instantiate the \Lword{} parameters of~\autoref{thm:Exact_fixedpoint} for either the Gaussian $\gamma(\omega)$ or Metropolis $\gamma^M(\omega)$ transition weight. Then, the Lindbladian evolution 
\begin{align}
\e^{\CL_{\beta}t} \quad \text{for each} \quad t \ge 1 
\end{align}
can be implemented efficiently in $\epsilon$-diamond distance with cost:
\begin{align}
    &\tCO(t \cdot \beta) \quad \text{(total Hamiltonian simulation time)}\\
    &\tCO(t) \quad \text{(block-encodings for the jumps $\sum_{a\in A} \ket{a} \otimes \vA^a$)},
\end{align}
{and additionally to those in the block-encoding and the Hamiltonian simulation,
\begin{align}
    &\tCO(1) \quad \text{(resettable ancilla)}\\
    &\tCO(t) \quad \text{(two-qubit gates)}.
\end{align}}
The $\tCO(\cdot)$ notation absorbs polylogarithmic dependencies on $t,\beta, \norm{\vH}, n, 1/\epsilon, \labs{A}$. 
\end{thm}
Here, our Lindbladian is normalized~\eqref{eq:AAdagger},~\eqref{eq:OpOFT},~\eqref{eq:ExctDissip},~\eqref{eq:mainBDef} (or the Metropolis-like weight~\eqref{eq:Metropolish2} with its coherent term $\vB^M$) such that 
\begin{align}
    \norm{\CL_{\beta}}_{1-1} = \tCO(1).
\end{align}
Therefore, evolving for unit time $t=1$ corresponds to a $\tCO(1)$-strength update $\e^{\CL_{\beta}}$ and only requires a characteristic Hamiltonian simulation time $\sim \beta$. This is precisely the cost for implementing both the operator Fourier Transform~\eqref{eq:OpOFT} and the coherent term~\eqref{eq:mainBDef} via Linear Combination of Unitaries. To be more careful, the $\tCO(\cdot)$ notation also includes polylogarithmic factors due to discretization, truncation, and Hamiltonian simulation error, which is typical in quantum algorithms. To clarify, the idealized map $\CL_{\beta}$ remains exactly detailed-balanced (the nice object to analyze), and the algorithmic implementation error can be made arbitrarily small given the desired runtime.

We expect the \textit{total Hamiltonian simulation time} to be the figure of merit for the algorithmic cost, among others. The jump operators $\vA^a$ can be as simple as Pauli operators, but we consider a black-box query model in case more complex jumps are needed for faster mixing.  In the last line, the use of the other two-qubit gates comes from easier-to-implement unitaries, including the Quantum Fourier Transform, state preparation unitary for the Gaussian filter $\ket{f}$, controlled transition weight; see~\autoref{sec:block_encodings} and~\cite[Section III.B]{chen2023QThermalStatePrep}. 

Combining~\autoref{thm:Exact_fixedpoint} and~\autoref{thm:L_cost}, we can prepare the Gibbs state by simulating the Lindbladian until convergence, resulting in the cost
\begin{align}
	(\text{total Hamiltonian simulation time per Gibbs sample}) = \bigOt{ t_{mix}(\CL_{\beta}) \cdot \beta}.
\end{align}
Formally, the \textit{mixing time} $t_{mix}(\CL_{\beta})$ quantifies the shortest time scale that any two input states become indistinguishable (see~\autoref{prop:mixing_time_from_gap}). To obtain end-to-end gate complexities, we should also instantiate the Hamiltonian simulation cost, a subroutine whose complexity for various systems has been thoroughly studied. For example, on $D$-dimensional lattices with local jumps $\vA^a$, we expect the actual cost to be (up to logarithmic error dependence)
\begin{align}
   (\text{gate complexity per unit evolution time}) \sim   \underset{\text{Ham. sim. time}}{\underbrace{\beta}} \times \underset{\text{volume}}{\underbrace{(v_{LR}\beta)^{D}}},\tag*{(spatially local Hamiltonians)}
\end{align}
where $v_{LR}$ is the Lieb-Robinson velocity and $v_{LR} \beta$ is roughly the radius of the Heisenberg evolution $\vA^a(\cdot)$ at time $\sim\beta$. Indeed, the Lindbladian is a sum over quasi-local Lindbladian operators (Figure~\ref{fig:localized})
\begin{align}
    \CL_{\beta} = \sum_{a\in A}  \CL_{\beta}^{a}\quad \text{each centered at}\quad \vA^a\quad \text{with radius}\quad \tCO( v_{LR} \beta ).
\end{align}
In particular, the cost per unit Lindbladian evolution time is essentially \textit{independent} of the system size (up to logarithmic dependencies), as we only need to simulate the Hamiltonian patch surrounding each jump $\vA^a$.\footnote{In fact, we can further parallelize the Lindbladian evolution to improve the circuit depth; see~\autoref{sec:arealaw}.} 

\begin{figure}[t]
\includegraphics[width=0.8\textwidth]{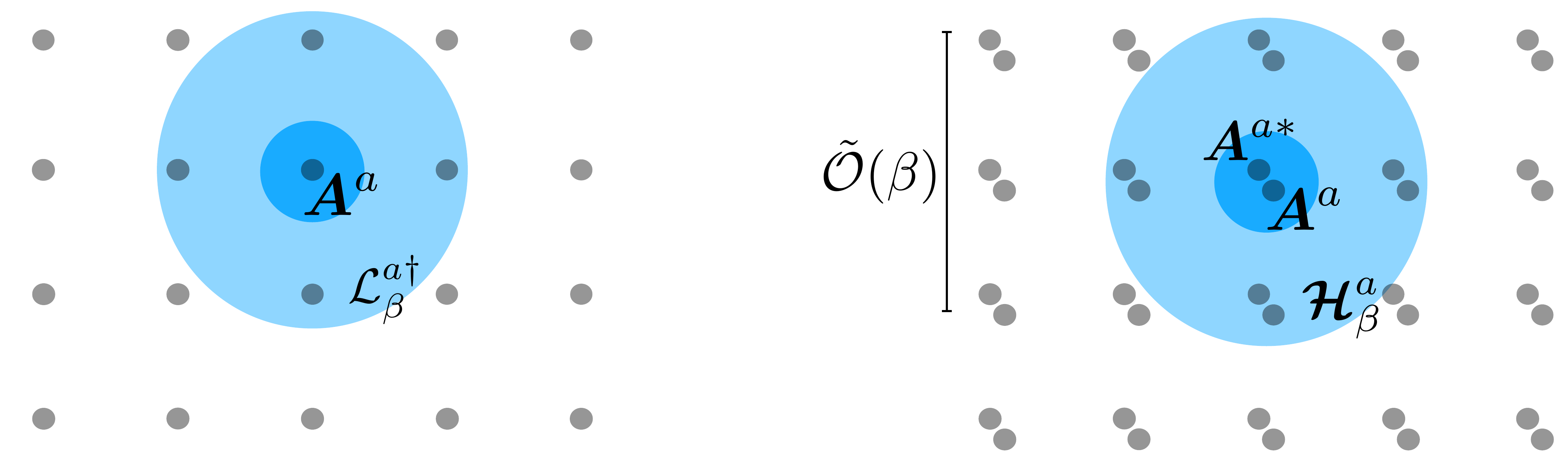}
\caption{
{(Left) For lattice Hamiltonians, our \Lword{} is a sum of quasi-local terms $\CL_{\beta}^{a}$ localized around each jump $\vA^a$ with radius $\tCO(\beta)$. Indeed, detailed balance is really about the \textit{energy difference}, which can be diagnosed by Fourier Transforming the Heisenberg evolution $\vA^a(t) = \e^{\ri \vH t}\vA^a \e^{-\ri \vH t}$. Due to the Lieb-Robinson bounds, the localized \Lword{} terms effectively only depend on the local Hamiltonian patch nearby (up to exponentially decaying tail). (Right) This locality persists after purification, where two copies of the system are glued together.}\label{fig:localized}
}
\end{figure}

\subsubsection{Purifying the \Lword{}s}
We may purify the \Lword{} to prepare the \textit{purified Gibbs state}~\cite{wocjan2021szegedy}
\begin{align}
    \ket{\sqrt{\vrho_{\beta}}} &:= \frac{1}{\sqrt{\tr[ \e^{-\beta \vH }]}} \sum_i \e^{-\beta E_i/2} \ket{\psi_i} \otimes \ket{\psi^*_i};\label{eq:purified_Gibbs}
\end{align}
this particular purification is reminiscent of the quantum walk formalism for detailed-balanced classical Markov chains. The relevant ``parent Hamiltonian'' (playing the role of a ``quantum walk'' operator) is the \textit{discriminant} associated with the \Lword{}; the expression may appear intimidating but resembles how classical detailed-balanced Markov chains are quantized
\begin{align}
\CH_{\beta}:=  \vrho^{-1/4}\CL[ \vrho^{1/4}\cdot  \vrho^{1/4}]  \vrho^{-1/4}.
\end{align}
The above should be regarded as the \Lword{} under similarity transformation. Further, to implement the quantum walk, the superoperator needs to be \textit{vectorized} into an operator on duplicated Hilbert spaces
\begin{align}
    (\text{vectorization})\quad \CH_{\beta}: \CB(\BC^{2^n})\rightarrow \CB(\BC^{2^n}) \simeq \vec{\CH_{\beta}} \in \CB(\BC^{2^n}\otimes \BC^{2^n}).
\end{align}
Indeed, the Gibbs state $\vrho_{\beta}$, as the \Lword{} stationary state, corresponds to the purified Gibbs state $\ket{\sqrt{\vrho_{\beta}}}$~\eqref{eq:purified_Gibbs}, as a zero-eigenvector of the discriminant. As a sanity check, the quantum detailed balance condition~\ref{eq:main_DB} naturally ensures that the operator is Hermitian $\vec{\CH_{\beta}} =\vec{\CH_{\beta}}^{\dagger}.$ 
\begin{prop}[Purifying \Lword{}s]\label{prop:exact_discriminant_fixedpoint}
    Instantiate the \Lword{} parameters of~\autoref{thm:Exact_fixedpoint} for the Gaussian or Metropolis transition weight.
    Then, the corresponding discriminant $\vec{\CH_{\beta}}$ is Hermitian, frustration-free, and annihilates the purified Gibbs state.
    \begin{align}
\vec{\CH_{\beta}} = \sum_{a\in A} \vec{\CH^a_{\beta}}\quad \text{such that}\quad \vec{\CH^a_{\beta}}\ket{\sqrt{\vrho_{\beta}}} = 0\quad \text{for each}\quad a\in A.
    \end{align}
\end{prop}
In other words, preparing the purified Gibbs state boils down to the \textit{ground state} problem (up to a negative sign) for a \textit{frustration-free} parent Hamiltonian parameterized by $\beta$\footnote{Strictly speaking, originating from a Lindbladian, here the parent Hamiltonian is negative semi-definite, and the purified Gibbs state is the top-eigenstate. Introducing a global negative sign will make it the ground state.}; for lattice Hamiltonians with local jumps $\vA^a$, the parent Hamiltonian inherits the \textit{quasi-locality} of our \Lword{}, with individual terms of radius $\sim v_{LR}\beta.$ The algorithmic cost to implement the parent Hamiltonian is analogous to the \Lword{} case, where $\vH_{\beta}$ roughly corresponds to a constant-time Lindbladian evolution.

\begin{thm}[Block encodings for the discriminants]\label{thm:D_cost} 
    Instantiate the \Lword{} parameters of~\autoref{thm:Exact_fixedpoint} for either the Gaussian or Metropolis transition weight. Then, the corresponding discriminant {$\frac{1}{2} \vec{\CH_{\beta}}$} can be block-encoded approximately in $\epsilon$-spectral norm using
\begin{align}
    &\tCO(\beta) \quad \text{(Hamiltonian simulation time)}\\
    &\tCO(1) \quad \text{(block-encodings for the jumps $\sum_{a\in A} \ket{a} \otimes \vA^a$ and its transposes $\sum_{a\in A} \ket{a} \otimes (\vA^{a})^T$)},
\end{align}
{and additionally to those in the block-encodings and the Hamiltonian simulation,
\begin{align}
    &\tCO(1) \quad \text{(resettable ancilla)}\\
    &\tCO(1) \quad \text{(two-qubit gates)}.
\end{align}}
The $\tCO(\cdot)$ notation absorbs {poly}logarithmic dependencies on $t,\beta, \norm{\vH}, n, 1/\epsilon, \labs{A}$.
\end{thm}
Like in~\autoref{thm:L_cost}, we have expanded many other unitaries in the last line; see~\autoref{sec:block_encodings}. Note that we had to downscale the discriminant to implement the block encodings; see~\autoref{sec:prove_D_cost} for the proof. For readers familiar with quantum walks, obtaining a block-encoding for $\vI + \vec{\CH_{\beta}}$ (instead of $\vec{\CH_{\beta}}$) would be nicer as it leads to a quadratic speedup in terms of the spectral gap of $\vec{\CH_{\beta}}$. However, we have not found such a direct block-encoding for our particular construction, leaving this an open problem. This is in contrast with our earlier approximate construction~\cite{wocjan2021szegedy,chen2023QThermalStatePrep}, where such a direct block-encoding was possible, which might be advantageous in some cases.

\subsection{Roadmap}
The remaining body of text is organized thematically by the analysis and the algorithm. We begin with the analysis (\autoref{sec:analysis}) regarding how we design our \Lword{} to satisfy quantum detailed balance exactly. Next, we give efficient algorithms (\autoref{sec:algorithm}) in terms of modularized block encodings to implement the advertised \Lword{} and its purification. 

In the appendix, we include a dictionary of notations (\autoref{sec:recap_notation}) and independent expositions: A connection between the discriminant gap, mixing time, and the area law of entanglement (\autoref{sec:arealaw}) and an alternative heuristic derivation of our detailed balance Lindbladian in the time domain (\autoref{sec:why_gaussian}).

\section{Analysis}\label{sec:analysis}
In this section, we execute the calculations circling the exact detailed balance condition. 
First, we review the operator Fourier Transform in the frequency domain. Second, we review the notion of detailed balance, including the stationary state and spectral theory of convergence. Third, we plug in the advertised functional forms and derive the required coherent term $\vB$ for achieving detailed balance.

\subsection{Operator Fourier Transform}
\label{sec:OFT}
Our Lindbladian features the operator Fourier Transform $\hat{\vA}(\omega)$ of a jump operator $\vA$ according to the Hamiltonian $\vH$ (dropping the jump label for this section). To analyze it, we need to consider the frequency domain representation instead of the time domain. Decompose the operator in the energy basis, and regroup in terms of the energy change (called the \textit{Bohr frequencies} $\nu \in B(\vH) := \{ E_i - E_j \, | \, E_i, E_j \in \mathrm{Spec}(\vH) \}$)
\begin{align}\label{eqn:Aoperator}
	\vA&:= \sum_{E_1,E_2 \in \mathrm{spec}(\vH)} \vP_{E_2} \vA \vP_{E_1} = \sum_{\nu \in B(\vH)} \sum_{E_2-E_1 = \nu} \vP_{E_2} \vA \vP_{E_1}
 =: \sum_{\nu \in B(\vH)} \vA_{\nu}\quad \text{such that}\quad (\vA_{\nu})^{\dagger} = (\vA^{\dagger})_{-\nu}
\end{align}
where $\vP_{E}$ denotes the orthogonal projector onto the eigensubspace of $\vH$ with \textit{exact} energy $E$. This decomposition naturally solves the Heisenberg evolution
\begin{align}
\e^{\ri \vH t} \vA \e^{-\ri \vH t} &= \sum_{\nu\in B} \vA_{\nu}\e^{\ri \nu t}\quad \text{since}\quad [\vH,\vA_{\nu}] = \nu \vA_{\nu}.
\end{align}
Indeed, the energy differences $\nu \in B$ (shorthand of $B(\vH)$) naturally arise from the commutator (as opposed to the absolute energies $E \in \mathrm{spec}(\vH)$).

At first glance, the reference to the exact energies seems unphysical as each of them individually requires a long (likely exponential) Hamiltonian simulation time to access algorithmically. Fortunately, all that we are manipulating are the \textit{smooth} weights on these Bohr frequencies; indeed, the operator Fourier Transform can be conveniently expressed by
\begin{align}\label{eq:BohrDecomposedOFT}
	\hat{\vA}_{f}(\omega) &=  \frac{1}{\sqrt{2\pi}}\int_{-\infty}^{\infty} \e^{\ri \vH t} \vA \e^{-\ri \vH t} \e^{-\ri \omega t} f(t)\rd t
 =\sum_{\nu \in B} \vA_{\nu} \hat{f}(\omega-\nu),
\end{align}
where $\hat{f}(\omega)=\frac{1}{\sqrt{2\pi}}\int_{-\infty}^{\infty} f(t) \e^{-\ri \omega t}\rd t$ is the Fourier Transform of the filter function $f(t)$. Our choice of $f(t)$ is
\begin{align}
    f(t):=\e^{-\sigma_E^2 t^2}\sqrt{\sigma_E \sqrt{2/\pi}}
    \quad \text{such that}\quad \hat{f}(\omega)=\frac{1}{\sqrt{\sigma_E\sqrt{2\pi}}} \exp\L(- \frac{\omega^2}{4\sigma_E^2}\R) \quad \text{and} \quad \int_{-\infty}^{\infty} \labs{f(t)}^2\rd t = 1, 
\end{align}
therefore $\hat{\vA}(\omega)$ becomes simply
\begin{align}\label{eq:OpFTDef}
	\hat{\vA}(\omega)&=\frac{1}{\sqrt{\sigma_E\sqrt{2\pi}}}\sum_{\nu \in B} \exp\L(- \frac{(\omega-\nu)^2}{4\sigma_E^2}\R) \vA_{\nu}.
\end{align}
The above uses the Gaussian integrals, which will also be constantly recalled.
\begin{fact}[Gaussian integrals]\label{fact:gaussian_int}
    For any $b\in \BC$ and $\sigma >0$, we have that $\int_{-\infty}^\infty \,\e^{-\frac{(\omega - b)^2}{2\sigma^2}}\rd\omega = \sqrt{2\pi} \sigma$.
\end{fact}
We can think of the width $\sigma_E$ as the uncertainty in energy, which scales inversely proportionally to the time width~$\sim\sigma_E^{-1}$.

\subsection{Exact detailed balance from that of the transition part}

Our notion of detailed balance for \Lword{}s is analogous to its classical cousin, ensuring a stationary state~$\vrho$. For the mathematical audience, we should mention that other forms of quantum detailed balance have also been studied (see~\autoref{sec:otherDB}), but we will dominantly focus on the following as it appears to be especially nice.
\begin{defn}[Kubo–Martin–Schwinger detailed balance condition]\label{defn:DB}
For a normalized, full-rank state $ \vrho\succ 0$, we say that an super-operator $\CL$ satisfies $\vrho$-detailed balance (or $\vrho$-DB in short) if
\begin{align}
     \CL^{\dagger}[ \cdot] =\vrho^{-1/2}\CL[ \vrho^{1/2}\cdot  \vrho^{1/2}]  \vrho^{-1/2},
\end{align}
or equivalently, whenever the associated {\rm discriminant} is self-adjoint with respect to $\vrho$, i.e., 
\begin{align}
\CD(\vrho, \CL)&:=  \vrho^{-1/4}\CL[ \vrho^{1/4}\cdot  \vrho^{1/4}]  \vrho^{-1/4}\\
 &=\vrho^{1/4}\CL^{\dagger}[ \vrho^{-1/4}\cdot  \vrho^{-1/4}]  \vrho^{1/4}=\CD(\vrho, \CL)^{\dagger}. 
\end{align}
In the above, $(\CL)^{\dagger}$ denotes the adjoint for superoperators with respect to trace (i.e., the Hilbert-Schmidt inner product).
\end{defn}
One may interpret the conjugation with the state as a similarity transformation under which the \Lword{} becomes Hermitian (w.r.t. to the KMS inner product). The above generalizes classical detailed balance by considering super-operators and permitting the stationary distribution to be an operator. 
\begin{prop}[Fixed point]\label{prop:DB_fixedpoint}
	If a \Lword{} $\CL$ is $\vrho$-detailed-balanced, then 
	\begin{align}
		\CL[\vrho] =0.
	\end{align}
\end{prop}

Recently, \textit{quantum approximate detailed balance} has also been studied in the precursor of this work~\cite[Section II.A]{chen2023QThermalStatePrep}, discussing nonasymptotic error bounds relating mixing times to fixed point error. Exact detailed balance gives a much simpler conceptual picture. Still, we may again need to recall approximate detailed balance for non-fine-tuned \Lword{}s (such as those from Nature) or amid intermediate steps of analysis (such as when truncating the radius of the quasi-local jumps). 

At first glance, the detailed balance condition is merely a linear equation that can be solved abstractly. However, the difficulty arises due to two additional constraints.
\begin{itemize}
\item (\textbf{Complete Positivity}.) \Lword{}s have a particular \textit{quadratic} dependence on the Lindblad operators to ensure complete positivity and trace preservation of $\e^{t\CL}$ for any $t$.
\item (\textbf{Efficiency}.) The \Lword{} (i.e., the block-encoding for the jumps and the Hamiltonian) must be efficiently implemented using a limited Hamiltonian simulation time.
\end{itemize}
The main challenge is to satisfy both constraints simultaneously. Indeed, Davies' generator~\cite{davies74,davies76} satisfies the first but not the second because it uses an infinite-time operator Fourier Transform; using QSVT, one might be able to directly implement the Boltzmann weight smoothly at moderate costs, but it may break the \Lword{} structure. Our approach begins by isolating the ``transition'' part of the \Lword{}~\eqref{eq:exact_DB_L} with abstract Lindblad operators $\vL_j$
\begin{align}
    \CL[\cdot] &:= -\ri[\vB,\cdot] +  \underset{\CT:=}{\underbrace{\sum_{j}\vL_j\cdot\vL_j^{\dagger}}} - \frac{1}{2}\{\underset{\vR:=}{\underbrace{\sum_j \vL_j^{\dagger}\vL_j}},\cdot\} 
\end{align}
where $\vB$ and $\vR$ are both Hermitian. This decomposition is helpful because conjugating with the stationary state preserves the form of the transition part
\begin{align}
  \sqrt{\vrho}^{-1}\L(\sum_{j} \vL_j( \sqrt{\vrho}\cdot\sqrt{\vrho} )\vL_j^{\dagger} \R)\sqrt{\vrho}^{-1} = \sum_j \vL'_j(\cdot) \vL_{j}^{'\dagger}.
\end{align}
However, the commutator $\vB$ and anti-commutator terms $\vR$ \textit{mix} with each other under conjugation with Gibbs state. 
Based on the above observation, our recipe for constructing a detailed-balanced Lindbladian consists of three steps:
\begin{enumerate}
    \item \textit{Guess} a set of Lindblad operators $\vL_j$ such that the transition part (which is $\CT =\sum_a \int_{-\infty}^{\infty}  \gamma(\omega) \hat{\vA}^a(\omega) \cdot \hat{\vA}^a(\omega)^{\dagger}\rd\omega$ in our case) obeys detailed balance. 
    \item According to the transition part $\CT$, determine the decay part parameterized by $\vR$. This gives a purely dissipative \Lword{}.
    \item According to the decay part $\vR$, tailor the commutator term $\vB$ to ensure detailed balance. Remarkably, such a $\vB$ always exists, can be found explicitly, and is essentially \textit{unique}. Of course, whether the map is efficiently implementable is a separate question.\footnote{This is inspired by a related work~\cite{guo2025designing} at an early stage.}
\end{enumerate}

To simplify the presentation, we introduce the following notation for conjugating any full-rank state $\vrho$:
\begin{eqnarray}
	\Gamma_{\vrho}[\cdot]:=\vrho^{1/2}(\cdot)\vrho^{1/2}\quad \text{and}\quad
	\Lambda_{\vrho}[\cdot]:=\vrho^{-1/2}(\cdot)\vrho^{1/2}. 
\end{eqnarray}
Observe that for Hermitian operator $\vX$, we have the identities $\Gamma_{\vrho}[\vX]^\dag=\Gamma_{\vrho}[\vX]$ and $\Lambda_{\vrho}[\vX]^\dag=\Lambda^{-1}_{\vrho}[\vX]$ 
When the context is clear, we will omit the subscript $\vrho$. 

The main calculation of this section is summarized as follows. For any $\vR$, we can give a general solution\footnote{In fact, it is possible to require a Gibbs fixed point without imposing detailed balance~\cite{guo2025designing}.} for the coherent term $\vB$; this calculation is possible because the coherent term is not constrained by the complete-positivity structure of \Lword{}s and only needs to be Hermitian.

\begin{lem}[Prescribing the coherent term]\label{lem:FindingCoherenceTerm}
For any full-rank state $\vrho$ and Hermitian operator $\vR$, there exists a unique Hermitian operator $\vB$ (up to adding any scalar multiples of the identity $\vI$) such that the super-operator
    \begin{align}
        \CS[\cdot]:=-\ri[\vB,\cdot] -\frac{1}{2}\{\vR,\cdot\}
    \end{align}
satisfies $\vrho$-DB.\footnote{Actually, our proof shows an even stronger statement: for any Hermitian $\vF$ commuting with $\vrho$ there is a unique $\vB$ (up to an additive term proportional to $\vI$) such that $\CS^{\dagger}[\cdot]-\Gamma^{-1}\circ\CS\circ\Gamma[\cdot]=-\ri[\vF,\cdot]$. The only change is that one should set $\vB_0:=-\frac{1}{2}\vF$. This relates to a more general notion of detailed balance that allows for a unitary drift \autoref{defn:unitary_DB}, see also~\cite[Section 5]{fangola2007GeneratorsDetailedBal}. This might be useful for breaking degeneracies of the state $\vrho$ (or the Hamiltonian $\vH$).} For a Gibbs state $\vrho\propto\exp(-\beta\vH)$, we can express the solution decomposed according to the Bohr frequencies $\nu \in B$ as
\begin{align}
\vB =\frac{\ri}{2} \sum_{\nu \in B} \tanh\left(\frac{\beta \nu}{4}\right) \vR_{\nu} .
\end{align}
\end{lem}
The above can be applied to a purely dissipative \Lword{}, where the transition part already satisfies $\vrho$-DB. 
\begin{cor}[$\vrho$-DB \Lword{}s]\label{cor:findingQ}
    Suppose we have a purely-dissipative $\CL_{diss}$ \Lword{} such that the transition part satisfies $\vrho$-DB for a full-rank state $\vrho$
\begin{align}
    \Gamma^{-1}_{\vrho}\circ \CT\circ \Gamma^{1}_{\vrho} = \CT^{\dagger},
\end{align}
then we can accordingly prescribe $\vB$ such that $-\ri[\vB,\cdot] +\CL_{diss}$ satisfies $\vrho$-DB.
\end{cor}
\begin{proof}[Proof of \autoref{lem:FindingCoherenceTerm}]

Let $\vK:=\vB-\frac{\ri}{2}\vR$ and observe that
\begin{align}
    \CS[\cdot] &:= - \ri [\vB,\cdot] - \frac{1}{2}\{\vR,\cdot\} = -\ri \vK(\cdot) + \ri (\cdot) \vK^{\dagger}\\
    \CS^{\dagger}[\cdot] &:= \ri [\vB,\cdot] - \frac{1}{2}\{\vR,\cdot\} = \ri \vK^{\dagger} (\cdot) - \ri (\cdot) \vK
\end{align}
using that $\vB$ and $\vR$ are Hermitian.
Then, 
\begin{align}
    \CS^{\dagger}[\cdot]-\Gamma^{-1}\circ\CS\circ\Gamma[\cdot] &= \ri \vK^{\dagger}(\cdot) - \ri(\cdot)\vK+\ri \Gamma^{-1}\L( \vK \Gamma[\cdot]-\Gamma[\cdot]\vK^{\dagger} \R)\\
    &=\ri \L(\vK^{\dagger}+\Lambda[\vK]\R)(\cdot) - \ri (\cdot) \L(\vK + (\Lambda[\vK])^\dagger\R)\tag*{(using that $\Gamma^{-1}\vK \Gamma[\cdot] = \Lambda[\vK](\cdot)$)}\\ 
	&=: \ri(\vQ(\cdot)-(\cdot)\vQ^\dag),
\end{align}
where we define $\vQ:= \vK^\dag+\Lambda[\vK]$. To ensure quantum detailed balance, we need the RHS to vanish
\begin{align}
    \vQ(\cdot)-(\cdot)\vQ^\dag = 0 \quad \iff \quad \vQ=0 + \lambda \vI \quad \text{for}\quad \lambda \in \BR,
\end{align}
that is, $\vQ$ vanishes up to a real multiple of the identity $\lambda \vI$; for simplicity, the identity part can be dropped for now and added back. Since $\vrho$ has full rank, we can assume without loss of generality that $\vrho\propto\exp(-\beta\vH)$ for some $\vH$. Now we compute 
\begin{align}
    \vQ &= \vK^{\dagger} + \Lambda[\vK]\\
    &= (1+ \Lambda )\vB + \frac{\ri}{2} (1-\Lambda) \vR \\
    &= \sum_{\nu \in B} (1+\e^{\beta\nu/2 }) \vB_{\nu} + \frac{\ri}{2}(1-\e^{\beta\nu/2}) \vR_{\nu}.\tag*{(using that $\Lambda(\vB_{\nu}) = \e^{\beta \nu/2}\vB_{\nu}$)}
\end{align}
where we denote $B:=B(\vH)$ the set of Bohr frequencies of $\vH$. Since the operators $\vB_{\nu}$ (and $\vR_{\nu}$) are linearly independent for different Bohr frequencies, 
\begin{align}
    \vQ = 0 \quad \iff \quad \vB_{\nu} = \frac{\ri}{2}\tanh\left(\frac{\beta \nu}{4}\right) \vR_{\nu}\quad \text{for each}\quad \nu \in B\tag*{(using that $\tanh(x) = \frac{\e^{2x}-1}{\e^{2x}+1}$)}.
\end{align}
In particular, $\vB_0 = 0$ since $\tanh(0)=0$.

Finally, since $\vR$ is Hermitian and $\tanh(\omega)$ is an odd function, we have that $\vB$ is Hermitian as well by 
\begin{align}
    \sum_{\nu \in B} (\vB_{\nu})^{\dagger}
    &= \frac{\ri}{2} \sum_{\nu \in B} \tanh\left(-\frac{\beta \nu}{4}\right)(\vR_{\nu})^{\dagger}\tag*{(Using $-\tanh(x)=\tanh(-x)$)}.\\
    &= \frac{\ri}{2} \sum_{\nu \in B} \tanh\left(-\frac{\beta \nu}{4}\right)\vR_{-\nu}\tag*{(Using that $\vR=\vR^{\dagger}$ implies $(\vR_{\nu})^{\dagger}=\vR_{-\nu}$)}\\
    &= \frac{\ri}{2} \sum_{\nu \in B} \tanh\left(\frac{\beta \nu}{4}\right)\vR_{\nu} = \sum_{\nu \in B} \vB_{\nu}. \tag*{(Change of variables $\nu\rightarrow -\nu$)\quad \qedhere}
\end{align}
Adding the identity part $\lambda\vI$ to conclude the proof.
\end{proof}

\subsection{Exact detailed balance for Gaussian weights and their linear combinations}
In this section, we carry out the abstract recipe for our advertised \Lword{} to prove the exact detailed balance condition (\autoref{defn:DB}). We begin with the transition part and then solve for the coherent term.
\subsubsection{Exact detailed balance of the transition part}
First, we show that the Gaussian ansatz indeed leads to detailed balance for the ``transition'' part
\begin{align}
    \Gamma^{-1}_{\vrho}\circ \CT\circ \Gamma_{\vrho} = \CT^{\dagger}.
\end{align}
It is instructive to rewrite the abstract equation above in terms of Bohr frequencies. Let $(\cdot)$ be any input matrix, then the transition part of our Lindbladian reads
 \begin{align}
 	\CT= \sum_{a\in A} \int_{-\infty}^{\infty}  \gamma(\omega) \hat{\vA^a}(\omega)(\cdot)\hat{\vA^a}(\omega)^{\dagger}\rd\omega
 	&=\frac{1}{\sigma_E\sqrt{2\pi}}\sum_{a\in A}\sum_{ \nu_1,\nu_2\in B} \int_{-\infty}^{\infty} \e^{- \frac{(\omega + \omega_\gamma)^2}{2\sigma_\gamma^2}}\e^{- \frac{(\omega-\nu_1)^2}{4\sigma_E^2}}\e^{- \frac{(\omega-\nu_2)^2}{4\sigma_E^2}} \vA^a_{\nu_1}(\cdot)(\vA^a_{\nu_2})^{\dagger}\rd\omega\\&
 	=: \sum_{a\in A} \sum_{ \nu_1,\nu_2\in B}\alpha^{(\omega_\gamma,\sigma_\gamma)}_{\nu_1,\nu_2} \vA^a_{\nu_1}(\cdot)(\vA^a_{\nu_2})^{\dagger}.\label{eq:exact_FT_form}
 \end{align}
Since the bilinear expression holds for any input $(\cdot)$, taking traces on both sides yields that 
\begin{align}
\sum_{a\in A}\int_{-\infty}^{\infty}  \gamma(\omega) \hat{\vA^a}(\omega)^{\dagger}\hat{\vA^a}(\omega)\rd\omega = \sum_{a\in A} \sum_{ \nu_1,\nu_2\in B}\alpha^{(\omega_\gamma,\sigma_\gamma)}_{\nu_1,\nu_2} (\vA^a_{\nu_2})^{\dagger}\vA^a_{\nu_1}.    
\end{align}
In terms of Bohr frequencies, the exact detailed balance condition is a certain symmetry of the coefficient matrix $\alpha^{(\omega_\gamma,\sigma_\gamma)}$.
\begin{prop}[Detailed balance in the Energy domain]\label{prop:DB_energy_domain}
Consider a super-operator parameterized by a Hamiltonian $\vH$, $\beta$, and a set of operators including its adjoints $\{\vA^a\colon a\in A\}=\{\vA^{a\dagg}\colon a\in A\}$:
\begin{align}
    \CT = \sum_{a\in A} \sum_{\nu_1,\nu_2\in B}\alpha_{\nu_1,\nu_2} \vA^a_{\nu_1}(\cdot)(\vA^a_{\nu_2})^{\dagger}\quad \text{such that}\quad \alpha_{\nu_1,\nu_2}= \alpha_{-\nu_2,-\nu_1} \e^{-\beta (\nu_1+\nu_2)/2}\quad \text{for each}\quad \nu_1,\nu_2 \in B
\end{align}
Then, 
    \begin{align}
\Gamma^{-1}_{\vrho}\circ \CT\circ \Gamma^{1}_{\vrho} = \CT^{\dagger}.
    \end{align}
\end{prop}
Indeed, one recovers the classical detailed balance condition for inputs diagonal in the energy basis. However, the quantum detailed balance condition also constrains the amplitudes between off-diagonal matrix elements (Figure~\ref{fig:QuantumDB}).  
\begin{proof}
We can directly calculate
      \begin{align}
        \Gamma^{-1}_{\vrho}\circ \CT\circ \Gamma^{1}_{\vrho} 	&=\sum_{a\in A}\sum_{\nu_1,\nu_2\in B} \alpha_{\nu_1,\nu_2}\sqrt{\vrho_{\beta}^{-1}}\vA^a_{\nu_1}\sqrt{\vrho_{\beta}}(\cdot)\sqrt{\vrho_{\beta}}(\vA^{a}_{\nu_2})^\dagg\sqrt{\vrho_{\beta}^{-1}}\\ &
        = \sum_{a\in A}\sum_{\nu_1,\nu_2\in B} \alpha_{\nu_1,\nu_2}\e^{\frac{\beta}{2}\nu_1}\vA^a_{\nu_1}(\cdot)(\vA^{a}_{\nu_2})^\dagg\e^{\frac{\beta}{2}\nu_2} \tag*{(since $\vrho_{\beta}\propto\e^{-\beta \vH}$)}\\&
		=\sum_{a\in A}\sum_{\nu_1,\nu_2\in B} \alpha_{-\nu_2,-\nu_1}\vA^a_{\nu_1}(\cdot)(\vA^{a}_{\nu_2})^\dagg \tag*{(since $\alpha_{\nu_1,\nu_2}\e^{\frac{\beta(\nu_1+\nu_2)}{2}}=\alpha_{-\nu_2,-\nu_1}$)}\\&
		=\sum_{a\in A}\sum_{\nu_1,\nu_2\in B} \alpha_{-\nu_2,-\nu_1}((\vA^{a\dagg})_{-\nu_1})^\dagg(\cdot)(\vA^{a\dagg})_{-\nu_2}\tag*{(since $(\vA_{\nu})^\dagg = (\vA^\dagg)_{-\nu}$)}\\&
		=\sum_{a\in A}\sum_{\nu_1,\nu_2\in B} \alpha_{-\nu_2,-\nu_1}(\vA^a_{-\nu_1})^\dagg(\cdot)\vA^{a}_{-\nu_2}\tag*{(since $\{\vA^a\colon a\in A\}=\{\vA^{a\dagg}\colon a\in A\}$)}\\&
		=\sum_{a\in A}\sum_{\nu_1,\nu_2\in B} \alpha_{\nu_2,\nu_1}(\vA^{a}_{\nu_1})^\dagg(\cdot)\vA^a_{\nu_2}  \tag*{(since $B = \text{spec}(\vH) -\text{spec}(\vH) = -B$)}\\&
        =\sum_{a\in A}\sum_{\nu_1,\nu_2\in B} \alpha_{\nu_1,\nu_2}(\vA^{a}_{\nu_2})^\dagg(\cdot)\vA^a_{\nu_1}
	= \CT^{\dagger}.\tag*{(relabelling $\nu_1\leftrightarrow \nu_2$)\qedhere}
    \end{align}
\end{proof}
While the above representation explicitly addresses the energy basis, we note that positivity becomes obscured as the left and right energy labels $\nu_1, \nu_2$ can differ. The positivity now becomes implicit in the coefficient matrix. 
\begin{prop}[Positive semi-definite]\label{prop:CP_implies_H}
    If the coefficients $\alpha_{\nu_1,\nu_2}$ as a matrix $\vec{\alpha}$ is positive-semi-definite 
    \begin{align}
        \vec{\alpha} \ge 0,
    \end{align}
    (and thus Hermitian $\alpha_{\nu_1,\nu_2}=\left(\alpha_{\nu_2,\nu_1}\right)^*$), then 
\begin{align}
\sum_{a\in A} \sum_{\nu_1,\nu_2\in B} \alpha_{\nu_1,\nu_2}\left(\vA^a_{\nu_1}(\cdot)(\vA^{a}_{\nu_2})^\dagg - \frac{1}{2}\{(\vA^{a}_{\nu_2})^\dagg\vA^a_{\nu_1},\cdot\} \right) \quad \text{gives a \Lword{}}.
\end{align}
\end{prop}
Positivity indeed holds for~\eqref{eq:exact_FT_form} and can be seen by the integral form of the coefficients.
To conclude this section, it remains to verify that the coefficients arising from the Gaussian indeed satisfy the symmetry.
\begin{lem}[Exact ``skew-symmetry'' of coefficients]\label{lem:exact_sym_coefficients}
For each $\omega_\gamma,\sigma_\gamma$, the coefficients $\alpha^{(\omega_\gamma,\sigma_\gamma)}_{\nu_1,\nu_2}$ defined by~\eqref{eq:exact_FT_form} factorize 
\begin{align}\label{eq:productGaussianAlphas}
    \alpha^{(\omega_\gamma,\sigma_\gamma)}_{\nu_1,\nu_2}=\frac{\sigma_\gamma}{\sqrt{\sigma_E^2+\sigma_\gamma^2}}\cdot \exp\L( - \frac{(\nu_1+\nu_2 +2\omega_\gamma)^2}{8(\sigma_E^2+\sigma_{\gamma}^2)}\R)  \cdot \exp\L(-\frac{(\nu_1-\nu_2)^2}{8\sigma_E^2}\R),
\end{align}
and have a certain ``skew-symmetry'' under negation and transpose
\begin{align}\label{eq:alpha_DB}
\alpha^{(\omega_\gamma,\sigma_\gamma)}_{\nu_1,\nu_2} = \alpha^{(\omega_\gamma,\sigma_\gamma)}_{-\nu_2,-\nu_1} \e^{-\beta (\nu_1+\nu_2)/2} \quad \text{for} \quad \beta := \frac{2\omega_\gamma}{\sigma_E^2+\sigma_\gamma^2}. 
\end{align}
\end{lem}
\begin{proof}
We directly calculate the Gaussian integrals in \eqref{eq:exact_FT_form}, preserving the quadratic nature of the exponents. 
\begin{align}
&\alpha^{(\omega_\gamma,\sigma_\gamma)}_{\nu_1,\nu_2} = \frac{1}{\sigma_E\sqrt{2\pi}} \int_{-\infty}^{\infty}  \exp\L( -\frac{\omega^2}{2}\cdot\left(\frac{1}{\sigma_\gamma^2}+\frac{1}{\sigma_E^2}\right) -\omega \cdot\L(\frac{\omega_\gamma}{\sigma_\gamma^2}-\frac{\nu_1+\nu_2}{2\sigma_E^2}\R) - \frac{\nu_1^2+\nu_2^2}{4\sigma_E^2} -\frac{\omega_\gamma^2}{2\sigma_\gamma^2}\R)\rd\omega = \\&
	= \frac{1}{\sigma_E\sqrt{2\pi}} \int_{-\infty}^{\infty}  \kern-3mm\exp\L( -\frac{\!\!\L(\omega + \left(\frac{1}{\sigma_\gamma^2}+\frac{1}{\sigma_E^2}\right)^{\!\!-1}\!\L(\frac{\omega_\gamma}{\sigma_\gamma^2}-\frac{\nu_1+\nu_2}{2\sigma_E^2}\R) \R)^{\!\!2}\!\!\!}{2\left(\frac{1}{\sigma_\gamma^2}+\frac{1}{\sigma_E^2}\right)^{\!\!-1}} + \frac{\!\left(\frac{1}{\sigma_\gamma^2}+\frac{1}{\sigma_E^2}\right)^{\!\!-1}\!\!\!\!}{2}\L(\frac{\omega_\gamma}{\sigma_\gamma^2}-\frac{\nu_1+\nu_2}{2\sigma_E^2}\R)^{\!\!2} \!- \frac{\nu_1^2+\nu_2^2}{4\sigma_E^2} -\frac{\omega_\gamma^2}{2\sigma_\gamma^2}\R)\rd\omega\\[-5mm]&        
	= \frac{1}{\sigma_E}\cdot \frac{1}{\sqrt{\frac{1}{\sigma_\gamma^2}+\frac{1}{\sigma_E^2}}} \exp\L( \frac{1}{2\left(\frac{1}{\sigma_\gamma^2}+\frac{1}{\sigma_E^2}\right)}\L(\frac{\omega_\gamma}{\sigma_\gamma^2}-\frac{\nu_1+\nu_2}{2\sigma_E^2}\R)^2 - \frac{\nu_1^2+\nu_2^2}{4\sigma_E^2} -\frac{\omega_\gamma^2}{2\sigma_\gamma^2}\R) \tag*{(by \autoref{fact:gaussian_int})}\\
&=\frac{1}{\sqrt{\frac{\sigma_E^2}{\sigma_\gamma^2}+1}}\cdot \exp\L(\!- \frac{(\nu_1+\nu_2 +2\omega_\gamma)^2}{8(\sigma_E^2+\sigma_{\gamma}^2)}\R)  \cdot \exp\L(\!-\frac{(\nu_1-\nu_2)^2}{8\sigma_E^2}\R)\tag*{(since $2(\nu^2_1\!+\!\nu_2^2) = (\nu_1\!+\!\nu_2)^2+(\nu_1\!-\!\nu_2)^2$)}\\
 &= \alpha^{(\omega_\gamma,\sigma_\gamma)}_{-\nu_2,-\nu_1} \e^{-\beta (\nu_1+\nu_2)/2}. \tag*{(due to our choice of $\beta$)\qedhere}
\end{align}
\end{proof}

We see that we may tune the parameters to match a desirable exponent $\beta$. Remarkably, the width $\sigma_E$ can be finite (i.e., do not scale linearly the precision $\frac{1}{\epsilon}$ like in metrology) while retaining exact detailed balance; a reasonable choice is, e.g.,
\begin{align}
	\omega_\gamma=\sigma_E =\sigma_\gamma = \frac{1}{\beta}.
\end{align}
This will imply that the algorithmic cost for implementing the Gaussian weighted operator Fourier Transform will only need to scale with $\beta$ (and polylogarithmically with the precision due to discretization and truncation error)! However, the Gaussian transition weight comes with the price of a narrower band of transitions peaked at $\omega_\gamma \pm \CO(\sigma_\gamma)$, which might result in a substantially increased mixing time compared to, e.g., Metropolis weight $\gamma(\omega) = \min (1,\e^{-\beta \omega})$; we will revisit this issue by taking \textit{linear combinations} of Gaussians at~\autoref{sec:convex_comb_gaussian} since all our calculation are linear. For clarity, we first focus on the Gaussian weights.

Why does Gaussian interplay so perfectly with quantum detailed balance?In~\autoref{sec:why_gaussian}, we attempted to derive Gaussian from the first principle. Indeed, Gaussians are very natural if we impose several conditions on the function.

\subsubsection{Adding the unitary term}\label{sec:prove_exactDB}

Now that the ``transition'' part satisfies the detailed balance condition exactly, we proceed to complete the Lindbladian by adding the ``decay'' part and the ``coherent'' part. The decay part is fixed by trace-preserving;~\autoref{lem:FindingCoherenceTerm} then uniquely prescribes the required coherent term, which we display as follows in the frequency domain. The explicit form will be useful for implementation.

\begin{cor}[An exactly detailed-balanced \Lword{} with Gaussian filtering]\label{cor:ExctDissip}
	The \Lword{} 
	\begin{align}
		\CL[\cdot] := -\ri \sum_{\nu \in B} [\vB_{\nu}, \cdot] + \sum_{a\in A} \sum_{\nu_1,\nu_2\in B} \alpha_{\nu_1,\nu_2}\left(\vA^a_{\nu_1}(\cdot)(\vA^{a}_{\nu_2})^\dagg - \frac{1}{2}\{(\vA^{a}_{\nu_2})^\dagg\vA^a_{\nu_1},\cdot\} \right)
	\end{align}
	corresponding to a self-adjoint set of jump operators $\{\vA^a\colon a\in A\}=\{\vA^{a\dagg}\colon a\in A\}$, parametrized by coefficients $\alpha_{\nu_1,\nu_2}\in \BC$ satisfying $\alpha_{\nu_1,\nu_2}^* = \alpha_{\nu_2,\nu_1}$ and
\begin{align}
\alpha_{\nu_1,\nu_2}\e^{\frac{\beta(\nu_1+\nu_2)}{2}}=\alpha_{-\nu_2,-\nu_1},
\end{align}
	and amended by the coherent terms
	\begin{align}
		\vB_{\nu}&:=\sum_{a\in A} \sum_{\stackrel{\nu_1-\nu_2 = \nu}{\nu_1,\nu_2\in B}}  \underset{\hat{f}(\nu_1,\nu_2) :=}{\underbrace{\frac{\tanh(-\beta (\nu_1-\nu_2)/4)}{2\ri}\alpha_{\nu_1,\nu_2}}} (\vA^{a}_{\nu_2})^\dagg\vA^{a}_{\nu_1}\label{eq:BE}
	\end{align}
	satisfies $\vrho_{\beta}$-detailed balance.
\end{cor}

\begin{proof}
Apply~\autoref{lem:FindingCoherenceTerm}. Note that the operator 
\begin{align}
(\vA^{a}_{\nu_2})^\dagg\vA^{a}_{\nu_1} = \sum_{E_i-E_j = \nu_1, E_i-E_k = \nu_2} \vP_{E_k} \vA^{a\dagger} \vP_{E_i} \vA^a \vP_{E_j}    
\end{align}
must have the energy difference contained in the set of Bohr frequencies $B$
\begin{align}
    \nu_1-\nu_2 &= E_i - E_j - (E_i -E_k) \quad \text{for some}\quad E_i, E_j, E_k \in Spec(\vH)\\
    &= E_k -E_i \in B.
\end{align}
\end{proof}

The above corollary essentially leads to~\autoref{thm:Exact_fixedpoint} but is written in the Bohr-frequency decomposition.
\begin{proof}[Proof of~\autoref{thm:Exact_fixedpoint}]
Combine~\autoref{prop:CP_implies_H}, \autoref{lem:exact_sym_coefficients}, \autoref{prop:DB_fixedpoint}, and \autoref{cor:ExctDissip} to conclude the proof.
\end{proof}

\subsubsection{Linear combination of Gaussians}
\label{sec:convex_comb_gaussian}

Can we go beyond Gaussians? As we discussed, the Gaussians have a narrower band of transitions; it would be desirable to lift this restriction to accelerate the mixing time. In this section, we give a family of filters by exploiting the freedom to tune the Gaussian parameters ($\omega_{\gamma},\sigma_E,\sigma_{\gamma}$) and taking a linear combination of Gaussians.

\begin{cor}[Linear combination of Gaussians]\label{cor:exact_sym_coefficients}
	Fix $\sigma_E$ and $g\in\ell_1(\BR)$ and set $\gamma^{(g)}(\omega):=\int_{\frac{\beta\sigma_E^2}{2}}^{\infty}g(x)\e^{-\frac{(\omega + x)^2}{4x/\beta-2\sigma_E^2}}\rd x$.\footnote{Note that in principle we could also vary $\sigma_E$, but that would complicate both the analysis and the implementation due to the required adjustments to the parameters of the performed operator Fourier Transform.}
	Then, analogous to \eqref{eq:exact_FT_form}, the coefficients as a linear combination over integration variable $x$ 
 \begin{align}
     \alpha^{(g)}_{\nu_1,\nu_2} := \int_{\frac{\beta\sigma_E^2}{2}}^{\infty}g(x) \alpha^{(\omega_{\gamma}, \sigma_{\gamma})}_{\nu_1,\nu_2} \rd x \quad \text{for}\quad (\omega_{\gamma}(x),\sigma_{\gamma}(x)) = (x,\sqrt{2x/\beta - \sigma_E^2}),\label{eq:alphaGenDef}
 \end{align}
satisfy the symmetries $\alpha^{(g)*}_{\nu_1,\nu_2} = \alpha_{\nu_2,\nu_1}^{(g)}$ and
		$\alpha^{(g)}_{\nu_1,\nu_2} = \alpha^{(g)}_{-\nu_1,-\nu_2} \e^{-\beta (\nu_1+\nu_2)/2}$ for each $\nu_1,\nu_2 \in \BR$. If $g(x) \ge 0$ for each $x$, then we also have that $\vec{\alpha^{(g)}}\ge 0$, however, this is not a necessary condition. 
\end{cor}
\begin{proof}
Recall the meaning of the superscripts $\alpha^{(\omega_\gamma,\sigma_\gamma)}_{\nu_1,\nu_2}$. The proof is merely the linearity of symmetries and the fact that a convex combination preserves the cone of positive semidefinite matrices.    
\end{proof}

To widen the band of transitions, a natural choice is to weigh each Gaussian $\e^{- \frac{(\omega + \omega_\gamma)^2}{2\sigma_\gamma^2}}$ with its inverse $\ell_1$-weight. Surprisingly, this leads to filters that resemble the Metropolis and Glauber weights; while other choices are plausible, we spell out the calculation for this as a natural representative.

\begin{prop}[Metropolis and Glauber-like filters]
     Setting $g(\omega_\gamma)=\!\frac{1}{\sqrt{2\pi}\sigma_\gamma}\!=\!\frac{1}{\!\sqrt{2\pi(\frac{2\omega_\gamma}{\beta}-\sigma_E^2)}}$ 
     yields
\begin{align}\label{eq:convolutionFilterFinite}
	\gamma^{(s)}_{\sigma_E}(\omega) &:=\int_{\frac{\beta\sigma_E^2}{2}}^{\frac{\beta\sigma_E^2}{2}+\frac{s^2}{\beta}}\frac{\e^{-\frac{(\omega + x)^2}{4x/\beta-2\sigma_E^2}}}{\sqrt{2\pi(2x/\beta-\sigma_E^2)}}\rd x\\&
	=\e^{-\beta\max\left(\omega +\frac{\beta \sigma_{E}^2}{2},0\right)}\!\!\cdot\underset{\leq 1}{\underbrace{\frac{1}{2}  \left[\left(\!\erf\left(\frac{s}{2}-\frac{\beta}{s}\left|\omega +\!\frac{\beta \sigma_{E}^2}{2}\right|\right)\!+1\right)\!+\e^{\beta\left|\omega +\frac{\beta \sigma_{E}^2}{2}\right|}\! \left(\!\erf\left(\frac{s}{2}+\frac{\beta}{s}\left|\omega +\!\frac{\beta \sigma_{E}^2}{2}\right|\right)\!-1\right)\right]}},\kern2mm
\end{align}
which, in the $s\rightarrow\infty$ limit, coincides with the Metropolis weight shifted by $\frac{\beta \sigma_{E}^2}{2}$:
\begin{align}\label{eq:convolutionFilter}
	\gamma^{(\infty)}_{\sigma_E}(\omega) =\e^{-\beta\max\left(\omega +\frac{\beta \sigma_{E}^2}{2},0\right)}.
\end{align}
Restricting the above $g(\omega_\gamma)$ to the interval $\omega_\gamma\in\left(\frac{3\beta\sigma_E^2}{2},\infty\right)$ results in the following smooth variant of~\eqref{eq:convolutionFilter}
\begin{align}\label{eq:convolutionFilterAPX}
\tilde{\gamma}^{(\infty)}_{\sigma_E}(\omega)&
=\e^{-\beta\max\left(\omega +\frac{\beta \sigma_{E}^2}{2},0\right)}\!\!\cdot\underset{\leq 1}{\underbrace{\frac{1}{2}  \left[\erfc\left(\frac{1}{\sigma_{E}}\left(\frac{\beta \sigma_{E}^2}{2}-\left|\omega +\frac{\beta \sigma_{E}^2}{2}\right|\right)\right)+\e^{\beta\left|\omega\!+\frac{\beta \sigma_{E}^2}{2}\right|} \erfc\left(\frac{1}{\sigma_{E}}\left(\frac{\beta \sigma_{E}^2}{2}+\left|\omega +\frac{\beta \sigma_{E}^2}{2}\right|\right)\right)\right]}},\kern2mm
\end{align}
which resembles the Glauber filter also shifted by $\frac{\beta \sigma_{E}^2}{2}$ (\autoref{fig:weightFunctions}).
\end{prop}
\begin{proof}
    We directly found the above by Mathematica; see~\autoref{apx:Mathematica} for the code. 
\end{proof}

\begin{figure}[ht]
	\begin{center}	
		\begin{tikzpicture}[scale=1.]
			\begin{axis}[xlabel={$\underset{\phantom{f}}{\beta\omega}$}, ylabel={$\phantom{e}^{\displaystyle\gamma(\omega)}$}, axis lines = middle, domain=-3.14:
				3.14, ymin=0, ymax=1,
				thick, trig format plots=rad,no markers,
				xtick = {-3,-2,-1,0,1,2,3},
				width=0.8\textwidth,
				xticklabels = {$-3\phantom{-}$,$-2\phantom{-}$,$-1\phantom{-}$,$0$,$1$,$2$,$3$},
				ytick = {0.5,1},
				yticklabels = {$\frac{1}{2}\kern-2mm$,$1$},
				height=0.3\textwidth,
				]	
				\addplot[teal,smooth,domain=0:3.1] {exp(-max(x,0))};				
				\addlegendentry{\(\exp(-\beta\max(\omega,0))\)}				
				\addplot[blue,smooth,domain=-0.5:3.14] {exp(-max(x+1/2,0))};
				\addlegendentry{\(\gamma^{(\infty)}_{\beta^{-1}}(\omega)\text{ in }\eqref{eq:convolutionFilter}\)}	
				\addplot[cyan,smooth,domain=-3.14:3.1] {1.6454275348389272*10^(-12)*x^(17)-3.1594574766943596*10^(-12)*x^(16)-5.879669816924585*10^(-11)*x^(15)+2.0696342289194524*10^(-10)*x^(14)+1.7256578885166587*10^(-9)*x^(13)-9.471557317729988*10^(-9)*x^(12)-3.868107642387552*10^(-8)*x^(11)+3.3465577750223796*10^(-7)*x^(10)+5.286965671404834*10^(-7)*x^9-9.25394938555528*10^(-6)*x^8+2.018957670733737*10^(-6)*x^7+0.000195809*x^6-0.000342935*x^5-0.00299429*x^4+0.0103029*x^3+0.0293662*x^2-0.219465*x+0.363887};
				\addlegendentry{\(\tilde{\gamma}^{(\infty)}_{\beta^{-1}}(\omega)\text{ in }\eqref{eq:convolutionFilterAPX}\)}	
				\addplot[orange,smooth,domain=-3.14:3.1] {1/(1+exp(x+0/2))};	
				\addlegendentry{\(1/(1+\exp(\beta\omega))\)}						
				\addplot[teal,smooth,domain=-3.14:0] {1};
				\addplot[blue,smooth,domain=-3.14:-0.5] {1};
			\end{axis}
		\end{tikzpicture}
	\end{center}
	\caption{
		A plot of the filter functions $\gamma(\omega)$ for Metropolis, Glauber and our filters arising from Gaussian linear combination~\eqref{eq:convolutionFilter}-\eqref{eq:convolutionFilterAPX}  (with $\sigma_{E}=\frac{1}{\beta}$).
	}\label{fig:weightFunctions}
\end{figure}

\section{Algorithms}\label{sec:algorithm}

This section presents efficient quantum algorithms for simulating the advertised \Lword{} and the associated parent Hamiltonian. The former mainly builds on black-box Lindbladian simulation algorithms~\cite{cleve2016EffLindbladianSim,childs2016SparseLindbladianSim,li2022SimMarkOpen}, more precisely their improved variant described in~\cite[Theorem III.2]{chen2023QThermalStatePrep} whose complexity boils down to constructing block encoding for the Lindblad operators (\autoref{def:blockLindladian}); the latter merely requires block-encoding the parent Hamiltonian, which feeds into quantum simulated annealing (see~\cite[Appendix G]{chen2023QThermalStatePrep} for a modern discussion) to prepare the purified Gibbs state.

Thus, the main algorithmic contribution is to assemble the block encodings associated with our synthetic Lindbladian~\eqref{eq:exact_DB_L} and our parent Hamiltonians (\autoref{thm:D_cost}). The frequency domain representation (\autoref{sec:analysis}), which is natural in the context of analyzing quantum detailed balance, is less instructive for algorithmic implementation. Indeed, addressing the exact energy eigenstates (or the exact Bohr frequencies) is generally inefficient. 
Nevertheless, the algorithmic task becomes straightforward in the time domain representation. Indeed, our Lindbladian can be expressed in terms of weighted time integrals $\int (\cdot) \rd t$ of some rapidly decaying functions; a standard Linear-Combination-of-Unitary argument (under suitable discretization) leads to the algorithmic complexity in terms of controlled Hamiltonian simulation time.

For the Lindbladian~\eqref{eq:exact_DB_L}, block-encodings for the dissipative part are already constructed in~\cite[Section III.B]{chen2023QThermalStatePrep}. Thus, it remains to construct the coherent term; for the parent Hamiltonian, we will need to construct block encodings from scratch, but the manipulations are analogous. In the following sections, we first present the time-domain expressions (\autoref{sec:timedomainL}-\ref{sec:timeDomainH}), which immediately yield the corresponding block encodings (\autoref{sec:block_encodings}) and the overall complexities (\autoref{sec:prove_L_cost}-\ref{sec:prove_D_cost}).

\subsection{Time-domain representation of our Lindbladians}
\label{sec:timedomainL}
Applying a two-dimensional Fourier Transform for the coherent term~\eqref{eq:BE} leads to the following time-domain representation, where LCU techniques are naturally applicable. See~\autoref{sec:L-time} for the calculations.
\begin{restatable}[Coherent term for the Gaussian case]{cor}{coherentTermGaussian}\label{cor:GaussianLikeWeight}
For each $\beta >0$ and parameters $\sigma_E = \sigma_{\gamma} = \omega_{\gamma} = \frac{1}{\beta}$, the coherent term $\vB$~\eqref{eq:BE} corresponding to the Gaussian weight $\gamma(\omega) = \exp\L(- \frac{(\beta\omega + 1)^2}{2}\R)$ can be written as
\begin{align}\label{eq:GaussianDimLessTime}
		\vB&:= \sum_{a\in A} 
		\int_{-\infty}^{\infty}b_1(t)\e^{-\ri\beta\vH t} \L(\int_{-\infty}^{\infty}b_2(t')\vA^{a\dagger}(\beta t')\vA^a(-\beta t')\rd t'\R) \e^{\ri \beta\vH t}\rd t,
	\end{align}
	where 
 \begin{align}
     b_1(t) &:=2\sqrt{\pi}\e^{\frac{1}{8}} \L(\!\frac{1}{\cosh\L(2\pi t\R)\!}*_t\sin\left(\!- t\right)  \exp\L(\!-2 t^2\R)\!\R)\quad &\text{such that}\quad \norm{b_1}_1 < 1\label{eq:b1}\\ 
 b_2(t) &:=\frac{1}{\pi}\sqrt{\frac{1}{\pi}}\exp\L(-4 t^2-2\ri t \R) \quad &\text{such that}\quad \norm{b_2}_1 < \frac{1}{8}\label{eq:b2}.
 \end{align}
\end{restatable}
Indeed, we can verify that $\vB$ is Hermitian by
\begin{align}
    \vB^{\dagger}&= \sum_{a\in A}
		\int_{-\infty}^{\infty}b^*_1(t)\e^{-\ri\beta\vH t} \L(\int_{-\infty}^{\infty}b^*_2(t') \vA^{a\dagger}(-\beta t')\vA^a(\beta t')\rd t'\R) \e^{\ri \beta\vH t}\rd t\\
  &=\sum_{a\in A}
		\int_{-\infty}^{\infty}b_1(t)\e^{-\ri\beta\vH t} \L(\int_{-\infty}^{\infty}b_2(-t') \vA^{a\dagger}(-\beta t')\vA^a(\beta t')\rd t'\R) \e^{\ri \beta\vH t}\rd t \tag*{(By $b_1^*(t)=b_1(t)$ and $b_2^*(t)=b_2(-t)$)}\\
  &=\sum_{a\in A}
		\int_{-\infty}^{\infty}b_1(t)\e^{-\ri\beta\vH t} \L(\int_{-\infty}^{\infty}b_2(t') \vA^{a\dagger}(\beta t')\vA^a(-\beta t')\rd t'\R) \e^{\ri \beta\vH t}\rd t = \vB \tag*{(Change of variable $t'\rightarrow -t'$)}.
\end{align}

The explicit weights corresponding to Metropolis weights are slightly more cumbersome due to taming a mild (logarithmic) singularity. We can also verify that $\vB^{M,\eta}$ is Hermitian by $b_2^{M,\eta*}(t) =b_2^{M,\eta}(-t)$. 

\begin{restatable}[Approximate coherent term for the Metropolis-like weight]{cor}{MetropolisLikeFilter}\label{cor:coherentMetropolis}
	If $\sigma_{E}=\frac{1}{\beta}$, then the coherent term $\vB^M$ corresponding to the Metropolis-like weight $\gamma^M(\omega) = \exp\L(-\beta\max\left(\omega+\frac{1}{2\beta},0\right)\R)$ satisfies 
	\begin{align}\label{eq:Bdiff}
		\nrm{\vB^M-\vB^{M,\eta}}\leq\nrm{\sum_{a\in A}\vA^{a\dagg}\vA^a}\min\left(\frac{\sqrt{2}\eta\beta\nrm{\vH}}{\pi},\bigO{(\eta\beta\nrm{\vH})^3}\right),
	\end{align}
	where 
\begin{align}\label{eq:MetropolisDimLessTime}
		\vB^{M,\eta}:=\sum_{a\in A}\int_{-\infty}^{\infty}b_1(t)\e^{-\ri\beta \vH t}\left( \int_{-\infty}^{\infty}b_2^{M,\eta}(t')\vA^{a\dagger}(\beta t')\vA^a(-\beta t')\rd t'+\frac{1}{8\sqrt{2}\pi}\vA^{a\dagger}\vA^a\right) \e^{\ri \beta\vH t}\rd t,
	\end{align}     
	with $b_1(t)$ as in~\eqref{eq:b1}, and 
\begin{align}
  b_2^{M,\eta}(t):=\frac{1}{2\sqrt{2}\pi}\frac{\exp\L(-2t^2 -\ri t\R)+\indicator(|t|\leq\eta)\ri\left(2t+\ri\right)}{t (2 t+\ri)}  
	\quad \text{such that}\quad \lVert b_2^{M,\eta}\rVert_1<\frac{2}{5}+\!\frac{1}{\sqrt{2}\pi}\ln(1/\eta).\label{eq:B2Meta}
\end{align}
Further, if $\left[\sum_{a\in A}\vA^{a\dagger}\vA^a\!,\vH\right]=0$, we can drop the second term in \eqref{eq:MetropolisDimLessTime} after the integral in $t'$ since $\int_{-\infty}^\infty b_1(t)=0$.    
\end{restatable}
See~\autoref{sec:cal_Metropolis} for the proof. Via LCU, using the weighing function $\frac{b_1(t)}{\norm{b_1}_1}\cdot\frac{b_2^{M,\eta}(t')}{\norm{b_2^{M,\eta}}_1}$ as per the above integral representation, we can directly get an approximate block-encoding of $\vB^{M,\eta}/\alpha$ where $\alpha=\norm{b_1}_1\cdot\norm{b_2^{M,\eta}}_1=\bigO{\log(1/\eta)}$. In order to get an $\frac{\epsilon}{2}$-accurate approximation of $\vB^{M,\eta}/\alpha$ in the block-encoding it suffices to truncate the respective integrals at $t'\sim \sqrt{\log(1/\epsilon)}$ and $t\sim \bigO{\log(1/\epsilon)}$, which in turn gives an $\frac{\epsilon}{2}+ \frac{\eta}{2}\beta\nrm{\vH}\nrm{\sum_{a\in A}\vA^{a\dagg}\vA^a}$ approximate block-encoding of $\frac{\vB^M}{\alpha}$. Due to such time truncation the resulting LCU implementation requires merely $\bigO{\beta\log(1/\epsilon)}$ (controlled) Hamiltonian simulation, and by setting $\eta=\frac{\epsilon}{\beta\nrm{\vH}\nrm{\sum_{a\in A}\vA^{a\dagg}\vA^a}}$ results in an $\epsilon$-approximate block-encoding of $\vB^M/\alpha$ where $\alpha=\bigO{\log(\frac{\beta\nrm{\vH}\nrm{\sum_{a\in A}\vA^{a\dagg}\vA^a}}{\epsilon})}$. Apart from this mild logarithmic subnormalization $\alpha$, we expect the utilized (controlled) Hamiltonian simulation time to be optimal.

\subsection{Time-domain representation of our parent Hamiltonians}
\label{sec:timeDomainH}

Recall that the central mathematical object for coherent Gibbs sampling is the \textit{discriminant} (i.e., the Lindbladian under a similarity transformation)
\begin{align}
	\CH(\vrho,\CL) :=  \vrho^{-1/4}\CL[ \vrho^{1/4}\cdot  \vrho^{1/4}]  \vrho^{-1/4}
\end{align}
This amounts to a mild calculation that transfers the heavy lifting done already in the Lindbladian context. One adaptation for the coherent algorithm is that the fundamental object is not a superoperator but an operator on a doubled Hilbert space. Formally, we define the \emph{vectorization} of a super-operator by
\begin{align*}
\mathcal{M}[\cdot]&=\sum_j  \alpha_j \vA_j[\cdot]\vB_j \rightarrow \vec{\mathcal{M}}=\sum_j \alpha_j \vA_j\otimes\vB^T_j \quad \text{(vectorization)}
\end{align*}
where $\vB^T_j$ denotes the transpose of the matrix $\vB_j$ in the computational basis $\ket{i}$. We use curly fonts $\CM$ for super-operators and bold fonts $\vec{\CM}$ for the vectorized super-operators (which is, a matrix). For a matrix $\vrho$, let us denote its vectorized version by 
\begin{align}
\ket{\vrho}:= (I \otimes T^{-1})\vrho \quad \text{where}\quad T \ket{i}=\bra{i}\quad \text{for each}\quad \ket{i}.
\end{align}
In the time domain, our parent Hamiltonian takes the following form (see~\autoref{sec:parentHamTime} for the calculations)
\begin{align}
        \vec{\CH}_{\beta} &= \undersetbrace{\text{transition part}}{\sum_{a \in A} \int_{-\infty}^{\infty}\int_{-\infty}^{\infty} h_-(t_-)h_+(t_+) \cdot \vA^a\L(t_+-t_-\R)\otimes\vA^{a}\L(-t_--t_+\R)^T \rd t_+\rd t_- }
        + \undersetbrace{\text{ decay and coherent part}}{\frac{1}{2}(\vN\otimes \vI + \vI\otimes\vN^{*})}.
\end{align}
Compared with the Lindbladian case, the first term is essentially the transition part under a similarity transformation, and the second term combines the decay and coherent part. Here, we do not care about complete positivity or trace-preserving but merely Hermiticity. 
\subsubsection{The transition part}
The transition parts for both Gaussian and Metropolis weights are as follows.
\begin{restatable}[The transition part for Gaussian weights]{cor}{transitionGaussian}\label{cor:transitionGaussian} 
For the Gaussian weight $\gamma(\omega) = \exp\L(- \frac{(\beta\omega + 1)^2}{2}\R)$ with $\omega_{\gamma} = \sigma_E = \sigma_{\gamma} =1/\beta$, the discriminant $\vec{\CH}_{\beta}$ is described in the time domain by
\begin{align}
h_+(t)&= \frac{1}{\beta} \e^{-1/4}\exp(-\frac{4t^2}{\beta^2}) \quad \text{and}\quad h_-(t) = \frac{2}{\pi\beta} \exp( -\frac{2t^2}{\beta^2})\quad \text{such that}\quad \normp{h_-}{1},\normp{h_+}{1} \le 1.\label{eq:h_-}
\end{align}
\end{restatable}

\begin{restatable}[The transition part for Metropolis weights]{cor}{transitionMetropolis}
\label{cor:transitionMetropolis}
For $\sigma_{E}=\frac{1}{\beta}$, the Metropolis-like weight $\gamma^M(\omega) = \exp\L(-\beta\max\left(\omega+\frac{1}{2\beta},0\right)\R)$ yields $\vec{\CH}_{\beta}$ described in the time domain by the same $h_-(t)$ as in Eq.~\eqref{eq:h_-} and by
\begin{align}
	h_+(t)=\frac{\e^{-1/8}}{\beta}\frac{\e^{-2t^2/\beta^2}}{4\sqrt{2\pi}(\frac{t^2}{\beta^2} + \frac{1}{16})}\quad \text{such that}\quad \norm{h_+(t)}_1 \le 1.	
\end{align}
\end{restatable}

\subsubsection{The $\vN$-term}
The time-domain presentation of the $\vN$ term
is as follows.
\begin{restatable}[$\vN$ term for Gaussian weights]{cor}{Ngaussian}\label{cor:NGaussian}
For each $\beta$, the Gaussian weight $\gamma(\omega) = \exp\L(- \frac{(\beta\omega + 1)^2}{2}\R)$ with $\sigma_E = \sigma_{\gamma} = \omega_{\gamma} = \frac{1}{\beta}$ corresponds to the discriminant where
\begin{align}\label{eq:NGaussianDimLessTime}
		\vN&= \sum_{a\in A} 
		\int_{-\infty}^{\infty}n_1(t)\e^{-\ri\beta\vH t} \L(\int_{-\infty}^{\infty}n_2(t')\vA^{a\dagger}(\beta t')\vA^a(-\beta t')\rd t'\R) \e^{\ri \beta\vH t}\rd t,
	\end{align}
	where 
 \begin{align}
     n_1(t) &:=\frac{1}{3}\cdot 2\sqrt{\pi} \L(\!\frac{1}{\cosh\L(2\pi t\R)\!}*_t  \exp\L(\!-2 t^2\R)\!\R)\quad &\text{such that}\quad \norm{n_1}_1 = \frac{\pi}{3\sqrt{2}} < 1\label{eq:n1}\\
 n_2(t) &:=3\cdot 2\cdot \frac{1}{2\pi}\sqrt{\frac{1}{\pi}}\exp\L(-4 t^2-2\ri t \R)=6\cdot b_2(t) \quad &\text{such that}\quad \norm{n_2}_1 < 1
 \end{align}
 with $b_2$ as in~\eqref{eq:b2}. 
\end{restatable}
\begin{restatable}[$\vN$ term for Metropolis weights]{cor}{NMetropolis}\label{cor:NMetropolis}
	If $\sigma_{E}=\frac{1}{\beta}$, then the Metropolis-like weight $\gamma^M(\omega) = \exp\L(-\beta\max\left(\omega+\frac{1}{2\beta},0\right)\R)$
	corresponds to the discriminant where $\vN^M$ satisfies
	\begin{align}\label{eq:Ndiff}
		\nrm{\vN^M-\vN^{M,\eta}}\leq\sum_{a\in A}\nrm{\vA^{a\dagg}\vA^a}\min\left(\frac{6\sqrt{2}\eta\beta\nrm{\vH}}{\pi},\bigO{(\eta\beta\nrm{\vH})^3}\right),
	\end{align}
	where 
\begin{align}\label{eq:N_MetropolisDimLessTime}
		\vN^{M,\eta}:=\sum_{a\in A}\int_{-\infty}^{\infty}n_1(t)\e^{-\ri\beta \vH t}\left( \int_{-\infty}^{\infty}n_2^{M,\eta}(t')\vA^{a\dagger}(\beta t')\vA^a(-\beta t')\rd t'+\frac{3}{4\sqrt{2}\pi}\vA^{a\dagger}\vA^a\right) \e^{\ri \beta\vH t}\rd t,
	\end{align}     
	with $n_1(t)$ as in~\eqref{eq:n1}, and $n_2^{M,\eta} = 6\cdot b_2^{M,\eta}$ as in~\eqref{eq:B2Meta}.
\end{restatable}
In the above cases, we can verify that $\vN$ is Hermitian by $n_1^*(t)=n_1(t)$ and $n_2^*(t)=n_2(-t)$. For $\vN^M$, we also use that $(\vA^{a\dagger}\vA^a)^{\dagger} = \vA^{a\dagger}\vA^a.$
\subsection{Block-encodings}\label{sec:block_encodings}
Our simulation algorithm extensively uses block encodings that are largely borrowed from~\cite{chen2023QThermalStatePrep}. This section aims to instantiate them to state the theorems appropriately, and the curious reader may refer to~\cite[Section III.B]{chen2023QThermalStatePrep}. 
\begin{defn}[Block-encoding for \Lword{}]\label{def:blockLindladian}
	Given a purely irreversible \Lword{} 
\begin{align}
     \CL[\vrho]:=\sum_{j \in J} \L(\vL_j\vrho \vL_j^{\dagger} - \frac{1}{2} \vL_j^{\dagger}\vL_j\vrho - \frac{1}{2}\vrho\vL_j^{\dagger}\vL_j\R), 
\end{align} 
we say that a unitary $\vU$ is a block-encoding for Lindblad operators $\{\vL_j\}_{j\in J}$ if \footnote{In the first register, we could use any orthonormal basis. Sticking to computational basis elements $\ket{j}$ is just for ease of presentation. Intuitively, one can think about $b$ as the number of ancilla qubits used for implementing the Lindblad operators $\vL_j$, while typically $a-b\approx \log|J|$.} 
 \begin{align}
     (\bra{0^b}\otimes \vI)\cdot \vU\cdot(\ket{0^{c}} \otimes \vI)=\sum_{j\in J} \ket{j} \otimes \vL_j \quad \text{for}\quad b\le c \in \mathbb{Z}^{+}.
 \end{align}
\end{defn}
Our implementation of the Gibbs sampler and parent Hamiltonian uses the following components:
\begin{itemize}
	\item Block-encoding $\vV_{jump}$ of the jump operators $\vA^a$ in the form of~\autoref{def:blockLindladian}:
	\begin{align}\label{eq:blockJumps}
		(\bra{0^b}\otimes \vI_{a}\otimes \vI_{sys})]\cdot\vV_{jump}\cdot (\ket{0^c}\otimes \vI_{sys})=\sum_{a\in A} \ket{a} \otimes \vA^a.
	\end{align}
To implement the discriminant, we also assume access to a block-encoding $\vV_{jumpT}$ for the partial transpose $\sum_{a\in A} \ket{a} \otimes (\vA^a)^{T}$.
    \item Quantum Fourier Transform
    \begin{align}
        \vec{QFT}_N: \ket{\bt} \rightarrow \frac{1}{\sqrt{N}} \sum_{\bomega \in S_{\omega_0}} \e^{-\ri \bomega \bt}\ket{\bomega}.
    \end{align}
    We use ``bar'' to denote variables taking discrete values. In particular, the Fourier frequencies $\bomega$ and times $\bt$ are integer multiples of $\omega_0$ and $t_0$ respectively such that 
\begin{align} 
    \omega_0 t_0  = \frac{2\pi}{N}, \quad \text{and} \quad 
    S^{\lceil N \rfloor} &:= \bigg\{ -\left\lceil(N-1)/2\right\rceil, \ldots, -1,0,1,\ldots, \left\lfloor(N-1)/2\right\rfloor \bigg\},\\ 
    \text{and}\quad S^{\lceil N \rfloor}_{\omega_0}&:= \omega_0 \cdot S^{\lceil N \rfloor}, \quad S^{\lceil N \rfloor}_{t_0}:= t_0 \cdot S^{\lceil N \rfloor}.
\end{align} 
    \item Controlled Hamiltonian simulation
    \begin{align}
	\sum_{\bt \in S_{t_0}}\ketbra{\bt}{\bt}\otimes \e^{\pm \ri \bt \vH}.
\end{align}
        \item State preparation oracles for the Fourier Transform weights, acting on the frequency register
        \begin{align}
            \vec{Prep_f}: \ket{\bar{0}} \rightarrow \ket{f}.
        \end{align}
	\item Controlled filter for the Boltzmann factors acting on the frequency register and the Boltzmann weight register 
	\begin{align}
	\vW := \sum_{\bomega\in S_{\omega_0}}  \vY_{1-\gamma(\bomega)}\otimes \ketbra{\bomega}{\bomega} 
	\quad &\text{where} \quad 0 \le \gamma(\bomega) \le 1 
	\end{align}
\item Single qubit Pauli-Y rotations
\begin{align}
\vY_\theta: =\begin{pmatrix} \sqrt{1-\theta} & -\sqrt{\theta}\\ \sqrt{\theta} &  \sqrt{1-\theta} \end{pmatrix}. \label{eq:Y_theta}  
\end{align}
\item Reflection on $b$-qubits
\begin{align}
        \vR_{b}: = 2\ket{0^{b}}\bra{0^{b}} - \vI_{b} 
    \end{align}
\end{itemize}
To feed into the black-box Lindbladian simulation algorithm~\cite[Theorem III.2]{chen2023QThermalStatePrep}, we need block-encodings for the dissipative part and the coherent term; for the coherent Gibbs sampler, we need a block-encoding for the discriminant $\CH_{\beta}$, which we obtain by adding the transition part and the $\vN$ part.

\begin{figure}[t]
\begin{center}
	\newcommand{\scalea}{1.}
	\newcommand{\scaleb}{1.2}
	\begin{quantikz}[wire types={q,b,b,b,b},classical	gap=1mm]
	 	\lstick{\scalebox{\scalea}{$\ket{0}$}}		&\qw &	\qw		&\gate[style={inner xsep=0mm, inner ysep=1mm}]{\scalebox{\scalea}{$\vY_{1/2}$}}	&\qw	&\gate[2,style={inner xsep=1mm, inner ysep=1mm}]{\scalebox{\scaleb}{$\vR_{b+1}$}} &\qw &\gate[style={inner xsep=0mm, inner ysep=1mm}]{\scalebox{\scalea}{$\vY_{1/2}$}}&\qw &\qw &\qw\rstick{\scalebox{\scalea}{$\bra{0}$}}\\
	 	\lstick{\scalebox{\scalea}{$\ket{0^b}$}}		&\qw	&\qw &\gate[3,style={inner xsep=1mm, inner ysep=2mm}]{\scalebox{\scaleb}{$\vV_{jp}$}}	&\qw	&\qw 	&\qw	&\gate[3,style={inner xsep=1mm, inner ysep=2mm}]{\scalebox{\scaleb}{$\vV^{\dagger}_{jp}$}}&\qw&\qw &\rstick{\scalebox{\scalea}{$\bra{0^b}$}} \qw \\[-1mm]
	 	\lstick{\kern-2mm\scalebox{\scalea}{$\ket{0^{c-b}}$}}	&\qw&\qw&\qw	&\qw 					&\qw		&\qw 			& \qw & \qw&\qw &\qw\rstick{\scalebox{\scalea}{$\bra{0^{c-b}}$}} \\[2mm]
	 	\lstick{\scalebox{\scalea}{$\vrho$}}		&\qw&\gate[style={inner xsep=0mm, inner ysep=2mm}]{\scalebox{\scaleb}{$\e^{-\ri\vH \bar{t}}$}}\qw	&	\qw	&\gate[style={inner xsep=2mm, inner ysep=2mm}]{\scalebox{\scaleb}{$\e^{\ri\vH \bar{t}}$}}\qw		&\qw&\gate[style={inner xsep=0mm, inner ysep=2mm}]{\scalebox{\scaleb}{$\e^{-\ri\vH \bar{t}}$}}\qw	&	\qw	&\gate[style={inner xsep=2mm, inner ysep=2mm}]{\scalebox{\scaleb}{$\e^{\ri\vH \bar{t}}$}}\qw		&\qw		& \qw \\
   \lstick{\scalebox{\scalea}{$\ket{\bar{0}}$}}&\gate[style={inner xsep=1mm, inner ysep=1mm}]{\scalebox{\scalea}{$\vec{Prep}'$}} \qw&\ctrl{-1}	&\qw	&\ctrl{-1}			&\qw & \ctrl{-1}&\qw &\ctrl{-1} &\gate[style={inner xsep=1mm, inner ysep=1mm}]{\scalebox{\scalea}{$\vec{Prep}^{\dagger}$}}&\qw \rstick{\scalebox{\scalea}{$\bra{\bar{0}}$}}
	 \end{quantikz}
\end{center} 
  \caption{Circuit for block-encoding~\eqref{eq:fAA}\label{fig:fAA}. The gate $\vec{Prep}$ is a shorthand for $\vec{prep}_{\sqrt{\labs{f_+}}}$ and $\vec{Prep}'$ for $\vec{prep}_{f_+/\sqrt{\labs{f_+}}}$.}
\end{figure}

\begin{prop}[Block-encoding for the coherent term]\label{prop:encoding_coherent}
    Suppose $\norm{f_-}_1,  \norm{f_+}_1 \le 1$. Then, there is a block encoding for
    \begin{align}
        \sum_{\bt_- \in S_{t_0}} f_-(\bt_-)\e^{-\ri \vH \bt_-} \left(\sum_{\bt_+ \in S_{t_0}} f_+(\bt_+) \sum_{a\in A} \vA^{a\dagg}(\bt_+)\vA^a(-\bt_+)\right)\e^{\ri \vH t_-}
    \end{align}
    using constant calls to controlled Hamiltonian simulation, $\vV_{jump}$, $\vec{prep}_{\sqrt{\labs{f_+}}}$,$\vec{prep}_{f_+/\sqrt{\labs{f_+}}}$,  $\vec{prep}_{\sqrt{\labs{f_+}}}$,$\vec{prep}_{f_+/\sqrt{\labs{f_+}}}$ and their adjoints.
\end{prop}
The identical statement applies to the $\vN$ term by replacing $f_{\pm}\rightarrow n_{\pm}$.
\begin{proof}
It suffices to construct 
    \begin{align}
    \sum_{\bt_+ \in S_{t_0}} f_+(\bt_+) \sum_{a\in A} \vA^{a\dagg}(\bt_+)\vA^a(-\bt_+) \label{eq:fAA}   
    \end{align}
    and then apply the operator Fourier Transform. When there is only one jump ($\labs{A} =1$), this merely uses iterations of LCU and controlled Hamiltonian simulation. When a block-encoding gives the set of jumps, see Figure~\ref{fig:fAA}. To see that this yields the desired expression, observe that 
    \begin{align}
        \vY_{1/2}\ket{0} = \ket{+}\quad \text{and}\quad \vY^{\dagger}_{1/2}\ket{0} = \ket{-}.
    \end{align}
Thus, the expression remains the same if we drop the $ - \vI_{b+1} $ term in $\vR_{b+1}$.
\end{proof}

\begin{prop}[Bilinear]\label{prop:bilinear_block}
Suppose $\norm{h_-}_1, \norm{h_+}_1 \leq 1$. Then, there is a block encoding for
\begin{align}
  \sum_{\bt_- \in S_{t_0}} \sum_{\bt_+ \in S_{t_0}} \sum_{a \in A} h_-(\bt_-)h_+(\bt_+) \cdot \vA^a\L(\bt_+-\bt_-\R)\otimes\vA^{a}\L(-\bt_--\bt_+\R)^T
\end{align}
    using constant calls to controlled Hamiltonian simulation, $\vV_{jump}$, $\vV_{jumpT}$, $\vec{prep}_{\sqrt{\labs{h_+}}}$, $\vec{prep}_{h_+/\sqrt{\labs{h_+}}}$,  $\vec{prep}_{\sqrt{\labs{h_+}}}$, $\vec{prep}_{h_+/\sqrt{\labs{h_+}}}$ and their adjoints.
\end{prop}
\begin{proof}
    The proof is a modification of~\autoref{prop:encoding_coherent} by replacing $\vV_{jump}^{\dagger}$ with $\vV_{jump}^{T}$ and duplicating the system register.
\end{proof}

\subsection{Proving~\autoref{thm:L_cost}: complexity for \Lword{} simulation}\label{sec:prove_L_cost}
We now combine the block-encoding to give the overall cost of Gibbs sampling. We can implement the controlled time-evolution up to a truncation time as long as the profiles $f_+ , f_-$ are well-concentrated and smooth in the time domain (and that the Fourier Transforms are well-defined), which is the case since the frequency profiles $\hat{f}_+,\hat{f}_-$ are smooth and concentrated~\autoref{cor:GaussianLikeWeight}-\autoref{prop:MetropolisLikeFilterf}. Also, the following results all require discretization of the time integrals, which is fortunately handled by~\cite[Appendix C]{chen2023QThermalStatePrep}; this sets the required size of the Fourier Transform register, which uses polylogarithmically many qubits. 

\begin{proof}[Proof of~\autoref{thm:L_cost}]
    Use the black-box \Lword{} simulation algorithm~\cite[Theorem III.2]{chen2023QThermalStatePrep} for block-encoding for the coherent term (\autoref{prop:encoding_coherent}) and the dissipative part~\cite[Section III.B.1]{chen2023QThermalStatePrep}. For the Metropolis weight, a logarithmic overhead is incurred for taming the mild singularity in the $b^M_2$; especially, $\vB$ is subnormalized by $ \norm{b^M_2}_1^{-1} = 1/\CO(\log(\beta\nrm{H}/\epsilon))$ to fit an approximation of $\vB$ into a unitary block encoding.
\end{proof}

\subsection{Proving~\autoref{thm:D_cost}: complexity for the discriminant}\label{sec:prove_D_cost}
We may now construct the advertised block encoding for the discriminant. The $\vN$ term is analogous to the coherent term (\autoref{prop:encoding_coherent}). In the implementation, we have discretized the times, which incurred a routine Riemann sum error, which we briefly estimate. It is instructive to visit the one-dimensional integral.

\begin{lem}[Discretization bounds on integrated Heisenberg dynamics]
For any $\vO,$ Hermitian operator $\vH$, and integrable function $f(t),$
\begin{align}
    \lnorm{\int_{-T}^{T} \vO(t) f(t) \rd t - \sum_{\bar{t}\in S_{t_0}} \vO(\bar{t}) f(\bar{t}) t_0} \le \frac{2T^2}{\labs{S_t}} \L(\norm{[\vO,\vH]}\sup_t\labs{f(t)} + \norm{\vO}\sup_t\labs{f(t)'}\R).
\end{align}
\end{lem}
\begin{proof}
Recall an elementary error bound on the Riemann sum over $N$ values
\begin{align}
    \labs{\int_{a}^bf(t) \rd t - \sum f(\bar{t}) t_0} \le \frac{(b-a)^2}{2N} \sup_{t\in [a,b]} \labs{f'(t)}.
\end{align}
Reduce the operator norm to a scalar Riemann sum for each pair of normalized vectors $\ket{\psi},\ket{\phi}$
\begin{align}
\lnorm{\int_{-T}^{T} \vO(t) f(t) \rd t - \sum_{\bar{t}\in S_{t_0}} \vO(\bar{t}) f(\bar{t}) t_0} \le \sup_{\psi,\phi} \labs{\bra{\phi} \L(\int_{-T}^{T} \vO(t) f(t) \rd t - \sum_{\bar{t}\in S_{t_0}} \vO(\bar{t}) f(\bar{t}) t_0 \R)\ket{\psi}}
\end{align}

and bound the derivatives $\labs{\bra{\phi}\frac{\rd}{\rd t} (\vO(t) f(t))\ket{\psi}} \le \norm{\frac{\rd}{\rd t} (\vO(t) f(t))} \le \norm{[\vO,\vH]}\labs{f(t)} +  \norm{\vO}\labs{f'(t)}.$
\end{proof}

\begin{proof}[Proof of~\autoref{thm:D_cost}]
Add the block encodings for the transition part (from \autoref{prop:bilinear_block}), $\vN\otimes \vI$, and $\vI \otimes \vN$ (from \autoref{prop:encoding_coherent}), { and subnormalize by $1/2$ (by~\autoref{prop:parentHam_OFT}).  To bound discretization error for the two-dimensional case with $N_1,$ $N_2$ values, recall
\begin{align}
\labs{\int_{a}^b\int_{c}^d f(t_1,t_2) \rd t_1\rd t_2 - \sum f(\bar{t}_1,\bar{t}_2) \delta t_1 \delta t_2} \le \frac{(b-a) (d-c)}{2} \L( \frac{(b-a)}{N_1}\sup\labs{\partial_{t_1} f(t_1,t_2)} + \frac{(d-c)}{N_2}\sup\labs{\partial_{t_2} f(t_1,t_2)}\R). 
\end{align}
Taking the operator norm for~\autoref{prop:encoding_coherent},\autoref{prop:bilinear_block} for the Gaussian and Metropolis weight, the Riemann sums with truncated integration range $[-T,T] \times [-T,T]$ has an error bound scaling linearly with $\frac{T^3\beta \norm{\vH} \norm{\sum_{a\in A} \vA^{a\dagger} \vA^{a}}}{\labs{S_{t_0}}}$.
}

\end{proof}
Note that we do not implement $\vI+\vec{\CH}_{\beta}$ (as in~\cite[Proposition III.5]{chen2023QThermalStatePrep}) but rather $\vec{\CH}_{\beta}$ itself; implementing the former would allow us to obtain a quadratic speedup on the discriminant gap, which we current do not have.

\section{Discussion}
\label{sec:discussion}
We have constructed the quantum analog of the classical Monte Carlo algorithms with desirable features. We highlight potential future directions as listed.
\begin{itemize}
\item \textbf{Quantum simulation applications.} A key factor in industrial quantum simulation applications~\cite{Chamberland2020BuildingAF,2021_Microsoft_catalysis,THC_google} is effective quantum algorithms for low-energy states. Our algorithm can be employed for any Hamiltonian \textit{without} substantial variational parameters or a case-by-case trial state or adiabatic path. While the mixing time can vary widely, the fact that physically relevant states (molecules or materials) \textit{exist} in Nature suggests a reasonable mixing time in practice. Regarding practical gate complexities, the locality of our algorithm for lattice Hamiltonians may be favorable as we merely need to simulate a $\tCO(\beta)$-radius Hamiltonian patch localized around each jump $\vA^a$.
\item \textbf{Locality and complexity of quantum Gibbs state.} Our algorithm opens new angles on the locality and complexity of Gibbs states (such as the decay of correlation, quantum conditional mutual information, recovery channels, and quantum belief propagation). In particular, the combination of localized jumps and exact detailed balance enables the rigorous study of convergence~\cite{kastoryano2016commuting,capel2021modified} for noncommuting lattice Hamiltonians. Rapid mixing also directly implies the circuit complexity of the purification (Appendix~\ref{sec:arealaw}) through the \Lword{} gap, giving a dynamic perspective on the area law of entanglement~\cite{Hastings2007AnAL}.
\item \textbf{New open-system physics.} 
Just as quantum computing lacks a go-to Monte Carlo algorithm, open system physics lacks a simple, universal \Lword{} that succinctly captures open system thermodynamics. Our algorithm qualifies due to its elegant properties. For example, our \Lword{} enables a precise definition of \textit{dynamical thermal phase transitions} in terms of mixing time that may contrast with static thermal phase transitions. Related concepts include metastable states, the energy landscape, \textit{quantum spin glass}, and self-correcting quantum memories, whose precise formulation for noncommuting Hamiltonians has also been lacking.
\item \textbf{A new algorithmic subroutine.} Classical MCMC algorithms have been widely employed to solve other problems beyond physical simulation, and we may expect the same for our algorithm. A natural example is optimization problems (e.g., constraint satisfaction problems and modern optimization problems), whether applying to classical Hamiltonian (in a setting similar to Quantum Approximate Optimization Algorithms (QAOA)~\cite{farhi2014QAOA}) or quantum Hamiltonians. Another application is Quantum Semidefinite Program Solvers~\cite{brandao2016QSDPSpeedup,apeldoorn2017QSDPSolvers}, where Quantum Gibbs state preparation is routinely invoked. 

\item \textbf{Comparison with existing algorithms.} With a new algorithm at hand, we expect fruitful comparison with existing (quantum or classical) algorithms such as the adiabatic algorithm~\cite{farhi2000QCompAdiabatic}, phase estimation with trial states, tensor network, Quantum Monte Carlo, etc. In particular, understanding the distinction from classical algorithms could either inspire better classical algorithms or expose potential sources of quantum advantage in quantum simulation (e.g., the sign problem or difficulty in contracting PEPS).
\item \textbf{Numerical studies.} As the complement to theory, the explicit form of our \Lword{} also enables direct numerical studies regarding the above notions, e.g., the scaling of mixing time for thermal state or ground states, dynamic phase transitions, and noncommuting quantum memories, and the interplay with tensor networks.
\end{itemize}

To conclude, given the celebrated theoretical and empirical triumph of Markov chain Monte Carlo methods and their successors over the past 70 years, we argue that this work should serve similar roles in quantum computing. Especially given the current skepticism on the practical applicability of quantum computers, our new algorithms bring hope to the community by initiating a new wave of directions covering theory, experiment, numerics, and application. 

\section*{Acknowledgments}
We thank Jonathan Moussa for raising the question of whether exact detailed balance is possible and pointing us to his related work~\cite{Moussa2019LowDepthQM} after the precursor of this work~\cite{chen2023QThermalStatePrep} became public; at that time, we thought that exact detailed balance is incompatible with the energy-time uncertainty principle. 
We also thank Aaron J. Friedman, Jinkang Guo, Oliver Hart, and Andrew Lucas for collaborations on related topics and
Alvaro Alhambra, Anurag Anshu, Simon Apers, Mario Berta, Thiago Bergamaschi, Fernando Brandao, Angela Capel, Garnet Chan, ChatGPT-4, Alex Dalzell, Zhiyan Ding, Steve Flammia, Hsin-Yuan (Robert) Huang, Lin Lin, Yunchao Liu, Sam McArdle, Akshar Ramkumar, Mehdi Soleimanifar, Frederik Nathan, Umesh Vazirani, and Tong Yu for helpful discussions. 
CFC was supported through an internship of the AWS Center for Quantum Computing. AG acknowledges funding from the AWS Center for Quantum Computing.

\bibliographystyle{alphaUrlePrint.bst}
\bibliography{ref,qc_gily}


\newpage
\appendix
\section*{Nomenclature}\label{sec:recap_notation}
This section recapitulates notations.
We write scalars, functions, and vectors in normal font, matrices in bold font $\vO$, and superoperators in curly font~$\CL$. Natural constants $\e, \ri, \pi$ are denoted in Roman font.
\begin{align}
\vH &= \sum_i E_i \ketbra{\psi_i}{\psi_i}&\text{the Hamiltonian of interest and its eigendecomposition}\\
\text{Spec}(\vH) &:= \{ E_i \} & \text{the spectrum of the Hamiltonian}\\
\nu \in B = B(\vH)&:= \text{Spec}(\vH) - \text{Spec}(\vH) &\text{the set of Bohr frequencies}\\
\vP_{E}&:= \sum_{i:E_i = E} \ketbra{\psi_i}{\psi_i}&\text{eigenspace projector for energy $E$}\\
\CL:& & \text{a Lindbladian in the Schrödinger Picture}\\
n: & &\text{ system size (number of qubits) of the Hamiltonian $\vH$}\\
\beta: & &\text{ inverse temperature}\\
\vrho_{\beta}&:= \frac{\e^{-\beta \vH }}{\tr[ \e^{-\beta \vH }]} \quad &\text{the Gibbs state with inverse temperature $\beta$}\\
\ket{\sqrt{\vrho_{\beta}}} &:= \frac{1}{\sqrt{\tr[ \e^{-\beta \vH }]}} \sum_i \e^{-\beta E_i/2} \ket{\psi_i} \otimes \ket{\psi_i^*}\kern-10mm &\text{the purified Gibbs state}\\
\{\vA^a\}_{a \in A}: & &\text{set of jump operators}\\
\labs{A}: & & \text{cardinality of the set of jumps}\\
\vI:& &\text{the identity operator}\\
\bigOt{\cdot},\tOmega (\cdot) :& &\text{complexity expression ignoring (poly)logarithmic factors}
\end{align}
Fourier transform notations:
\begin{align}
\hat{\vA}_{(f)}(\omega) &:= \frac{1}{\sqrt{2\pi}}\int_{-\infty}^{\infty} \e^{-\ri \omega t}f(t) \vA(t)\mathrm{d}t& \text{operator Fourier Transform for $\vA$ weighted by $f$}\\
\hat{f}(\omega)&=\CF(f)=\lim_{K\rightarrow  \infty}\frac{1}{\sqrt{2\pi}}\int_{-K}^{K}\e^{-\ri\omega t} f(t)\mathrm{d}t & \text{the Fourier transform of a scalar function $f$ over inputs $t$}\\		
\vA_\nu&:=\sum_{E_2 - E_1 = \nu } \vP_{E_2} \vA \vP_{E_1} &\text{operator $\vA$ at exact Bohr frequency $\nu$}\end{align}
Norms: 
\begin{align}
	\norm{f(x)}_p&: = \L(\int \labs{f(x)}^p \mathrm{d} x \R)^{\! 1/p} \quad &  \text{the $p$-norm of a scalar function $f$ over inputs $x$ for $p\in[1,\infty]$}\\
\ell_p(\BR)&:= \{f:\BR \rightarrow \BC,\quad \norm{f}_p < \infty\}\quad &\text{the set of integrable functions}\\ 
	\norm{f(x)}&: = \norm{f(x)}_2 = \sqrt{\int \labs{f(x)}^2 \mathrm{d} x} \quad & \text{the 2-norm of a scalar function $f$ over inputs $x$}\\	
	\norm{\ket{\psi}}&: \quad &\text{the Euclidean norm of a vector $\ket{\psi}$}\\
	\norm{\vO}&:= \sup_{\ket{\psi},\ket{\phi}} \frac{\bra{\phi} \vO \ket{\psi}}{\norm{\ket{\psi}}\cdot \norm{\ket{\phi}}} \quad &\text{the operator norm of a matrix $\vO$}\\
  	\norm{\vO}_p&:= (\tr \labs{\vO}^p)^{1/p}\quad&\text{the Schatten p-norm of a matrix $\vO$}\\
  \norm{\CL}_{p-p} &:= \sup_{\vO} \frac{\normp{\CL[\vO]}{p}}{\normp{\vO}{p}}\quad&\text{the induced $p-p$ norm of a superoperator $\CL$}
\end{align}
Linear algebra:
\begin{align}
    \lambda_i(\vO): & \quad &\text{ the $i$-th largest eigenvalue of a matrix $\vO$ sorted by their real parts}\\
 	\lambda_{gap}(\vO)&:=\Re\lambda_1(\vO)-\Re\lambda_2(\vO)\ge 0 \quad &\text{the real spectral gap of a matrix $\vO$}\\
	\vO^*: & \quad & \text{the entry-wise complex conjugate of a matrix $\vO$}\\
 	\vO^\dagger: & \quad & \text{the Hermitian conjugate of a matrix $\vO$}\\
  \ket{\psi^*}&: \quad&\text{ entry-wise complex conjugate of a vector $\ket{\psi}$}
\end{align}

\section{Deriving time-domain represantations}\label{sec:timedomain}
Our calculation for detailed balance has focused on the frequency domain. This appendix applies Fourier transforms to obtain the time-domain representation. The arguments are conceptually straightforward but require some bookkeeping.

For both the Lindbladians and the parent Hamiltonians, we will often encounter a two-dimensional sum over Bohr frequencies. Since there are two energy labels, we employ a two-dimensional Fourier Transform. For any function of frequencies $\hat{f}(\nu_1,\nu_2)$, the time-domain representation of the bilinear expression gives
\begin{align}
    \sum_{a\in A} \sum_{\nu_1,\nu_2\in B} \hat{f}(\nu_1,\nu_2)\vA^{a}_{\nu_1}\otimes (\vA^{a}_{\nu_2})^\dagg =     \sum_{a\in A} \frac{1}{2\pi}\int_{-\infty}^{\infty}\int_{-\infty}^{\infty} f(t_1,t_2) \vA^a(t_1)\otimes \vA^{a\dagger}(-t_2) \rd t_1 \rd t_2 
\end{align}
where we introduced the two-dimensional Fourier Transform
\begin{align}
    f(t_1,t_2): = \frac{1}{2\pi}\int_{-\infty}^{\infty} \int_{-\infty}^{\infty} \hat{f}(\nu_1,\nu_2) \e^{\ri \nu_1 t_1} \e^{\ri \nu_2 t_2}\rd\nu_1\rd \nu_2.\label{eq:2D_FT}
\end{align}
Fortunately, for our usage, the Fourier Transform \textit{decouples} into two iterations of one-dimensional Fourier Transforms, significantly simplifying the presentation and implementation.

\begin{cor}[Factorized time-domain functions]\label{prop:B_timedomain}
If the function factorizes in the energy domain such that
\begin{align}
    \frac{1}{2\pi}\hat{f}(\nu_1,\nu_2) = \hat{f}_+(\nu_1+\nu_2) \cdot \hat{f}_-(\nu_1-\nu_2), 
\end{align}
then  
\begin{align}
 \sum_{a\in A} \sum_{\nu_1,\nu_2\in B} \hat{f}(\nu_1,\nu_2)\vA^{a}_{\nu_1}\otimes (\vA^{a}_{\nu_2})^\dagg = \sum_{a\in A}\int_{-\infty}^{\infty}\int_{-\infty}^{\infty} f_-(t_-) f_+(t_+) \vA^a(-t_--t_+) \otimes \vA^{a\dagg}(t_+-t_-)\rd t_+\rd t_- \label{eq:BGenForm}
\end{align}
and
\begin{align}
    \sum_{a\in A} \sum_{\nu_1,\nu_2\in B} \hat{f}(\nu_1,\nu_2)(\vA^{a}_{\nu_2})^\dagg\vA^{a}_{\nu_1} = \sum_{a\in A}\int_{-\infty}^{\infty}f_-(t_-)\e^{-\ri \vH t_-} \left(\int_{-\infty}^{\infty}f_+(t_+) \vA^{a\dagg}(t_+)\vA^a(-t_+)\rd t_+\right)\e^{\ri \vH t_-}\rd t_-
\end{align}
where the function $f_{\pm}$ are inverse Fourier Transforms of $\hat{f}_{\pm}$.
\end{cor}
Crucially, the RHS can be implemented via Linear Combination of Unitaries by discretizing the integral. 

\begin{proof}
Since the expression is linear in the sum over jumps $a\in A$, it suffices to prove for any operator $\vA$, dropping the jump labels $a$.
\begin{align}
&\sum_{\nu_1,\nu_2\in B} 2\pi\hat{f}_+(\nu_1+\nu_2) \cdot \hat{f}_-(\nu_1-\nu_2) \vA_{\nu_1}\otimes (\vA_{\nu_2})^\dagg \\[-1mm]
&= \sum_{\nu_1,\nu_2\in B} 2\pi\hat{f}_+(\nu_+) \cdot \hat{f}_-(\nu_-)\L(\vA_{\frac{\nu_++\nu_-}{2}}\otimes \vA_{\frac{\nu_+-\nu_-}{2}}\R)^\dagg \tag*{(Let $\nu_+\!:=\nu_1+\nu_2$ and$\nu_-\!:=\nu_1-\nu_2$)}\\&
		= \sum_{\nu_1,\nu_2\in B} \int_{-\infty}^{\infty}f_+(t_+)e^{-\ri \nu_+t_+}\rd t_+\int_{-\infty}^{\infty}f_-(t_-)e^{-\ri \nu_-t_-}\rd t_- \vA_{\frac{\nu_++\nu_-}{2}}\otimes \L(\vA_{\frac{\nu_+-\nu_-}{2}}\R)^\dagg\tag*{(Fourier Transform)}\\&	
		= \sum_{\nu_1,\nu_2\in B} \int_{-\infty}^{\infty}\int_{-\infty}^{\infty}f_+(t_+)f_-(t_-) \vA_{\frac{\nu_++\nu_-}{2}}(-t_+-t_-)\otimes \L(\vA_{\frac{\nu_+-\nu_-}{2}}\R)^\dagg(t_+-t_-)\rd t_+\rd t_-\tag*{(Since $(\vA_{\nu})^\dagg = (\vA^\dagg)_{-\nu}$)}\\
		&= \int_{-\infty}^{\infty}\int_{-\infty}^{\infty}f_+(t_+)f_-(t_-) \vA(-t_+-t_-)\otimes \vA^{\dagg}(t_+-t_-)\rd t_+ \rd t_-\label{eq:A(t_+-t_-)}.	
\end{align}
The fourth equality uses that
\begin{align} 
\e^{-\ri \nu_+t_+}\e^{-\ri \nu_-t_-} = \exp(\frac{\ri(\nu_+-\nu_-)(t_+-t_-)}{2})\cdot \exp(\frac{\ri(\nu_+ +\nu_-)(-t_+-t_-)}{2})
\end{align}
to conclude the proof.
\end{proof}

Now, we plug in the appropriate functions to arrive at the time-domain functions.

\subsection{Our \Lword{}s} \label{sec:L-time}
We evaluate the Fourier transform for the coherent term $\vB$. The expression looks intimidating, but all that matters for the algorithmic complexity is that they decay rapidly (in the time domain).
\begin{cor}[Explicit time-domain functions]\label{cor:timedomain_explicit}
In the time domain, the coherent term $\vB$ in \eqref{eq:BE} corresponding to coefficients constructed in \eqref{eq:alphaGenDef} reads
\begin{align}\label{eq:fpl}
     f_+(t) & =  2 \int_{\frac{\beta\sigma_E^2}{2}}^{\infty}g(x)\sigma_\gamma(x)\exp\L(-\frac{4 t^2 x}{\beta }-2\ri t x\R) \rd x,\\
 f_-(t) & =\frac{\sigma_E}{\pi\beta}\e^{\frac{\beta^2 \sigma_E^2}{8}} \L(\!\frac{1}{\cosh\L(\frac{2\pi t}{\beta}\R)\!}*_t\sin\left(-\beta  \sigma_E^2 t\right)  \e^{-2\sigma_E^2 t^2}\!\R),
 \label{eq:fmn}
\end{align}
depending on parameters $\sigma_E, g(x),\beta$.\footnote{Note that the function $b_1(t)$ seems to have width $\sim\beta\sigma_E$ due to the convolution by $\frac{1}{\cosh\L(4\pi t/(\beta\sigma_E)\R)}$, meaning that the integral in $t$ seems to require Hamiltonian evolution times up to $\sim\beta$. In fact, the numerics show a $1/\mathrm{poly}$ decay until about $\sim\beta\sigma_E$ (after which the exponential decay starts), suggesting that about $\min\L(\sigma_E^{-1}\poly(1/\epsilon),\beta \log(1/\epsilon)\R)$ Hamiltonian simulation time is required in order to achieve $\epsilon$ precision for the block-encoding of the coherent term. Thus, it might be difficult to obtain exact detailed balance below $\Omega(\beta)$ Hamiltonian evolution times. On the contrary, the function $b_2(t')$ has a width only about $1/\sqrt{\beta \omega_{\gamma}}$, implying that the corresponding other integral in $t'$ can be well-approximated by only using Hamiltonian evolution time $\sim \sqrt{\beta/ \omega_{\gamma}}$.
}
\end{cor}
Recall that the convolution of two functions over variable $t$ is defined by
\begin{align}
    (f *_t g) (t) := \int_{-\infty}^{\infty} f(s) g(t-s)\rd s.
\end{align}
Nicely, the product structure persists under a convex combination of Gaussians as only $f_+(t)$ depends on $g(x)$. Otherwise, we could have had to consider a linear combination of function products, which is messier to implement. 
\begin{proof}
First, we confirm that the energy domain function indeed has a product structure $\frac{1}{2\pi}\hat{f}(\nu_1,\nu_2) = \hat{f}_+(\nu_1+\nu_2) \cdot \hat{f}_-(\nu_1-\nu_2)$ due to \eqref{eq:productGaussianAlphas}, \eqref{eq:alphaGenDef} and \eqref{eq:BE} for
\begin{align}
	\hat{f}_+ (\nu) & =  \int_{\frac{\beta\sigma_E^2}{2}}^{\infty}\frac{g(x)\sigma_\gamma(x)}{\sqrt{\sigma_E^2+\sigma_\gamma^2(x)}}\exp\L( - \frac{(\nu +2x)^2}{16x/\beta}\R)\rd x, \text{ and}\label{eq:f+}\\
	\hat{f}_- (\nu) & =  \frac{1}{2\pi}\frac{\tanh(-\beta \nu/4)}{2\ri}\exp\L(-\frac{\nu^2}{8\sigma_E^2}\R)
	=\frac{1}{2\pi}\frac{1}{\cosh(-\beta \nu/4)}\cdot \frac{\sinh(-\beta \nu/4)}{2\ri}\exp\L(-\frac{\nu^2}{8\sigma_E^2}\R)\label{eq:f-}.
\end{align}	

Since we work with well-concentrated integrable functions, the Fourier Transforms exist, and we can compute them as follows. We begin with the Gaussian integral associated with $f_+$
\begin{align}
    \frac{1}{\sqrt{2\pi}}\int_{-\infty}^{\infty} \exp\L( - \frac{(\nu +2x)^2}{16x/\beta}\R)\e^{\ri \nu t}\rd \nu=2\sqrt{\frac{2x}{\beta}}\exp\L(-\frac{4 t^2 x}{\beta }-2\ri t x\R).
\end{align}
Thus, using \eqref{eq:f+} and the definition of $\sigma_{\gamma}(x)= \sqrt{2x/\beta-\sigma_E^2} $, we get \eqref{eq:fpl}. In order to compute $f_-(t)$, we use the convolution theorem $\CF^{-1}\L(\CF(f)\cdot\CF(g)\R)=f*g/\sqrt{2\pi}$. Individually, we have   
\begin{align}
    \frac{1}{\sqrt{2\pi}}\int_{-\infty}^{\infty} \frac{\sinh(-\beta \nu/4)}{2\ri}\exp\L(-\frac{\nu^2}{8\sigma_E^2}\R)\e^{\ri \nu t}\rd \nu&= \sigma_E\e^{\frac{\beta^2 \sigma_E^2}{8}} \sin\left(-\beta  \sigma_E^2 t\right)  \e^{-2 \sigma_E^2 t^2}\label{eq:sinhFT}\\
    \frac{1}{\sqrt{2\pi}}\int_{-\infty}^{\infty}\frac{1}{2 \pi \cosh(-\beta \nu/4)} \e^{\ri \nu t}\rd \nu&=\frac{2}{\sqrt{2\pi} \beta\cosh\L(\frac{2\pi t}{\beta}\R)}.
\end{align}
Take the convolution to conclude the calculation.
\end{proof}

The above explicit form allows us to compute the explicit form of the function $f_+(t)$ corresponding to our two main settings. We begin with the Gaussian case.
\coherentTermGaussian*

Indeed, both functions of time are rapidly decaying and have bounded $\ell_1$-norm (as required for LCU implementation).

\begin{proof}
    Setting $g(x)=\delta(x -\omega_\gamma)$ in \autoref{prop:B_timedomain} yields the desired Gaussian weight and $f_+(t)$ in \eqref{eq:fpl} becomes
\begin{align}\label{eq:fplGauss}
	f_+(t) =  2\int_{-\infty}^{\infty}\delta(x -\omega_\gamma)\sigma_\gamma\exp\L(-\frac{4 t^2 x}{\beta }-2\ri t x\R)dx
	&= 2 \sqrt{\frac{2\omega_{\gamma}}{\beta}-\sigma_{E}^2}\exp\L(-\frac{4 t^2 \omega_{\gamma}}{\beta }-2\ri t \omega_{\gamma}\R).
\end{align}
Setting $b_1(t):=\pi\sqrt{\pi}f_-(t/\sigma_{E})$ and $b_2(t'):=f_+(\beta t')/(\pi\sqrt{\pi})$ and applying a change of variables in the integral \eqref{eq:BGenForm} yields the desired result \eqref{eq:GaussianDimLessTime}. Note that the convolution $*_t$ implicitly is an integral over $t$, so we should not forget to rescale $\rd t$ there. 

Lastly, we bound the $\ell_1$-norm of the functions by Hölder's inequality
\begin{align}
\nrm{b_1}_1&\leq\nrm{(1+t^2)^{-1}}_2\nrm{(1+t^2)b_1}_2=\sqrt{\frac{\pi}{2}}\nrm{(1+t^2)b_1}_2\leq 1,
\end{align}
using individual bounds $\int_{-\infty}^\infty\frac{1}{(1+2t^2)^2}\rd t=\frac{\pi}{2}$ and $\int_{-\infty}^\infty\Big|(1+t^2)b_1(t)\Big|^2 \rd t<0.625$. The norm $\norm{b_2}_1=\frac{e^{-1/4}}{2\pi}< \frac{1}{8}$ is a Gaussian integral (\autoref{fact:gaussian_int}).
\end{proof}

The explicit weights corresponding to Metropolis weights are slightly more cumbersome due to taming a logarithmic singularity; see~\autoref{sec:cal_Metropolis}. 

\subsection{Our parent Hamiltonians}\label{sec:parentHamTime}
Based on the Lindbladian, we explicitly evaluate the discriminant in the frequency domain.

\begin{prop}[Symmerized discriminant]
\label{prop:symmetrized_Dis}
In the setting of~\autoref{cor:ExctDissip}, the discriminant corresponding to the $\vrho_{\beta}$-DB \Lword{} reads
    \begin{align}
        \CH_{\beta}&=  \sum_{a\in A} \sum_{\nu_1,\nu_2 \in B} h_{\nu_1,\nu_2}\vA^a_{\nu_1}(\cdot)(\vA^{a}_{\nu_2})^{\dagger} + \frac{1}{2}(\vN (\cdot) + (\cdot) \vN)\\
        \text{or}\quad\vec{\CH}_{\beta} &= \sum_{a\in A} \sum_{\nu_1,\nu_2 \in B} \undersetbrace{\text{transition part}}{h_{\nu_1,\nu_2}\vA^a_{\nu_1}\otimes (\vA^{a}_{\nu_2})^{*}} + \undersetbrace{\text{coherent and decay part}}{\frac{1}{2}(\vN\otimes \vI + \vI\otimes\vN^{*})}
    \end{align}
where $h_{\nu_1,\nu_2}:= \e^{\beta (\nu_1+\nu_2)/4} \alpha_{\nu_1,\nu_2} = h_{-\nu_2,-\nu_1}$ and 
\begin{align}
    \vN &:= -\sum_{a\in A} \sum_{\nu_1,\nu_2 \in B} \frac{\alpha_{\nu_1,\nu_2}}{\cosh(\beta (\nu_1-\nu_2)/4)}  (\vA^{a}_{\nu_2})^\dagg \vA^a_{\nu_1} = \vN^{\dagger}.
\end{align}
\end{prop}
\begin{proof}
We calculate 
\begin{align}
    \CH_{\beta}&=\vrho^{-1/4}\CL[ \vrho^{1/4}\cdot  \vrho^{1/4}]  \vrho^{-1/4} \\
    &= \sum_{a\in A} \sum_{\nu_1,\nu_2 \in B} \e^{\beta (\nu_1+\nu_2)/4} \alpha_{\nu_1,\nu_2}\vA^a_{\nu_1}(\cdot)(\vA^{a}_{\nu_2})^\dagger\\
&- \ri \sum_{\nu \in B} (\e^{\frac{\beta \nu}{4}}\vB_{\nu}(\cdot)-\e^{-\frac{\beta \nu}{4}}(\cdot) \vB_{\nu}) -\frac{1}{2}\sum_{a\in A} \sum_{\nu_1,\nu_2 \in B} \L(\alpha_{\nu_1,\nu_2} \e^{\frac{\beta (\nu_1-\nu_2)}{4}} (\vA^{a}_{\nu_2})^\dagg \vA^a_{\nu_1} (\cdot) + \alpha_{\nu_1,\nu_2} \e^{-\frac{\beta (\nu_1-\nu_2)}{4}}(\cdot) (\vA^{a}_{\nu_2})^\dagg \vA^a_{\nu_1}
		\R)\\
   &= \sum_{a\in A} \sum_{\nu_1,\nu_2 \in B} h_{\nu_1,\nu_2}\vA^a_{\nu_1}(\cdot)(\vA^{a}_{\nu_2})^\dagger + \frac{1}{2}\vN (\cdot) + \frac{1}{2} (\cdot) \vN^{\dagger}
\end{align}
where
\begin{align}
    \vN &= -2\ri\sum_{\nu\in B}\e^{\beta \nu/4}\vB_{\nu} - \sum_{a\in A} \sum_{\nu_1,\nu_2 \in B} \alpha_{\nu_1,\nu_2} \e^{\frac{\beta (\nu_1-\nu_2)}{4}}(\vA^{a}_{\nu_2})^\dagg \vA^a_{\nu_1}.
\end{align}
We further simplify $\vN$ by expressing $\vB$ as a linear combination of $(\vA^{a}_{\nu_2})^\dagg \vA^a_{\nu_1}$ as in~\eqref{eq:BE}.
\begin{align}
    \vN&=\sum_{a\in A} \sum_{\nu_1,\nu_2 \in B} \bigg(\exp(\beta (\nu_1-\nu_2)/4)\tanh(\beta (\nu_1-\nu_2)/4)- \exp(\beta (\nu_1-\nu_2)/4 )\bigg) \alpha_{\nu_1,\nu_2}(\vA^{a}_{\nu_2})^\dagg \vA^a_{\nu_1}\\
    &= -\sum_{a\in A} \sum_{\nu_1,\nu_2 \in B} \frac{\alpha_{\nu_1,\nu_2}}{\cosh(\beta (\nu_1-\nu_2)/4)}  (\vA^{a}_{\nu_2})^\dagg \vA^a_{\nu_1}\tag*{(By $\e^{x}(\tanh(x)-1) = \frac{-1}{\cosh(x)}$)}\\
    &= \vN^{\dagger} \tag*{(Invariance under $\nu_1\leftrightarrow \nu_2$)},
\end{align}
as advertised.

\end{proof}

{
\begin{prop}[The parent Hamiltonian in terms of operator Fourier transforms]
\label{prop:parentHam_OFT}
    In the setting of~\autoref{prop:symmetrized_Dis}, suppose $\sum_{\nu_1,\nu_2 \in B} \alpha_{\nu_1,\nu_2}\vA^a_{\nu_1}(\cdot)(\vA^{a}_{\nu_2})^{\dagger} = \int_{-\infty}^{\infty}\gamma(\omega) \hat{\vA}^a(\omega)(\cdot)\hat{\vA}^a(\omega)^{\dagger} \rd \omega$ for operator Fourier transform with uncertainty $\sigma.$ Then, the transition part reads
\begin{align}
    e^{-\sigma^2\beta^2/8} \int_{-\infty}^{\infty} \gamma(\omega - \beta \sigma^2/2)e^{\beta \omega /2} \hat{\vA}(\omega)\otimes (\hat{\vA}(\omega))^* \rd \omega, 
\end{align}
and 
\begin{align}
     \vN &= \frac{1}{\sqrt{2\pi}} \int_{-\infty}^{\infty} n(t)\e^{\ri \vH t}\L(\int_{-\infty}^{\infty} \sum_{a} \gamma(\omega) \hat{\vA}^a(\omega)^{\dagger}\hat{\vA}^a(\omega)\rd\omega \R)\e^{-\ri \vH t}\rd t.
\end{align}
    where
    \begin{align}
        n(t) = \frac{2\sqrt{2\pi}}{\beta} \frac{1}{\cosh(2\pi t/\beta )} \quad \text{such that}\quad \norm{n}_1 = \sqrt{2\pi}.
    \end{align}
    Consequently,
    \begin{align}
        \norm{\vec{\CH}} \le (\sup_{\omega} \gamma(\omega) + \sup_{\omega}\gamma(\omega)e^{\beta \omega/2 + \beta^2\sigma^2/4}) \lnorm{\sum_{a}\vA^{a\dagger}\vA^{a}}.
    \end{align}
\end{prop}
In particular, for $\gamma(\omega)$ being Metropolis or Gaussian weights, $\norm{\vec{\CH}} \le 2 \lnorm{\sum_{a}\vA^{a\dagger}\vA^{a}}$. 

\begin{proof}
\begin{align}
    \sum_{\nu_1,\nu_2 \in B} h_{\nu_1,\nu_2}\vA^a_{\nu_1}\otimes (\vA^{a}_{\nu_2})^{*} &= \sum_{\nu_1,\nu_2 \in B} \e^{\beta (\nu_1+\nu_2)/4} \alpha_{\nu_1,\nu_2}\vA^a_{\nu_1}\otimes (\vA^{a}_{\nu_2})^{*}\\
    &= \int_{-\infty}^{\infty} \gamma(\omega) e^{\beta \vH/4}\hat{\vA}(\omega)e^{-\beta \vH/4}\otimes (e^{\beta \vH/4}\hat{\vA}(\omega)e^{-\beta \vH/4} )^* \rd \omega\\
    &= \int_{-\infty}^{\infty} \gamma(\omega)e^{\beta \omega /2 + \sigma^2\beta^2/8} \hat{\vA}(\omega+\frac{\sigma^2\beta}{2})\otimes (\hat{\vA}(\omega+\frac{\sigma^2\beta}{2}))^* \rd \omega\\
    &= e^{-\sigma^2\beta^2/8} \int_{-\infty}^{\infty} \gamma(\omega)e^{\beta \omega_+ /2} \hat{\vA}(\omega_+)\otimes (\hat{\vA}(\omega_+))^* \rd \omega_+ \tag*{(Setting $\omega_+ := \omega+\sigma^2\beta/2$)}.
\end{align}
Therefore, the norm of the transition part, for both Metropolis and Gaussian weights, is bounded by
\begin{align}
    e^{-\sigma^2\beta^2/8} {\sup_{\omega} \gamma(\omega)e^{\beta \omega_+ /2}} \cdot \lnorm{\int_{-\infty}^{\infty} \sum_{a} \hat{\vA}^a(\omega)^{\dagger}\hat{\vA}^a(\omega)\rd\omega} \le {\sup_{\omega} \gamma(\omega)e^{\beta \omega_+ /2}} \norm{\sum_{a} \vA^{a\dagger}\vA^{a}}
\end{align}
using~\cite[Proposition A.2]{chen2023QThermalStatePrep}. 

And the $\vN$ term reads
\begin{align}
    \vN 
    &= \sum_{\nu \in B} \frac{1}{\cosh(\beta \nu/4)}\L(\int_{-\infty}^{\infty} \sum_{a} \gamma(\omega) \hat{\vA}^a(\omega)^{\dagger}\hat{\vA}^a(\omega)\rd\omega \R)_{\nu}\\
    &= \frac{1}{\sqrt{2\pi}} \int_{-\infty}^{\infty} n(t)\e^{\ri \vH t}\L(\int_{-\infty}^{\infty} \sum_{a} \gamma(\omega) \hat{\vA}^a(\omega)^{\dagger}\hat{\vA}^a(\omega)\rd\omega \R)\e^{-\ri \vH t}\rd t \quad \text{where}\quad \hat{n}(-\nu) = \frac{1}{\cosh(\beta \nu/4)}\tag*{(OFT)}.
\end{align}
\begin{align}
    n(t) = \frac{2\sqrt{2\pi}}{\beta} \frac{1}{\cosh(2\pi t/\beta )} \quad \text{such that}\quad \norm{n}_1 = \sqrt{\frac{2}{\pi}} \int_{-\infty}^{\infty}\frac{1}{\cosh(x)}\rd x= \sqrt{2\pi}.
\end{align}
Therefore, 
\begin{align}
    \lnorm{\vN} \le \lnorm{\int_{-\infty}^{\infty} \sum_{a} \gamma(\omega) \hat{\vA}^a(\omega)^{\dagger}\hat{\vA}^a(\omega)\rd\omega} \le \sup_\omega \gamma(\omega) \cdot \norm{\sum_{a} \vA^{a\dagger}\vA^{a}}.
\end{align}
\end{proof}
}
\subsubsection{The transition part}
For the transition part, we quickly obtain the time-domain representation using a two-dimensional Fourier Transform as a corollary of~\autoref{prop:B_timedomain}.
\begin{cor}[Time integrals]\label{cor:b_2dFT}
Suppose $h_{\nu_1,\nu_2} =2\pi \cdot\hat{h}_+(\nu_+)\cdot \hat{h}_-(\nu_-)$, then 
\begin{align}
\sum_{a\in A} \sum_{\nu_1,\nu_2 \in B} h_{\nu_1,\nu_2}\vA^a_{\nu_1}\otimes (\vA^{a}_{\nu_2})^*
 & = \sum_{a \in A} \int_{-\infty}^{\infty}\int_{-\infty}^{\infty} h_-(t_-)h_+(t_+) \cdot \vA^a\L(t_+-t_-\R)\otimes\vA^{a}\L(-t_--t_+\R)^T \rd t_+\rd t_-
\end{align}
where $h_{\pm}(t)$ are Fourier Transforms of $\hat{h}_{\pm}(\nu)$.
\end{cor}

Now, we can evaluate the Fourier transforms explicitly.
\begin{prop}[Linear combination for $h$]\label{prop:linear_comb_h}
For $\alpha_{\nu_1,\nu_2}$ defined in ~\eqref{eq:alphaGenDef} and each $\sigma_E, g(x)$, we have that $h_{\nu_1,\nu_2} =2\pi \hat{h}_+(\nu_+)\cdot \hat{h}_-(\nu_-)$ for the discriminant (\autoref{prop:symmetrized_Dis}) where
\begin{align}
    \hat{h}_+ (\nu) & =  \int_{\frac{\beta\sigma_E^2}{2}}^{\infty}\frac{g(x)\sigma_\gamma(x)}{2\sqrt{\sigma_E^2+\sigma_\gamma^2(x)}}\exp\L( - \frac{\beta\nu^2}{16x} - \frac{\beta x}{4}\R)\rd x \quad \text{and} \quad \hat{h}_-(\nu):= \frac{1}{\pi}\exp\L(-\frac{\nu^2}{8\sigma_E^2}\R),
\end{align}
with the corresponding time-domain functions
\begin{align}
    h_+(t) =& \int_{\frac{\beta\sigma_E^2}{2}}^{\infty}g(x)\sigma_\gamma(x)\exp\L(-\frac{4 t^2 x}{\beta }-\frac{\beta x}{4} \R) \rd x 
    \quad \text{and}\quad h_-(t) = \frac{2\sigma_E}{\pi} \exp( -2\sigma_E^2t^2). \label{eq:h_+t} 
\end{align}
\end{prop}
\begin{proof}
    For $\hat{h}_+(\nu)$ term, the cross term in the exponential $\e^{- \beta \nu/4 }$ is precisely cancelled by $h_{\nu_1,\nu_2}:= \e^{\beta (\nu_1+\nu_2)/4} \alpha_{\nu_1,\nu_2}$; the $\hat{h}_-(\nu)$ term remains the same. To obtain the time-domain functions, we simply carry out the Gaussian integral
\begin{align}
    \frac{1}{\sqrt{2\pi}}\int_{-\infty}^{\infty} \exp\L( - \frac{\beta\nu^2}{16x} - \frac{\beta x}{4}\R) \e^{\ri \nu t} \rd \nu=2\sqrt{\frac{2x}{\beta}}\exp\L(-\frac{4 t^2 x}{\beta }- \frac{\beta x}{4}\R)
\end{align}
and use that $\sigma_E^2+\sigma_\gamma^2(x) = \beta/2x $.
\end{proof}
\begin{prop}[The Gaussian case] 
In the setting of~\autoref{prop:linear_comb_h} and for $\alpha_{\nu_1,\nu_2}$ defined in ~\eqref{eq:alpha_DB}, we have
\begin{align}
\quad \hat{h}_+(\nu) &:= \frac{\sigma_\gamma}{2\sqrt{\sigma_E^2+\sigma_\gamma^2}} \cdot \exp\L( - \frac{\nu^2+(2\omega_\gamma)^2}{8(\sigma_E^2+\sigma_{\gamma}^2)}\R)
\quad \text{and} \quad \hat{h}_-(\nu):= \frac{1}{\pi}\exp\L(-\frac{\nu^2}{8\sigma_E^2}\R).
\end{align}
\end{prop}

\transitionGaussian*
\transitionMetropolis*
\begin{proof} 
    Setting $g(x)=\frac{1}{\sqrt{2\pi}\sigma_\gamma(x)}$ on the interval $\left(\frac{\beta\sigma_E^2}{2},\infty\right)$, yields in \eqref{eq:h_+t}
\begin{align}
        h_+(t) = \int_{\frac{\beta\sigma_E^2}{2}}^{\infty}g(x)\sigma_\gamma(x)\exp\L(-\frac{4 t^2 x}{\beta }-\frac{\beta x}{4} \R) \rd x &=\frac{1}{\sqrt{2\pi}}\int_{\frac{\beta\sigma_E^2}{2}}^{\infty} \exp\L(-\frac{4 t^2 x}{\beta }- \frac{\beta x}{4} \R) \rd x\\
        & =\frac{1}{\beta}\frac{\e^{-2\sigma_E^2t^2 -\beta^2 \sigma_E^2/8}}{4\sqrt{2\pi}(\frac{t^2}{\beta^2} + \frac{1}{16})}.
\end{align}
Set $\sigma_{E} = 1/\beta$ to conclude the calculation. To obtain the $\ell_1$-norm bound, integrate $\frac{\e^{-1/8}}{4 \sqrt{2\pi}} \int_{-\infty}^{\infty}\frac{\e^{-2x^2}}{x^2+1/16} \rd x = \sqrt{\frac{\pi}{2}} \erfc(1/\sqrt{8}) < 0.78$. 
\end{proof}
\subsubsection{The $\vN$-term}
We instantiate the time-domain presentation of the $\vN$ term (adapted from the calculation for the coherent term $\vB$ (\autoref{cor:timedomain_explicit})).

\begin{cor}[Explicit time-domain functions]\label{cor:N_timedomain_explicit}
In the time domain, the $\vN$ term corresponding to coefficients constructed in \eqref{eq:alphaGenDef} can be written as
\begin{align}
    \vN = \sum_{a\in A}\int_{-\infty}^{\infty}f_-(t_-)\e^{-\ri \vH t_-} \left(\int_{-\infty}^{\infty}f_+(t_+) \vA^{a\dagg}(t_+)\vA^a(-t_+)\rd t_+\right)\e^{\ri \vH t_-}\rd t_-\label{eq:Nterm}
\end{align}
for $f_+(t)$ as in~\eqref{eq:fpl} and 
\begin{align}
 f_-(t) & =\frac{2\sigma_E}{\pi\beta} \L(\!\frac{1}{\cosh\L(\frac{2\pi t}{\beta}\R)\!}*_t \e^{-2\sigma_E^2 t^2}\!\R),
\end{align}
depending on parameters $\sigma_{\gamma},\sigma_E, g(x),\beta$.
\end{cor}
\begin{proof}
    Follow the proof of~\autoref{cor:timedomain_explicit}, but drop the $\sinh(-\beta \nu/4)/(2 \ri)$ term in~\eqref{eq:sinhFT}. The convolution then gives
    \begin{align}
    \frac{1}{\sqrt{2\pi}}\int_{-\infty}^{\infty} \exp\L(-\frac{\nu^2}{8\sigma_E^2}\R)\e^{\ri \nu t}\rd \nu&= 2\sigma_E \e^{-2 \sigma_E^2 t^2}.
    \end{align}
\end{proof}
 We instantiate similar results with minor modifications without replicating the proofs.
\Ngaussian* 

\begin{proof}
    As both functions are positive $\frac{1}{\cosh(2\pi t)}, \exp(-2 t^2)>0$, we have that $\norm{n_1(t)}_1=\int_{-\infty}^\infty n_1(t)\rd t=\frac{2\sqrt{\pi}}{3}\L(\int_{-\infty}^\infty \frac{1}{\cosh(2\pi t)}\rd t\R)\cdot \L(\int_{-\infty}^\infty \exp(-2 t^2)\rd t\R)$, where the second equality follows from the fact that the integral of a convolution is the product of the integrals of the convolved functions due to Fubini's theorem. The first integral evaluates to $1/2$ and the second evaluates to $\sqrt{\pi/2}$ due to \autoref{fact:gaussian_int}, therefore $\norm{n_1(t)}_1=\pi/\sqrt{18}$.
    The second bound follows from the analogous bound on $b_2$ in \autoref{cor:GaussianLikeWeight}. The factors of $3$ and $\frac{1}{3}$ are redistributed to ensure both $n_1$ and $n_2$ are normalized. Note the overall extra factor of $2$ in~\autoref{cor:N_timedomain_explicit} compared with~\autoref{cor:timedomain_explicit}.
\end{proof}
\NMetropolis*

\section{Calculating the coherent term for the Metropolis-like weights}\label{sec:cal_Metropolis}

In this section, we dedicate to calculating the weights corresponding to the Metropolis-like weight in order to arrive at an expression that enables efficient implementation via LCU.

\begin{prop}[Metropolis-like weights]\label{prop:MetropolisLikeFilterf}
    Setting $g(x)=\frac{1}{\sqrt{2\pi}\sigma_\gamma(x)}$ on the interval $\left(\frac{\beta\sigma_E^2}{2},\frac{s^2}{\beta}\right)$, yields in \eqref{eq:fpl}
\begin{align}\label{eq:fplTrun}
	f^{(s)}_+(t)=
	\frac{1}{\beta}\frac{\e^{-2\sigma_E^2t^2 -\ri \beta \sigma_E^2 t}-\e^{-4\frac{t^2}{\beta^2} s^2 -2\ri \frac{t}{\beta} s^2}}{\sqrt{2\pi} \frac{t}{\beta} (2 \frac{t}{\beta}+\ri)}.
\end{align}
In particular, in the $s\rightarrow \infty$ limit, the second term vanishes, but we can only interpret the result as a distribution 
\begin{align}\label{eq:fplTrunInf}
	f^{(\infty)}_+(t)=\lim_{\eta\rightarrow 0+}\indicator(|t|\geq\eta)\frac{1}{\beta}\frac{\e^{-2\sigma_E^2t^2 -\ri \beta \sigma_E^2 t}}{\sqrt{2\pi} \frac{t}{\beta} (2 \frac{t}{\beta}+\ri)}+\sqrt{\frac{\pi}{2}}\delta(t).
\end{align}
\end{prop}
Which again decays rapidly; we would need to pay attention to the $t\sim 0$ regime due to the delta function $\delta(t)$ and $1/t$ dependence.

\begin{proof} We substitute $g(x)=\frac{1}{\sqrt{2\pi}\sigma_\gamma(x)}$ in \eqref{eq:fpl}, which due to $\sigma_E^2+\sigma_\gamma^2=\frac{2\omega_{\gamma}}{\beta}$ gives
	\begin{align}
		f^{(s)}_+(t) :=2 \int_{\frac{\beta\sigma_E^2}{2}}^{\infty}\frac{1}{\sqrt{2\pi}\sigma_\gamma(x)}\sigma_\gamma(x)\exp\L(-\frac{4 t^2 x}{\beta }-2\ri t x\R) \rd x
		&=\frac{2}{\sqrt{2\pi}}\int_{\frac{\beta\sigma_E^2}{2}}^{\frac{s^2}{\beta}}\exp\L(-\frac{4 t^2 x}{\beta }-2\ri t x\R) \rd x\\&
		=\frac{1}{\beta}\frac{\e^{-2\sigma_E^2t^2 -\ri \beta \sigma_E^2 t}-\e^{-4\frac{t^2}{\beta^2} s^2 -2\ri \frac{t}{\beta} s^2}}{\sqrt{2\pi} \frac{t}{\beta} (2 \frac{t}{\beta}+\ri)}.
	\end{align}
In order to properly understand the $s\rightarrow \infty$ limit we also need to analyze \eqref{eq:f+}:
\begin{align}\label{eq:fplTrunHat}
	\hat{f}^{(s)}_+(\nu) &=   \int_{\frac{\beta\sigma_E^2}{2}}^{\infty}\frac{g(x)\sigma_\gamma}{\sqrt{\sigma_E^2+\sigma_\gamma^2}}\exp\L( - \frac{(\nu +2x)^2}{16x/\beta}\R)\rd x
	=\int_{\frac{\beta\sigma_E^2}{2}}^{\frac{s^2}{\beta}}\sqrt{\frac{\beta}{4\pi x}}\exp\L( - \frac{(\nu +2x)^2}{16x/\beta}\R)\rd x\\&
	=\frac{1}{2} \left(\e^{-\frac{\beta \nu}{2}}\erf\left(\frac{\nu}{\sqrt{8} \sigma_E\!}-\frac{\beta 
		\sigma_E}{\sqrt{8}}\right)
	\!-\erf\left(\frac{\nu}{\sqrt{8} \sigma_E\!}+\frac{\beta 
		\sigma_E}{\sqrt{8}}\right) 
	\!+ \e^{-\frac{\beta \nu}{2}}\erf\left(\frac{s}{2}-\frac{\beta \nu}{4 s}\right)
	\!+\erf\left(\frac{s}{2}+\frac{\beta \nu}{4 s}\right)\right),
\end{align}
which in the $s\rightarrow \infty$ limit becomes
\begin{align}\label{eq:fplMetHat}
	\hat{f}^{(\infty)}_+(\nu) =\frac{1}{2} \left(\erfc\left(\frac{\beta 
		\sigma_E}{\sqrt{8}}+\frac{\nu}{\sqrt{8} \sigma_E\!}\right)+\e^{-\frac{\beta \nu}{2}}\erfc\left(\frac{\beta 
		\sigma_E}{\sqrt{8}}-\frac{\nu}{\sqrt{8} \sigma_E\!}\right)\right).
\end{align}

In order to compute $f^{(\infty)}_+(\nu)$, let us consider the function $\hat{f}^{(\infty)}_{+,0}(\nu):=\hat{f}^{(\infty)}_+(\nu)-\frac{1-\sgn(\nu)}{2}$, where 
\begin{align}
 \sgn(\nu) = \begin{cases}
     0 \quad &\text{if}\quad \nu =0\\
\sgn(\nu)=\nu/|\nu| &\text{if}\quad \nu\neq 0     
 \end{cases}.   
\end{align}
Since $\hat{f}^{(\infty)}_{+,0}(\nu)\in\ell_1$, its inverse Fourier Transform exists, and it can be computed as $f^{(\infty)}_{+,0}(t)=\frac{1}{\beta}\frac{\e^{-2\sigma_E^2t^2 -\ri \beta \sigma_E^2 t}+2\ri\frac{t}{\beta}-1}{\sqrt{2\pi} \frac{t}{\beta} (2 \frac{t}{\beta}+\ri)}$ after a change of variables by utilizing the (complex) Laplace transform of the error function $\erf$ \cite{wiki2023LaplaceTransform}. Moreover, since $\hat{f}^{(\infty)}_{+,0}(\nu)$ is smooth, apart from $\nu=0$ where its value is the mean of its left and right side limits, standard results in the theory of Fourier integrals (see, e.g., \cite[Chapter 18.3.1.d]{zorich2016MathAnalII}) \nolinebreak imply
\begin{align}
	\hat{f}^{(\infty)}_{+,0}(\nu)=\lim_{\eta\rightarrow 0+}\frac{1}{\sqrt{2\pi}}\int_{-\frac{1}{\eta}}^{\frac{1}{\eta}}f^{(\infty)}_{+,0}(\nu)e^{-\ri \nu t}\rd t
	=\lim_{\eta\rightarrow 0+}\frac{1}{\sqrt{2\pi}}\int_{-\frac{1}{\eta}}^{\frac{1}{\eta}}\frac{1}{\beta}\frac{\e^{-2\sigma_E^2t^2 -\ri \beta \sigma_E^2 t}+2\ri\frac{t}{\beta}-1}{\sqrt{2\pi} \frac{t}{\beta} (2 \frac{t}{\beta}+\ri)}\rd t.
\end{align} 
On the other hand, it is a standard calculation that
\begin{align}
	\frac{1}{2}\sgn(\nu)=
	\lim_{\eta\rightarrow 0+}\frac{1}{4\pi}\int_{\eta}^{\infty}\frac{\sin(\nu t)}{t}\rd t
	=\lim_{\eta\rightarrow 0+}\frac{1}{\sqrt{2\pi}}\left(\int_{-\frac{1}{\eta}}^{-\eta}\frac{\ri}{\sqrt{2\pi}t }\e^{-\ri \nu t}\rd t + \int_{\eta}^{\frac{1}{\eta}}\frac{\ri}{\sqrt{2\pi}t }\e^{-\ri \nu t}\rd t\right).
\end{align}
Since $f^{(\infty)}_{+,0}(\nu)$ is bounded in the neighborhood of $0$ we can combine the above two expression, yielding 
\begin{align}
	\hat{f}^{(\infty)}_{+}(\nu)-\frac{\sgn(\nu)}{2}&=\hat{f}^{(\infty)}_{+,0}(\nu)-\frac{1}{2}\\&
	=\lim_{\eta\rightarrow 0+}\frac{1}{\sqrt{2\pi}}\int_{-\infty}^{\infty}\indicator(|t|\in (\eta,1/\eta))\left(f^{(\infty)}_{+,0}(\nu)-\frac{\ri}{\sqrt{2\pi}t}\right)e^{-\ri \nu t}\rd t\\&
	=\lim_{\eta\rightarrow 0+}\frac{1}{\sqrt{2\pi}}\int_{-\infty}^{\infty}\indicator(|t|\in (\eta,1/\eta))\frac{1}{\beta}\frac{\e^{-2\sigma_E^2t^2 -\ri \beta \sigma_E^2 t}}{\sqrt{2\pi} \frac{t}{\beta} (2 \frac{t}{\beta}+\ri) }\e^{-\ri \nu t}\rd t\\&
	=\lim_{\eta\rightarrow 0+}\frac{1}{\sqrt{2\pi}}\int_{-\infty}^{\infty}\indicator(|t|\geq\eta)\frac{1}{\beta}\frac{\e^{-2\sigma_E^2t^2 -\ri \beta \sigma_E^2 t}}{\sqrt{2\pi} \frac{t}{\beta} (2 \frac{t}{\beta}+\ri) }\e^{-\ri \nu t}\rd t,
\end{align}
where the last equality holds because of the rapid decay of the function for large $t$ values. We can conclude the proof by observing that the Fourier Transform of $\sqrt{\frac{\pi}{2}}\delta(t)$ is $\frac{1}{2}$.
\end{proof}

In the $s=\infty$ case, the diverging $\ell_1$ norm of $f^{(\infty)}_+(t)$ raises the question: how big could the norm $\nrm{\vB^M}$ be? 
In the following we spell out the explicit form of $\vB$ by removing the singularity at $t= 0$, which shows that in fact $\nrm{\vB^M}\leq\bigO{\log(\beta\nrm{\vH})}$, because $\nrm{f_-}_1=\bigO{1}$.
Although the exact formula becomes cumbersome, it enables us to find a very precise approximation by a much simpler formula in \autoref{cor:coherentMetropolis}.

\begin{prop}[Exact form of Metropolis coherent term]\label{prop:exactMetro}
	The coherent term $\vB^M$ corresponding to the quasi-Metropolis weight $\gamma^{(\infty)}_{\sigma_E}(\omega)$ in \eqref{eq:convolutionFilter} can be written as
	\begin{align}\vB^M =  \sum_{a\in A}\int_{-\infty}^{\infty}f_-(t_-)\e^{-\ri \vH t_-} \vO_a^M\e^{\ri \vH t_-}\rd t_-	,
	\end{align}
	where $f_-(t_-)$ is defined in \eqref{eq:fmn}, and for arbitrary $\theta > 0$ we have 
	\begin{align}
		\vO_a^M&=\int_{-\infty}^{\infty}\e^{\ri \vH t}\vA^{a\dagg}\left(\frac{\e^{-2\sigma_E^2t^2 -\ri \beta \sigma_E^2 t}+\indicator(|t|\leq\theta)\ri\left(2\frac{t}{\beta}+\ri\right)}{\sqrt{2\pi} t (2 \frac{t}{\beta}+\ri)}\e^{-2\ri \vH t}-\indicator(|t|\leq\theta)\sqrt{\frac{2}{\pi}}\vH\sinc(2\vH t)\right) \vA^a\e^{\ri \vH t}\rd t\\&
		\kern-6mm+\frac{1}{\sqrt{2\pi}}\int_{-\theta}^{\theta}\cos(\vH t)\vA^{a\dagg}\cos(2 \vH t)
		\vA^{a}\vH\sinc(\vH t)+\sinc(\vH t)\vH\vA^{a\dagg}\cos(2 \vH t)
		\vA^{a}\cos(\vH t)\rd t
		+\sqrt{\frac{\pi}{2}}\vA^{a\dagg}\vA^a.
	\end{align}
\end{prop}
\begin{proof}
	By \autoref{prop:B_timedomain}-\autoref{prop:MetropolisLikeFilterf} we know that 
	\begin{align}
		\vO_a^M&=\lim_{\eta\rightarrow 0+}\int_{-\infty}^{\infty}\left(\indicator(|t|\geq\eta)\frac{\e^{-2\sigma_E^2t^2 -\ri \beta \sigma_E^2 t}}{\sqrt{2\pi} t (2 \frac{t}{\beta}+\ri)}+\sqrt{\frac{\pi}{2}}\delta(t)\right) \vA^{a\dagg}(t)\vA^a(-t)d t.
	\end{align}
	We decompose the above integral in order to remove its singularity at $t=0$. We start with the decomposition 
 \begin{align}
     \frac{\e^{-2\sigma_E^2t^2 -\ri \beta \sigma_E^2 t}}{\sqrt{2\pi} t (2 \frac{t}{\beta}+\ri)}=\frac{\e^{-2\sigma_E^2t^2 -\ri \beta \sigma_E^2 t}+\indicator(|t|\leq\theta)\ri\left(2\frac{t}{\beta}+\ri\right)}{\sqrt{2\pi} t (2 \frac{t}{\beta}+\ri)}
	-\indicator(|t|\leq\theta)\frac{\ri}{\sqrt{2\pi} t},
 \end{align} 
implying that
	\begin{align}
		\frac{\e^{-2\sigma_E^2t^2 -\ri \beta \sigma_E^2 t}}{\sqrt{2\pi} t (2 \frac{t}{\beta}+\ri)}\vA^{a\dagg}(t)\vA^a(-t)
		&=\e^{\ri \vH t}\vA^{a\dagg}\cdot\left(
		\frac{\e^{-2\sigma_E^2t^2 -\ri \beta \sigma_E^2 t}}{\sqrt{2\pi} t (2 \frac{t}{\beta}+\ri)}\e^{-2\ri \vH t}\right)\vA^{a}\e^{\ri \vH t}\\&	
		=\e^{\ri \vH t}\vA^{a\dagg}\cdot \left(\frac{\e^{-2\sigma_E^2t^2 -\ri \beta \sigma_E^2 t}+\indicator(|t|\leq\theta)\ri\left(2\frac{t}{\beta}+\ri\right)}{\sqrt{2\pi} t (2 \frac{t}{\beta}+\ri)}
		-\indicator(|t|\leq\theta)\frac{\ri}{\sqrt{2\pi} t}\right)\e^{-2\ri \vH t}\vA^{a}\e^{\ri \vH t},\\[-10mm]	
	\end{align}
	where
	\begin{align}
		\frac{\ri}{\sqrt{2\pi} t}\e^{-2\ri \vH t}
		&=\frac{\ri}{\sqrt{2\pi} t}\cos(2 \vH t)
		+\frac{1}{\sqrt{2\pi} t}\underset{\!\!2\vH t\sinc(2\vH t)\!\!}{\underbrace{\sin(2\vH t)}}	
		=\ri\frac{\cos(2 \vH t)}{\sqrt{2\pi}t}
		+\sqrt{\frac{2}{\pi}}\vH\sinc(2\vH t).	
	\end{align}
	Finally, observe that due to parity reasons, we have
	\begin{align}
		\int_{-\infty}^{\infty}&\indicator(|t|\geq\eta)\e^{\ri \vH t}\vA^{a\dagg}\indicator(|t|\leq\theta)\frac{\ri}{\sqrt{2\pi}}\frac{\cos(2 \vH t)}{t}
		\vA^{a}\e^{\ri \vH t}\rd t\\&
		=\frac{-1}{\sqrt{2\pi}}\int_{-\infty}^{\infty}\indicator(\eta\leq|t|\leq\theta)\left(\cos(\vH t)\vA^{a\dagg}\frac{\cos(2 \vH t)}{t}
		\vA^{a}\sin(\vH t)+\sin(\vH t)\vA^{a\dagg}\frac{\cos(2 \vH t)}{t}
		\vA^{a}\cos(\vH t)\right)\rd t\\&
		=\frac{-1}{\sqrt{2\pi}}\int_{-\infty}^{\infty}\indicator(\eta\leq|t|\leq\theta)\bigg(\cos(\vH t)\vA^{a\dagg}\cos(2 \vH t)
		\vA^{a}\vH\sinc(\vH t)+\sinc(\vH t)\vH\vA^{a\dagg}\cos(2 \vH t)
		\vA^{a}\cos(\vH t)\bigg)\rd t.
	\end{align}	
	Since in the above provided decomposition, every (matrix) function is bounded in the neighborhood of $0$, we can obtain the $\eta\rightarrow 0+$ limit by simply removing the indicator $\indicator(|t|\geq\eta)$.
	\begin{align}
		\lim_{\eta\rightarrow 0+}\int_{-\infty}^{\infty}&\indicator(|t|\geq\eta)\frac{\e^{-2\sigma_E^2t^2 -\ri \beta \sigma_E^2 t}}{\sqrt{2\pi} t (2 \frac{t}{\beta}+\ri)} \vA^{a\dagg}(t)\vA^a(-t)d t\\&
		=\int_{-\infty}^{\infty}\e^{\ri \vH t}\vA^{a\dagg}\cdot \left(\frac{\e^{-2\sigma_E^2t^2 -\ri \beta \sigma_E^2 t}+\indicator(|t|\leq\theta)\ri\left(2\frac{t}{\beta}+\ri\right)}{\sqrt{2\pi} t (2 \frac{t}{\beta}+\ri)}\e^{-2\ri \vH t}
		-\indicator(|t|\leq\theta)\sqrt{\frac{2}{\pi}}\vH\sinc(2\vH t)\right)\vA^{a}\e^{\ri \vH t}\rd t\\&
		+\frac{1}{\sqrt{2\pi}}\int_{-\theta}^{\theta}\cos(\vH t)\vA^{a\dagg}\cos(2 \vH t)
		\vA^{a}\vH\sinc(\vH t)+\sinc(\vH t)\vH\vA^{a\dagg}\cos(2 \vH t)
		\vA^{a}\cos(\vH t)\rd t.
	\end{align}			
	We conclude by noting that $\int_{-\infty}^{\infty}\sqrt{\frac{\pi}{2}}\delta(t) \vA^{a\dagg}(t)\vA^a(-t)d t=\sqrt{\frac{\pi}{2}}\vA^{a\dagg}\vA^a$.
\end{proof}
Note that this exact formula could be directly and efficiently implemented using QSVT. However, the gains are minimal, as it would only reduce subnormalization from $\bigO{\log(\beta\nrm{H}/\epsilon)}$ (this being the $\ell_1$ norm required to achieve $\epsilon$ precision in \autoref{cor:coherentMetropolis}) to $\bigO{\log(\beta\nrm{H})}$ (the $\ell_1$ norm of the weight function corresponding to the natural choice $\theta=1/{\beta\nrm{\vH}}$ in \autoref{prop:exactMetro}). 

More interestingly, this exact formula also seems to hold in the infinite-dimensional case, giving rise to exact detailed balance in the infinite-dimensional version as well. We leave it for future work to verify that the construction and its analysis can indeed be generalized to infinite-dimensional systems.

Moreover, the above result actually also implies that the $\Lword{}$ corresponding to \eqref{eq:fplTrun} also well-approximates its $s\rightarrow \infty$ limit, i.e., the above Metropolis-like $\Lword{}$. The argument goes as follows: first, ``round'' the Hamiltonian $\vH$ to discretize its spectrum at some finite resolution $\ll 1/\beta$. Due to the form of \eqref{eq:fplTrun}, the resulting perturbation of the \Lword{} is bounded. Then, take the $s\rightarrow \infty$ limit. Since \eqref{eq:fplTrunHat} is exponentially close to its limit \eqref{eq:fplMetHat} and there is a limited number of Bohr frequencies due to rounding, the resulting perturbation of the \Lword{} is once again bounded. Finally, undo the rounding, which again causes a bounded perturbation due to the exact form of the Metropolis coherent term (\autoref{prop:exactMetro}). Carefully executing these bounds in the $\sigma_{E}=\beta$ case should show that for $s=\Theta((\beta\nrm{H}+\log(1/\epsilon)+1)^2)$ the resulting \Lword{} is $\epsilon$-close to its $s=\infty$ limit.

Now we show how to approximate the exact Metropolis \Lword{} in a different efficient way.

\MetropolisLikeFilter*
\begin{proof}
First, let us establish the norm bound~\eqref{eq:B2Meta}. Observe that
\begin{align}
	b_2^{M,\eta}(t)&=b_2^{M,1}(t)-\frac{1}{2\sqrt{2}\pi}\indicator(\eta<|t|\leq1)\frac{\ri}{t}.
\end{align}
Using the above identity and applying a triangle inequality and then Hölder's inequality gives 
\begin{align}
    \normp{b_2^{M,\eta}}{1} &\leq\nrm{(1+4t^2)^{-1}}_2\nrm{(1+4t^2)b_2^{M,1}}_2\!+\!\frac{1}{2\sqrt{2}\pi}\lnormp{\indicator(\eta<|t|\leq1)\frac{\ri}{t}}{1}
\end{align}
To obtain~\eqref{eq:B2Meta}, use that $\int_\eta^1\frac{1}{t}\rd t=\ln(1/\eta)$,  $\int_{-\infty}^\infty\frac{1}{(1+4t^2)^2}\rd t=\frac{\pi}{4}$, $\int_{1}^\infty\big|\frac{\exp(-2t^2)(2t-\ri)}{t}\big|^2\rd t=\e^{-4}-\sqrt{\pi}\erfc(2)$, and a direct computation of 
$\int_{-1}^1\Big|\frac{\exp\L(-2t^2 -\ri t\R)(2t-\ri)+\ri\left(1+4t^2\right)}{t}\Big|^2<16$.	
	
Next, setting $\theta=\eta\beta$ in \autoref{prop:exactMetro} and scaling the variables $t,t'$ by a factor of $\beta$ reveals that 
	\begin{align}\label{eq:BdiffProof}
		\vB^M-\vB^{M,\eta}=\frac{1}{\sqrt{2}\pi^2}\sum_{a\in A}\int_{-\infty}^{\infty}b_1(t)\e^{-\ri \beta\vH t} \vQ_a\e^{\ri \beta\vH t}\rd t,
	\end{align}
	where
	\begin{align}
		\vQ_a&=
		\frac{1}{2}\int_{-\eta}^{\eta}\cos(\beta\vH t')\vA^{a\dagg}\cos(2 \beta\vH t')
		\vA^{a}\beta\vH\sinc(\beta\vH t')+\sinc(\beta\vH t')\beta\vH\vA^{a\dagg}\cos(2 \beta\vH t')
		\vA^{a}\cos(\beta\vH t')\rd t'\\&
		-\int_{-\eta}^{\eta}\e^{\ri \beta\vH t'}\vA^{a\dagg}\big(\beta\vH\sinc(2\beta\vH t')\big) \vA^a\e^{\ri \beta\vH t'}\rd t'.
	\end{align}
	We can decompose in the second term $\e^{\ri \beta\vH t'}=\cos(\beta\vH t')+\ri\sin(\beta\vH t')$. Due to parity reasons, we can see that the second term can be replaced by 
	\begin{align}
		-\int_{-\eta}^{\eta}\cos(\beta\vH t')\vA^{a\dagg}\big(\beta\vH\sinc(2\beta\vH t')\big) \vA^a\cos(\beta\vH t')-\sin(\beta\vH t')\vA^{a\dagg}\big(\beta\vH\sinc(2\beta\vH t')\big) \vA^a\sin(\beta\vH t')\rd t'.
	\end{align}	
	Now, let us define
	\begin{align}
		\mu_a(\vX):=&\cos(\vX)\vA^{a\dagg}\cos(2 \vX)
		\vA^{a}\vX\sinc(\vX)+\sinc(\vX)\vX\vA^{a\dagg}\cos(2 \vX)
		\vA^{a}\cos(\vX)\\&
		-2\cos(\vX)\vA^{a\dagg}\big(\vX\sinc(2\vX)\big) \vA^a\cos(\vX)
		+2\sin(\vX)\vA^{a\dagg}\big(\vX\sinc(2\vX)\big) \vA^a\sin(\vX),
	\end{align}	
	so that 
	\begin{align}
		\vQ_a&=
		\int_{-\eta}^{\eta}\frac{\mu_a(\beta\vH t')}{2t'}\rd t'.
	\end{align}
	If $\vX$ is Hermitian, by the triangle inequality, we can see that $\nrm{\mu_a(\vX)}\leq 6\nrm{\vA^a}^2\nrm{\vX}$.
	Moreover, from its definition, we can see that $\mu_a(\vX)$ is an odd analytic entire function, and its derivative is $0$ at $\vX=0$.
	From this it directly follows that $\mu_a(\vX)=\bigO{\nrm{\vA^a}^2\nrm{\vX}^3}$ for every Hermitian $\vX$. 
	Due to the triangle inequality,   
	\begin{align}
		\nrm{\vQ_a}&\leq
		\int_{-\eta}^{\eta}\!\frac{\nrm{\mu_a(\beta\vH t')}}{2|t'|}\rd t'
		\leq \int_{-\eta}^{\eta}\nrm{\vA^a}^2\nrm{\beta\vH}\min\left(\frac{6}{2},\bigO{\nrm{\beta\vH t'}^2}\right)\rd t'
		=\nrm{\vA^a}^2\min\bigg(6\eta\beta\nrm{\vH},\bigO{(\eta\beta\nrm{\vH})^3}\bigg).
	\end{align}	
	Finally, we can prove \eqref{eq:Bdiff} for individual jumps by combining the above bound and the inequality $\nrm{b_1}_1\leq 1$ within \eqref{eq:BdiffProof}. To obtain the bound in terms of the norm-of-sum $\norm{\sum_a \vA^{a\dagger}\vA^a}$, note that for any matrices $\vX,\vY,\vZ$, we have
 \begin{align}
     \norm{\sum_a \vX \vA^{a\dagger}\vY \vA^a\vZ}&\le \norm{\vX}\cdot \norm{\sum_a \vA^{a\dagger}\otimes \bra{a}}\cdot \norm{\vY\otimes\vI} \cdot\norm{\sum_a \vA^a\otimes \ket{a}}\cdot \norm{\vZ} \\
     &=\norm{\vX} \norm{\vY} \norm{\vZ} \cdot \norm{\sum_a \vA^{a\dagger}\vA^a}.
 \end{align}
 And rewrite $\mu_a(\vX)$ as a linear combination of terms each bounded by $\CO(\norm{\sum_a \vA^{a\dagger}\vA^a}\norm{\vX}^3)$ by canceling out the linear $\CO(\vX)$ terms.
\end{proof}

\section{The parent Hamiltonian and quasi-adiabatic algorithms}\label{sec:arealaw}

The parent Hamiltonian $\vec{\CH}_{\beta}$ can be understood as a frustration-free Hamiltonian (up to a sign, \autoref{prop:exact_discriminant_fixedpoint}). In this section, we import existing analytic techniques for studying gapped phases to understand the structure of the purified Gibbs state. Consider a geometrically local Hamiltonian on $D$-dimensional lattice with single-site jumps $\norm{\vA^a} = 1$ which acts \textit{extensively} $\norm{\sum_{a \in A} \vA^{a\dagger}\vA^a} \propto n$) on the lattice (here, the natural normalization is such that $\norm{\vec{\CH}_{\beta}} \propto n.$). Then, the associated parent Hamiltonian also respects the geometric locality, up to an exponentially decaying tail. In particular, if the parent Hamiltonian has a system-sized independent spectral gap across a range of temperatures $[0,\beta],$ then one can efficiently prepare the purified Gibbs state by employing \textit{quasi-adiabatic evolution}~\cite{hastings2005quasiadiabatic,hastings2010quasi,bachmann2012automorphic} across the temperatures.

\begin{lem}[{Quasi-adiabatic evolution~\cite[Proposition 2.4]{bachmann2012automorphic}}]\label{lem:quasi_adibatic_W}
    Consider a one-parameter family of Hamiltonians 
    \begin{align}
     \vH(s)\quad \text{for}\quad s \in [0,1]   \quad \text{with minimal gap}\quad \Delta.
    \end{align} 
    Then, the family of ground-state projectors can be generated by a time-dependent Hamiltonian
    \begin{align}
        \vP'(s) = \ri [ \vW(s),\vP(s)]\quad \text{where}\quad 
    \vW(s):= \int_{-\infty}^{\infty}\rd t\ w(t) \int_0^t du \ \e^{\ri t \vH(s)}\vH'(s)\e^{-\ri t \vH(s)}\rd t\quad \text{for}\quad s \in [0,1]\label{eq:W}
\end{align}
for any weight function $w(t)$ satisfying $\int_{-\infty}^{\infty} \labs{w(t)}\rd t < \infty,$ 
\begin{align}
    \int_{-\infty}^{\infty} w(t)\rd t &=1 \quad \text{and}\quad \tilde{w}(\omega) = 0 \quad \text{if}\quad \labs{\omega} \ge \Delta.
\end{align}
\end{lem}
In particular, a handy choice of the filter is \cite[Lemma 2.3]{bachmann2012automorphic}. 
\begin{align}
    w_{\Delta}(t) = c_{\Delta} \prod_{n=1}^{\infty} \L(\frac{\sin(a_n t)}{a_n t}\R)^{2}\quad \text{where}\quad a_n = \frac{a_1}{n\ln(n)} \quad \text{for}\quad n \ge 2
\end{align}
and $a_1$ such that $\sum_{n=1}^{\infty} a_n = \Delta/2$.\footnote{This is a convolution of uniform distribution over intervals $[-a_1,a_1],[-a_1,a_1],\cdots, [-a_n,a_n],[-a_n,a_n],\cdots $ in the frequency domain.} It holds that
\begin{align}\label{eq:w_t}
    0 \le w_{\Delta}(t) \le 2(e\Delta)^2 t \exp\L(-\frac{2}{7}\frac{\Delta t}{\ln(\Delta t)^2}\R)\quad \text{where}\quad t \ge \frac{e^{1/\sqrt{2}}}{\Delta}.
\end{align}

For our purposes, we simply identify
\begin{align}
    \vH(s)\leftarrow -\vec{\CH}_{s\beta} \quad \text{and}\quad \Delta \leftarrow \min_{s}\lambda_{gap}( \vec{\CH}_{s\beta})
\end{align}
and note that at $\beta =0,$ the purified Gibbs state is simply the Bell state across two copies. We need to flip the sign since the parent Hamiltonian is $\CH$ (and the Lindbladian spectrum) is always nonpositive. The discriminant gap $\Delta$ is expected to be independent of the system size when the Gibbs state has decay of correlation (under the normalization $\norm{\vA^a} =1$), but may diverge with the system size near phase transitions. 

    Algorithmically, the above should be considered an alternative to modern adiabatic algorithms based on amplitude amplification~\cite[Appedix G]{chen2023QThermalStatePrep}. The quasi-adiabatic evolution is particularly powerful in the case of geometrically local Hamiltonians, as it manifestly gives a quasi-local unitary circuit that may be significantly more efficient than general-purpose global adiabatic algorithms. We claim that, the unitary $\vV(1)$ be implemented using
    \begin{align}
        n \cdot \frac{\poly_{D}(\beta+1)}{\poly(\Delta)} \cdot  \poly\log(n/\epsilon)\quad \text{2-qubit gates and ancillas}
    \end{align}
    up to $\epsilon$-error in the operator norm by the HHKL scheme~\cite{haah2018QAlgSimLatticeHam}. In the rest of the section, we instantiate the main algorithmic components (\autoref{prop:high_precision_quasi},\autoref{lem:beta_ctrl_parent_Ham},\autoref{lem:beta_ctrl_parent_Ham_derivative}) and defer the error and locality analysis to~\cite{bergamaschi2025quantum}.

\subsection{Simulating the quasi-adiabatic operator using time-dependent Hamiltonian simulation}

The main idea of the HHKL scheme~\cite{haah2018QAlgSimLatticeHam} is to glue together individual polylogarithmically-sized patches, which are each simulated to a high precision.
Here, we estimate the algorithmic cost for the high-precision time-dependent Hamiltonian simulation subroutine. Consider a discretization $\bar{s}\in S_{s_0}$ of $[0,1]$ with evenly spaced points.

\begin{prop}[High precision simulation of quasi-adiabatic evolution]
\label{prop:high_precision_quasi}
Consider the one-parameter family of unitaries generated by the time-dependent Hermitian matrix $\vW(s)$ as in~\eqref{eq:W}, defined by the parent Hamiltonian
\begin{align}
     \vV'(s) = \ri \vW(s)\vV(s) \quad \text{for each}\quad s \in [0,1], \quad \text{and}\quad \vV(0)=\vI
\end{align}
with parameter $\Delta.$ Suppose that the Hamiltonian and its derivatives satisfy
\begin{align}
    \sup_{s\in [0,1]}\norm{\vH(s)} \le \alpha, \quad \sup_{s\in [0,1]}\norm{\vH'(s)} \le \alpha_1\quad \text{and}\quad \sup_{s\in [0,1]}\norm{\vH''(s)} \le \alpha_2.
\end{align}
Then, $\vV(s)$ can be simulated to $\epsilon$-spectral norm using 
\begin{align}
&\tCO\L((\frac{\alpha}{\Delta}+1)(\frac{\alpha_1}{\Delta}+1)\R)\quad &\text{block-encodings of}\quad \sum_{\bar{s} \in S_{s_0}} \frac{\vH(\bar{s})}{\alpha}  \otimes \ket{\bar{s}}\bra{\bar{s}}\\
&\tCO(\frac{\alpha_1}{\Delta}+1)\quad &\text{block-encodings of} \sum_{\bar{s} \in S_{s_0}} \frac{1}{\alpha_1}\vH'(\bar{s}) \otimes \ketbra{\bar{s}}{\bar{s}},\\
&\tCO(\frac{\alpha_1}{\Delta}+1) &\text{additional two-qubit gates and additional ancillas}.
\end{align}
with discretization parameter $\log(S_{s_0}) =\tCO(1).$ Here, $\tCO(1)$ absorbs polylogarithmic factors of $1/\epsilon, \alpha_1,\alpha_2,\Delta.$ 
\end{prop}

Note that the dependence on $1/\Delta$ is quadratic: one from the width of the filter, the other from the normalization $\int\labs{w(t)t}\rd t \propto 1/\Delta.$
This is simply an application of the high-precision time-dependent Hamiltonian simulation~\cite[Corollary 4]{low2018hamiltonianInt} for the quasi-adiabatic operator $\vW(s).$ Here, $\alpha,$ $\alpha'$ both scale with the number of jumps, so we cannot directly simulate the full system; the above is intended to be invoked, according to~\cite{haah2018QAlgSimLatticeHam}, for suitable subsets of the jumps and the neighbouring Hamiltonian supported on polylogarithmic patches. 

To construct $\vW(s),$ the main calculation is an explicit implementation of the LCU weight for $w(t).$
Consider
\begin{align}
v(t):=\sqrt{c_{\Delta}} \prod_{n=1}^{\infty} \frac{\sin(a_n t)}{a_n t} \quad \text{such that}\quad v(t)^2 = w(t).
\end{align}
Since this function in the frequency domain has compact support in $[-\frac{\Delta}{2},\frac{\Delta}{2}],$ we can consider the Chebyshev series.

\begin{lem}
    The Chebyshev series for the Fourier weight $\hat{v}(\nu),$ supported on $[-\Delta/2,\Delta/2]$,
\begin{align}
    \hat{v}_{\Delta}(\nu) = \sqrt{\Delta} \sum_{j=1}^{\infty} a_jT_j(2\nu/\Delta )\quad \text{satisifies that}\quad \labs{a_j} \le c_2 e^{ - c_1 p/\ln^2(p)}
\end{align}
for some absolute constants $c_1,c_2>0$.
\end{lem}
\begin{proof}
Extract the Chebyshev coefficients by the orthogonality of Chebyshev polynomials
\begin{align}
    \sqrt{\Delta}a_p &= \int_0^{\pi} T_p(\cos(\theta)) \hat{v}_{\Delta}(\frac{\Delta}{2}\cos(\theta)) \rd \theta\\
    &= \int_{0}^{\pi} \cos(p\theta) \hat{v}_{\Delta}(\frac{\Delta}{2}\cos(\theta)) \rd \theta \\
    &= \int_{0}^{\pi} \cos(p\theta) \frac{1}{\sqrt{2\pi}}\int_{-\infty}^{\infty} v_{\Delta}(t)e^{-i t\frac{\Delta}{2} \cos(\theta)} \rd t\rd \theta\\
    &= \frac{\pi i^{p}}{\sqrt{2\pi}}\int_{-\infty}^{\infty} v_{\Delta}(t)  J_n(-\Delta t/2) \rd t \tag*{(By definition of Bessel functions )}\\
    &= \frac{\pi i^{p}}{\sqrt{2\pi}}\int_{-\infty}^{\infty} \frac{2v_{\Delta}(2t/\Delta )}{\Delta}  J_p(-t) \rd t \tag*{(Change of integration variable)}.
\end{align} 

We evaluate the integral using the piecewise bound that $J_{\alpha}(t)\le \frac{\labs{t/2}^{\alpha}}{\Gamma(\alpha+1)} $ for real $t$,\footnote{https://dlmf.nist.gov/10.9} $\alpha > -1/2,$ and $\labs{J_{\alpha}(t)}\le 1$ for real $t$, $\alpha \ge 0$ \footnote{https://dlmf.nist.gov/10.14}
\begin{align}
    \labs{a_p}&\le \frac{c}{\sqrt{\Delta}}\int_{\labs{t}\le T} \labs{v_{\Delta}(2t/\Delta)} \frac{\labs{t/2}^{p}}{p!} \frac{2\rd t}{\Delta} + \frac{c}{\sqrt{\Delta}}\int_{\labs{t}> T} \labs{v_{\Delta}(2t/\Delta)} \frac{2\rd t}{\Delta}\\
    &\le c 2^{-p} \undersetbrace{\le const.}{\int_{-\infty}^{\infty} \labs{v_{\Delta}(2t/\Delta)} \frac{2\rd t}{\Delta\sqrt{\Delta}}} + c \int_{\labs{t}> T} \labs{v_{\Delta}(2t/\Delta)} \frac{2\rd t}{\Delta \sqrt{\Delta}} \tag*{(Choosing that $T =\frac{p}{e}.$)}\\
    &\le c 2^{-p}+ c e^{ - c_1 p/\ln^2(p)} \le c_2 e^{ - c_1 p/\ln^2(p)},
\end{align}
as advertised.
\end{proof}

\begin{lem}[State preparation (frequency domain)]\label{lem:stateprep_v}
For any discretization of an interval $\{\bar{\nu}\} \subset I \subset [\Delta/2,\Delta/2]$ into $2^{n}$ many points with spacing $\omega_0$, there exists an absolute constant $c$  and an $\epsilon$ approximate block encoding such that 
\begin{align}
\frac{c}{\sqrt{\Delta}} \sum_{\bar \nu} \hat{v}(\bar{\nu}) \ket{\bar{\nu}}\bra{\bar{\nu}}
\end{align}
at gate complexity $\CO( n\log \frac{1}{\epsilon} (\log \log \frac{1}{\epsilon})^3).$ 
\end{lem}
\begin{proof}
    Truncate the Chebyshev series at order $p = \Theta(\log \frac{1}{\epsilon} \cdot (\log \log \frac{1}{\epsilon})^3 )$ and apply QSVT to the operator $ \frac{2}{\Delta}\sum_{\bar \nu = -\Delta/2}^{\Delta/2} \bar{\nu} \ket{\bar{\nu}}\bra{\bar{\nu}}.$ 
\end{proof}

\begin{lem}[State preparation (time domain)]\label{lem:stateprep_t}
For the $\bt \in S_{t_0},$ with $S_{t_0}= N = 2^n$, suppose that $2\pi/t_0 \ge \Delta.$ Then, there exists a state preparation circuit for an $\CO((\Delta N t_0)^2 e^{-\frac{2\Delta N t_0}{ 7\ln^2(2\Delta N t_0)}})$-approximation for
\begin{align}
\frac{\sum_{\bt \in S_{t_0}} v(\bt) \ket{\bt}}{\norm{\sum_{\bt \in S_{t_0}} v(\bt) \ket{\bt}}_2}
\end{align}
at gate complexity $\CO( \log(\Delta Nt_0)\log \frac{1}{\epsilon} (\log \log \frac{1}{\epsilon})^3 + \log(N)^2).$ 
\end{lem}
\begin{proof}
First, apply $\frac{c}{\sqrt{\Delta}} \sum_{\bar \nu = \omega_0(-N'/2 -1) }^{\omega_0N'/2} \hat{v}(\bar{\nu}) \ket{\bar{\nu}}\bra{\bar{\nu}}$ (\autoref{lem:stateprep_v}) to the uniform superposition $\frac{1}{\sqrt{N'}}\sum_{\bar \nu = \omega_0(-N'/2 -1) }^{\omega_0N'/2} \ket{\bar{\nu}}$ on $N'=2^{n'} \le N$ points such that $\omega_02^{n'} > \Delta,$ which gives $\propto \sum_{\nu \in S_{\omega_0}} \hat{v} (\bnu) \ket{\nu}.$ Next, consider a Quantum Fourier transform on $N$ dimensions with spacing 
$\omega_0, t_0$ such that $\omega_0t_0 = \frac{2\pi}{N}$~\cite{chen2023QThermalStatePrep}. Then,
\begin{align}
\frac{\sqrt{2\pi}}{\sqrt{t_0N}}(QFT)^{-1} \sum_{\bnu \in S_{\omega_0}} \hat{v} (\bnu) \ket{\bnu}&=\frac{\sqrt{2\pi}}{\sqrt{t_0N}}(QFT)^{-1} \sum_k \sum_{\bnu \in S_{\omega_0}} \hat{v} (\bnu + kN\omega_0) \ket{\bnu} \tag*{(Since $N\omega_0 = \frac{2\pi}{t_0}\ge \Delta$)}\\
&= \sqrt{t_0}\sum_{k=-\infty}^{\infty} \sum_{\bt \in S_{t_0}} (v (\bt+kNt_0))\ket{\bt} \tag*{(By~\cite[Fact A.1]{chen2023QThermalStatePrep})}\\
&\approx \sqrt{t_0}\sum_{\bt \in S_{t_0}} v (\bt)\ket{\bt}
\end{align}
the last approximation is bounded by $\CO((\Delta N t_0)^2 e^{-\frac{2\Delta N t_0}{ 7\ln^2(2\Delta N t_0)}})$, due to the decay of $v(t)$ at $t\ge N t_0/2.$ Lastly, tune a suitable subnormalization constant and apply constant rounds of exact amplitude amplification~\cite[Theorem 3]{mcArdle2022QStatePreparationWOArithm}. The gate complexity is due to~\autoref{lem:stateprep_v} and QFT.
\end{proof}

\begin{cor}[Block-encoded, controlled Quasi-adiabatic evolution]\label{cor:beta_ctrl_Qausi}
In the setting of~\autoref{prop:high_precision_quasi}, for any discretization $\bar{s},$ there exists an $\epsilon$-approximate block encoding for the controlled Quasi-adiabatic evolution 
\begin{align}
    \sum_{s \in } \frac{\vW(\bar{s})}{c_1\alpha_1} \otimes \ketbra{\bar{s}}{\bar{s}} 
\end{align}
for an absolute constant $c_1$, using 
\begin{align}
    &\tCO(\frac{\alpha}{\Delta}+1)\quad &\text{Block-encodings of}\quad \frac{\vH(\bar{s})}{\alpha}  \otimes \ket{\bar{s}}\bra{\bar{s}},\\
    &\tCO(1)\quad &\text{Block-encodings of}\quad \frac{\vH'(\bar{s})}{\alpha_1} \otimes \ket{\bar{s}}\bra{\bar{s}},\\
    &\tCO(1) &\text{additional two-qubit gates and additional ancillas}.
\end{align}
\end{cor}
Here, the additional ancillas exclude those involved with the block-encodings for $\vH$, $\vH'.$
\begin{proof}
Our analysis will be pointwise for each $\bar{s}.$ To implement the double integral, we write that
\begin{align}
\int_{-\infty}^{\infty}\rd t\ w(t) \int_0^t du e^{\ri u \vH(s)}\vH^{'a}(s)\e^{-\ri u \vH(s)} \approx T\int_{-T}^{T}\rd t\ w(t) \frac{1}{T}\int_0^t du e^{\ri u \vH(s)}\vH^{'a}(s) e^{-\ri u \vH(s)},    
\end{align}
and implement the discrete approximation $\frac{1}{T} \int_0^{\bar{t}} du e^{\ri u \vH(s)} \vH^{'a}(s)e^{-\ri u \vH(s)}\otimes \ket{\bar{t}}\bra{\bar{t}}$ by another LCU and combine with the state preparation (\autoref{lem:stateprep_t}).
\end{proof}

\subsubsection{Proof of \autoref{prop:high_precision_quasi}}
We quickly checked that this fits into the conditions for~\cite[Theorem 3, Corollary 4]{low2018hamiltonianInt}. 
\begin{proof}
We compute
    \begin{align}
        b&= \sup_{s}\norm{\vW(s)} \le \sup_{s} \norm{\vH'(s)} \int_{-\infty}^{\infty} \labs{w(t)t} \rd t  \le \frac{c_1\alpha_1}{\Delta}\\
        b_1&= \int_0^{1} \norm{\vW'(s)}\rd s \\ 
        &\le \sup_{s}\norm{\vW'(s)} \le \sup_{s}\norm{\vH''(s)}\int_{-\infty}^{\infty} \labs{w(t)t} \rd t  + \sup_{s}\norm{[\vH(s),\vH'(s)]}\int_{-\infty}^{\infty} \labs{w(t)t^2} \rd t \le \frac{c_1\alpha_1}{\Delta} + \frac{c_2\alpha\alpha_2}{\Delta^2},
    \end{align}
    and suppose the block encoding for $\sum_{\bar{s} \in S_{s_0}} \vW(\bar{s})  \otimes \ket{\bar{s}}\bra{\bar{s}}$ uses $n_a$-many ancillas. Then, by~\cite[Theorem 3, Corollary 4]{low2018hamiltonianInt}, we can choose the number of discretization points $M = \labs{S_{s_0}} = \Theta(\frac{b_1+ b^2}{b\epsilon})$ such that $\vV(s)$ can be simulated using $\CO( b \log(b/\epsilon)/\log\log(b/\epsilon)+ \log(1/\epsilon)/\log\log(1/\epsilon))$ queries to $\sum_{\bar{s} \in S_{s_0}} \vW(\bar{s})  \otimes \ket{\bar{s}}\bra{\bar{s}}$, $\CO( n_a + \log(M))$ ancillas, and $\CO(b (n_a + \log(M)))$ primitive gates. Instantiate block-encodings for $\vW$ (\autoref{cor:beta_ctrl_Qausi}) using that of $\vH,\vH'$ to conclude the proof.
\end{proof}

\subsection{Additional gates for $\beta$-control gates}
The quantum algorithm subroutine for HHKL uses high-precision time-dependent Hamiltonian simulation. The main input to this framework is to have an additional control register for the cooling schedule.

Consider a discretized, linear annealing schedule 
\begin{align}
    S_{\beta_0} = \{0, \beta_0, \cdots ,\beta \}\label{eq:linear-schedule}
\end{align}
and the associated set of uncertainties 
\begin{align}
    \{\sigma(0),\sigma(\beta_0),\cdots \sigma(\beta)\}. 
\end{align}
as long as $\sigma \le \frac{1}{\beta}.$ For simplicity, we take $\sigma = \frac{1}{\beta},$ and in that case, the Gaussian filters $f_{\sigma(\beta)}\equiv f_{\beta}$, and transition weights $\gamma_{\sigma,\beta}\equiv \gamma_{\beta}$ are functions of the temperature. (More generally, if one wishes to independently tune $\sigma < 1/\beta,$ one would need a controlled version of the state-preparation unitaries for the double time-integrals~\autoref{cor:b_2dFT}, \autoref{cor:N_timedomain_explicit}.)

The easiest extension to the algorithmic framework here is to use temperature-time controlled Hamiltonian simulation
    \begin{align}
	\sum_{\bbeta \in S_{\beta_0}} \sum_{\bt \in S_{t_0}} \e^{\pm \ri \bt \bbeta \vH} \otimes \ketbra{\bt}{\bt} \otimes \ketbra{\bbeta}{\bbeta}
\end{align}
which can be obtained by adding $\bbeta$-controls on the synthesis of the time-controlled Hamiltonian simulation. Indeed, all the integral expressions can be written so that the only $\beta$ dependence is in the $\beta$-rescaling of the Heisenberg dynamics.

\begin{lem}[Block-encoded {$\beta$}-controlled parent Hamiltonian]\label{lem:beta_ctrl_parent_Ham}
Consider the linear temperature schedule~\eqref{eq:linear-schedule} for $0\rightarrow\beta.$ Then, in the setting of ~\autoref{thm:D_cost}, the gate complexity extends to the block-encoding for the $\beta$-controlled parent Hamiltonian 
\begin{align}
    \sum_{\bbeta \in S_{\beta_0}} \frac{1}{2}\vec{\CH}_{\bbeta} \otimes \ketbra{\bbeta}{\bbeta},
\end{align}
by extending all uses of controlled Hamiltonian simulation with $\beta$-controls. 
\end{lem}

To implement the quasi-adiabatic evolution, we also need access to the derivatives.
\begin{lem}[Block-encoded {$\beta$}-controlled parent Hamiltonian]\label{lem:beta_ctrl_parent_Ham_derivative}
Consider the linear temperature schedule~\eqref{eq:linear-schedule} for $0\rightarrow\beta,$ and suppose $\norm{\sum_{a\in A} [\vH,\vA^a]^{\dagger}[\vH,\vA^a]}^{1/2} \le \alpha.$ Then, there is an absolute constant $c$ such that block-encoding for the $\beta$-controlled parent Hamiltonian 
\begin{align}
\frac{1}{c\alpha}\sum_{\bbeta \in S_{\beta_0}} \vec{\CH}'_{\bbeta} \otimes \ketbra{\bbeta}{\bbeta}.
\end{align}
can be implemented using the gate complexity ~\autoref{thm:D_cost}, and $\tCO(1)$ additional block encoding for $\frac{1}{\alpha}\sum_{a\in A} [\vH,\vA^a] \otimes \ket{a}$ and $\frac{1}{\alpha} \sum_{a\in A} [\vH,\vA^a]^{T} \otimes \ket{a}.$ 
\end{lem}
In both \autoref{lem:beta_ctrl_parent_Ham} and \autoref{lem:beta_ctrl_parent_Ham_derivative}, we assumed the main text normalization $\norm{\sum_{a\in A} \vA^{\dagger a}\vA^{a}} \le 1.$
 
\subsection{Properties of the derivatives in the Gaussian case}
Recall, the parent Hamiltonian generally takes the following form
\begin{align}
    \vec{\CH} = \vec{\CS} + \frac{1}{2}(\vN\otimes \vI + \vI \otimes \vN^*).
\end{align}
The $\beta$-derivatives of the parent Hamiltonian can be inferred from the time-domain expressions~\autoref{cor:b_2dFT},\autoref{cor:N_timedomain_explicit}. For illustration, we explicitly calculate the case of Gaussian transition weights with $\sigma = \sigma_{\gamma} = 1/\beta$ (see also~\cite{rouze2025efficient}). There, the expressions are involved with double integrals, and the $\beta$-dependences are isolated in the rescaling of the Heisenberg dynamics
\begin{align}
    \vec{\CS} &= \sum_{a \in A} \iint_{-\infty}^{\infty} h_1(t)h_2(t') \vA^a\L(\beta(t'-t)\R)\otimes\vA^{a}\L(-\beta(t'+t)\R)^T \rd t'\rd t \\
    \vN &:=\sum_{a\in A}\iint_{-\infty}^{\infty}n_1(t)n_2(t')\vA^{a\dagger}(\beta (t'-t))\vA^a(-\beta (t'+t))\rd t' \rd t,
\end{align}
where
\begin{align}
h_1(t)&= \frac{2}{\pi} e^{ -2t^2}, \quad h_2(t)= e^{-1/4}\cdot  e^{-4t^2}\\ 
     n_1(t) &:=\frac{2\sqrt{\pi}}{3}\cdot  \L(\!\frac{1}{\cosh\L(2\pi t\R)\!}*_t  e^{-2 t^2}\!\R), \quad
 n_2(t) := \frac{3}{\pi\sqrt{\pi}}\exp\L(-4 t^2-2\ri t \R),
\end{align}

\begin{prop}[First and second derivatives] For the Gaussian transition weights, there exist absolute constants $c_1,c_2 >0$ such that
\begin{align}
    \lnorm{\frac{\rd}{\rd \beta}\vec{\CH}_{\beta}} &\le c_1 \lnorm{\sum_{a \in A}[\vH,\vA^a][\vH,\vA^a]^{\dagger}}^{1/2} \lnorm{\sum_{a \in A} \vA^{a\dagger}\vA^{a}}^{1/2}.\\
    \lnorm{\frac{\rd^2}{\rd \beta^2}\vec{\CH}_{\beta}} &\le c_2 \lnorm{\sum_{a \in A}[\vH,[\vH,\vA^a][\vH,[\vH,\vA^a]^{\dagger}}^{1/2} \lnorm{\sum_{a \in A} \vA^{a\dagger}\vA^{a}}^{1/2} + c_2 \lnorm{\sum_{a \in A}[\vH,\vA^a][\vH,\vA^a]^{\dagger}}.
\end{align}
\end{prop}

\begin{proof}
    Take derivative
\begin{align}
    \frac{\rd}{\rd \beta} \CS 
    & = \sum_{a \in A} i\iint_{-\infty}^{\infty}h_1(t)h_2(t') \\
    &\L( (t'-t)[\vH,\vA^a\L(\beta(t'-t)\R)]\otimes\vA^{a}\L(-\beta(t+t')\R)^T - (t+t')\vA^a\L(\beta(t'-t)\R)\otimes[\vH,\vA^{a}\L(-\beta(t+t')\R)]^T \R)\rd t'\rd t.
\end{align}
Introduce purifying operators  to absorb the sum over jumps
\begin{align}
    \sum_{a\in A} [\vH,\vA^a\L(\beta(t'-t)\R)] \otimes \bra{a},\quad \sum_{a\in A} \vA^a\L(\beta(t'-t)\R)^{T} \otimes \ket{a}, \\
    \quad \sum_{a\in A} \vA^a\L(\beta(t'-t)\R)\otimes \bra{a}, \quad \sum_{a\in A} [\vH,\vA^a\L(\beta(t'-t)\R)]^{T} \otimes \ket{a},
\end{align}
and take the operator norm obtain
\begin{align}
    \norm{\frac{\rd}{\rd \beta} \CS} &\le  \norm{\sum_{a \in A}[\vH,\vA^a][\vH,\vA^a]^{\dagger}}^{1/2} \norm{\sum_{a \in A} \vA^{a}\vA^{a\dagger}}^{1/2}\iint_{-\infty}^{\infty}2(\labs{t'}+\labs{t})\labs{h_1(t)h_2(t')}  \rd t'\rd t.
\end{align}
Similarly, 
\begin{align}
		\norm{\frac{\rd}{\rd \beta}\vN} \le \norm{\sum_{a \in A}[\vH,\vA^a][\vH,\vA^a]^{\dagger}}^{1/2} \norm{\sum_{a \in A} \vA^{a}\vA^{a\dagger}}^{1/2} \iint_{-\infty}^{\infty}2(\labs{t'}+\labs{t})\labs{n_1(t)n_2(t')}  \rd t\rd t' 
\end{align} 
Note that $\sum_{a \in A} \vA^{a}\vA^{a\dagger} = \sum_{a \in A} \vA^{a\dagger}\vA^{a}$, and $\sum_{a \in A}[\vH,\vA^a][\vH,\vA^a]^{\dagger} = \sum_{a \in A}[\vH,\vA^a]^{\dagger}[\vH,\vA^a]$ since the set of jumps contains their adjoints.
The calculation for the second derivatives is routine, with different integration constants.
\end{proof}

\subsection{Comparing with the Lindbladian case}
Of course, we may apply the same line of thought to Lindbladians, assuming the quasi-local Lindbladian mixes in logarithmic time. However, there are subtle differences. The spectral gap of the parent Hamiltonian is qualitatively the same as the mixing time of our Lindbladian (which has a vanishing anti-Hermitian part)
\begin{align}
    \frac{\ln(2)}{\lambda_{gap}(\CH)} \le t_{mix}(\CL) \le \frac{\ln(2\norm{\vrho^{-1/2}})}{\lambda_{gap}(\CH)}.
\end{align}
The conversion overhead $\ln(2\norm{\vrho^{-1/2}})= \CO(\beta \norm{\vH})$ could scale with the system size in general. However, to drop this $\CO(\beta \norm{\vH})$ factor (i.e., proving \textit{rapid mixing} such that the mixing time scales logarithmically with the system size) amounts to proving a more stringent quantum log-Sobolev inequality, which has been highly nontrivial~\cite{capel2021modified}. Even if we assume rapid mixing, to get a low-depth circuit for the Gibbs state, we still need an efficient algorithm to implement the Lindbladian evolution that parallelizes the Lindblad operators. However, the HHKL algorithm~\cite{haah2018QAlgSimLatticeHam}, achieving this parallelization for local Hamiltonians, made critical use of \textit{reversing} time evolutions, which does not obviously apply to the dissipative setting. 

As a remedy, one may consider trotterizing into quasi-local brickwork nonunitary circuits, giving a \textit{discrete-time} Quantum Markov chain\footnote{The observation that discretization can be helpful despite deviating from the continuum was made in~\cite{ding2023single}.}
\begin{align}
    \e^{\CL \theta} \rightarrow \CN := \prod_{g}\e^{\CL_g \theta} \quad \text{where}\quad \CL = \sum_{g} \CL_g \quad \text{such that each} \quad \CL_g\quad \text{is a sum of quasi-local, nearly commuting block}.
\end{align}
That is, we regroup the \Lword{} $ \CL = \sum_{a\in A} \CL^a$ into quasi-local blocks to exploit the fact that nearly disjoint Lindbladians can be efficiently implemented in parallel (after spatial truncation).
However, the Trotter error from first-order product formulas is extensive; thus, the discrete-time channel would not approximate the continuous-time evolution in general, and the mixing time analysis might require quantitatively different formalism. Nevertheless, the Gibbs state remains stationary 
\begin{align}
    \CN[\vrho_{\beta}] = \vrho_{\beta},
\end{align}
and could potentially share similar mixing behavior as the continuous case in practice.

\section{Why Gaussians?}
\label{sec:why_gaussian}
Our direct calculation confirms the correctness of our ansatz. But why does the filter $f(t)$ need to be Gaussian~\eqref{eq:OpOFT}? In this section, we try to derive the Gaussians from scratch, which can be viewed as an alternative view of detailed balance in the time domain. While we try to make our arguments precise, we stop at the physicists' level of rigor and do not attempt to extract a mathematical theorem. 

Assume that the filter function is real on the real axis
\begin{align}
f(t) = f(t)^* \quad \text{for each}\quad t \in \BR  
\end{align}
and complex analytic, independent of $\beta$, and that the jump operator $\vA = \vA^{\dagger}$ is Hermitian (Let us focus on a single jump and drop the jump labels $a\in A$.). The goal is to solve for the viable choices of $f(t)$.

We calculate the associated ``transition'' part~\eqref{eq:exact_DB_L} in the time domain
\begin{eqnarray}
    \CT^{\dagger}[\cdot]&=& \int_{-\infty}^\infty \gamma(\omega) \hat{\vA}(\omega)^\dag(\cdot)\hat{\vA}(\omega)\rd\omega\\
    &=& \int  \gamma(\omega) \e^{\ri\omega(t-t')}f(t)f(t')\vA(t)(\cdot)\vA(t')\rd\omega\rd t\rd t'\quad \text{(using} \quad f(t) = f(t)^*)\\
    &=& \frac{1}{\sqrt{2\pi}}\int   c(t-t') f(t)f(t')\vA(t)(\cdot)\vA(t') \rd t\rd t'.
\end{eqnarray}
The third line uses the inverse Fourier Transform
\begin{align}
    c(t) = \frac{1}{\sqrt{2\pi}}\int_{-\infty}^{\infty}\e^{\ri \omega t} \gamma(\omega) \rd\omega .
\end{align}
We may interpret $c(t)$ as a certain correlation function (hence the notation). Now, note that conjugating the Gibbs state yields
\begin{align}
    \Lambda(\hat{\vA}(\omega))&=\int_{-\infty}^{\infty} \e^{-\ri \omega t}f(t) \e^{\ri\vH(t-\ri\beta/2)}\vA \e^{-\ri \vH(t-  \ri\beta/2)} \rd t\\
    &=\int \e^{-\ri s\omega }\e^{\beta \omega/2} \vA(s) f(s+\ri\beta/2) \rd s\quad &\text{(setting}\quad s:=t-\ri\beta/2).
\end{align}
Then, we get that 
\begin{align}
    &\Gamma^{-1}\circ \CT \circ \Gamma[\cdot]\\
    &= \int_{-\infty}^\infty \gamma(\omega) \Lambda(\hat{\vA}(\omega))(\cdot)\Lambda^{-1}(\hat{\vA}(\omega)^\dag)\rd\omega\\
    &= \int_{-\infty}^{\infty}\int\int \gamma(\omega) \e^{-\ri\omega(s-s'+\ri\beta)}f(s+\ri\beta/2)f(s'-\ri\beta/2)\vA(s)(\cdot)\vA(s')\rd s \rd s'\rd\omega\quad &\text{(setting}\quad s':= t' + \ri \beta)\\
    &= \frac{1}{\sqrt{2\pi}} \int\int c(s'-s-\ri\beta) f(s+\ri\beta/2)f(s'-\ri\beta/2)\vA(s)(\cdot)\vA(s')\rd s \rd s'\quad &\text{(continuing}\quad c(z),\  z \in \BC) \\
    &= \frac{1}{\sqrt{2\pi}} \int \int  c(t'-t-\ri\beta) f(t+\ri\beta/2)f(t'-\ri\beta/2)\vA(t)(\cdot)\vA(t') \rd t\rd t' \quad &\text{(shifting integration)}.
\end{align} 
The third line uses that $(\vA(z))^{\dagger} = \vA(z^*)$ for Heisenberg evolution at complex times. The last line shifts the integrals by $s \rightarrow t$ and $s' \rightarrow t'$ (assuming the absence of poles across the strip).

Comparing the coefficient of product integrals over $\vA(t)\cdot \vA(t')$ for each $t, t'$, the condition 
\begin{equation}\label{eq:constaint_gamma_f}
c(t-t') f(t)f(t')= c(t'-t-\ri\beta) f(t+\ri\beta/2)f(t'-\ri\beta/2)\quad\text{for each}\quad t, t'\in \BR \quad \text{ensures}\quad 
     \CT^{\dagger}=\Gamma^{-1}\circ \CT \circ \Gamma.
\end{equation}
 This condition can be rearranged as (whenever the denominators are non-zero)
\begin{align}\label{eq:constaint_gamma_f_v2}
    \frac{c(t-t')}{c(t'-t-\ri\beta)} =  \frac{f(t+\ri\beta/2)f(t'-\ri\beta/2)}{f(t)f(t')}\quad \text{for each}\quad t, t'\in \BR.
\end{align}
\subsection{Solving the functional equation}
We proceed to solve~\eqref{eq:constaint_gamma_f_v2} for a fixed $\beta$. Consider a change of variable $t\leftrightarrow t'$
\begin{align}
        \frac{c(t'-t)}{c(t-t'-\ri\beta)} =  \frac{f(t'+\ri\beta/2)f(t-\ri\beta/2)}{f(t')f(t)}\label{eq:alternative_tt'}.
\end{align}
Then, divide the two equations and observe that the LHS depends only on the time difference $t-t'$
\begin{align}
\eqref{eq:constaint_gamma_f_v2}/\eqref{eq:alternative_tt'} :  f_1(t-t') = \frac{f(t+ i\beta/2)}{f(t-\ri\beta/2)} \cdot \frac{f(t'-\ri\beta/2)}{f(t'+\ri\beta/2)} : = g(t)/g(t').
\end{align}
We must have that 
\begin{align}
    g(t) = a \e^{bt} \quad \text{for constants}\quad a, b \quad \text{independent of}\quad t.
\end{align}
Thus, 
\begin{align}
    a\e^{bt} f(t-\ri\beta/2) &= f(t+\ri\beta/2) \\
    \implies\quad a\e^{bt} \e^{-i \beta/2 \cdot \partial_t } f(t) &= \e^{ i \beta/2 \cdot \partial_t} f(t) \quad &\text{(analyticity)}\\
    \implies \quad \e^{a'+b't-i \beta \cdot \partial_t } f(t) &= f(t) \quad &\text{(since}\quad [t, \partial_t] \propto 1).
\end{align}
The last equality regards $t$ and $\partial_t$ as linear operators acting on functions and uses the matrix-exponential fact that
\begin{align}
    [\vX,\vY] \quad \text{commutes with}\quad \vX, \vY \quad \text{implies}\quad \e^{\vX}\e^{\vY} =\e^{\vX+\vY+[\vX,\vY]/2}. 
\end{align}
We can get rid of the exponential by regarding it as an eigenproblem
\begin{align}
    \e^{\vX}\ket{v} = \ket{v} \quad \text{implies}\quad  \vX \ket{v} = \ri 2\pi \BZ \cdot \ket{v}.
\end{align}
This amounts to solving the differential equation
\begin{align}
    (a_1 \partial_t + a_2 t +a_3 ) f(t) = 0 
\end{align}
which has the general solution being Gaussians (the constants may depend on the fixed $\beta$.). Plugging back into~\eqref{eq:constaint_gamma_f_v2}, we can solve for $c(t)$, which would not be unique as it allows for linear combinations.

As a sanity check, we re-derive the constraints between the Gaussian parameters and the temperature in the time picture. Let 
\begin{equation}
    c(x)=e^{ - i a x }\e^{-x^2/\delta^2}, ~~~~~~{\rm and}~~~~~~f(t)=e^{-t^2/\kappa^2},
\end{equation}
where $a,\delta,\kappa$ are all real parameters. Then, letting $t-t':=x$ we get:
\begin{eqnarray}
    \log\left[\frac{c(t-t')}{c(t'-t-\ri\beta)} \right]&=&  \log\left[\frac{f(t+\ri\beta/2)f(t'-\ri\beta/2)}{f(t)f(t')}\right]\\
    \implies \quad \log\left[\frac{e^{-iax }\e^{-x^2/\delta^2}}{e^{ia(x+\ri\beta) }\e^{-(x+\ri\beta)^2/\delta^2}} \right]&=&  \log\left[\frac{e^{-(t+\ri\beta/2)^2/\kappa^2 }\e^{-(t'-\ri\beta/2)^2/\kappa^2}}{e^{-t^2/\kappa^2-t'^2/\kappa^2}}\right],\\
    \text{it suffices if}\quad 2ix(\frac{\beta}{\delta^2} - a)+ a\beta-\frac{\beta}{\delta^2}&=& \frac{-ix\beta}{\kappa^2}+\frac{\beta^2}{2\kappa^2}.
\end{eqnarray}

Equating the imaginary and real parts leads to two linearly dependent equations with the same solution:
\begin{equation}
    a=\beta(\frac{1}{\delta^2}+\frac{1}{2\kappa^2}).
\end{equation}
Then identifying $a\equiv \omega_\gamma$, $\gamma^{-2}=\sigma^2_\gamma/2$, and $\kappa^{-2}=\sigma^2_E$ leads to the same relationship between $\omega_\gamma, \sigma^2_\gamma, \sigma^2_E$ and $\beta$ as in~\eqref{eq:ExctDissip}.
\section{Other notions of detailed balance}\label{sec:otherDB}
In addition to our KMS detailed balance condition (\autoref{defn:DB}), other quantum detailed balance conditions have also been studied (see, e.g.,~\cite{Alicki1976OnTD,fangola2007GeneratorsDetailedBal,carlen2017gradient}). In this section, we will discuss two variants of quantum detailed balance, and only the first one seems to work.
\subsection{Detailed balance with unitary drift}
We constructed a \Lword{} satisfying the KMS detailed balance condition. However, many of our results remain valid even if we add a suitable Hamiltonian term, which might enable alternative constructions.  
\begin{defn}[KMS-detailed balance with unitary drift~{\cite[Section 5]{fangola2007GeneratorsDetailedBal}}]\label{defn:unitary_DB}
	We say that the \Lword{} $\CL$ satisfies $\vrho$-detailed balance with unitary drift (in short $\vrho$-DBU) with respect to a full-rank state $\vrho$ if there exists a Hermitian operator $\vQ$ such that
	\begin{align}\label{eq:detailedBalUpToCoh}
		\CL^{\dagger}[\cdot] - \sqrt{\vrho}^{-1}\CL[ \sqrt{\vrho} \cdot \sqrt{\vrho}]\sqrt{\vrho}^{-1} = \ri [\vQ,\cdot].
	\end{align}
Or, in terms of discriminants,
\begin{align}
	\CA(\vrho,\CL) := 
    \frac{\CD(\vrho,\CL) - \CD(\vrho,\CL)^{\dagger}}{2} = -\ri \vrho^{1/4}[\vQ, \vrho^{-1/4}(\cdot)\vrho^{-1/4}]\vrho^{1/4}.
\end{align}
\end{defn}
That is, we relax the KMS detailed balance condition (\autoref{defn:DB}) by allowing the RHS to be any commutator term $\ri [\vQ,\cdot]$\footnote{In an email exchange with Jonathan Moussa, he mentioned an analogous unitary effect for discrete quantum channels.}; this enlarges the family of possible Lindbladians for the stationary state $\vrho$.
\begin{prop}[Fixed point]\label{prop:DBU_fixedpoint}
	If a \Lword{} $\CL$ is $\vrho$-detailed-balanced with unitary drift, then 
	\begin{align}
		\CL[\vrho] =0.
	\end{align}
\end{prop}
\begin{proof}
	Evaluate the superoperator~\eqref{eq:detailedBalUpToCoh} for the identity $\vI$  and conclude $\sqrt{\vrho}^{-1}\CL[ \sqrt{\vrho} \vI \sqrt{\vrho}]\sqrt{\vrho}^{-1}=\CL^{\dagger}[\vI]=0$.
\end{proof}
Intuitively, we can certainly add any Hermitian $\vQ$ that commutes with $\vrho$ without changing the stationary state. In fact, this is the only possibility\footnote{We thank Jonathan Moussa for pointing us to~\cite[Lemma 28]{fangola2007GeneratorsDetailedBal}. This further simplifies the derivation and the presentation, clarifying that our construction actually has $\vQ=0$ and a vanishing anti-Hermitian component.}.
\begin{prop}[Structure of $\vQ${~\cite[Lemma 28]{fangola2007GeneratorsDetailedBal}}]
    In the setting of~\autoref{defn:unitary_DB}, $\vQ$ must commute with $\vrho$.
\end{prop}

Indeed, we can solve for $\vB$ and $\vQ$ by modifying the argument for \autoref{cor:findingQ}
	\begin{align}
 \vQ &= \vQ^{\dagger}\\
 &\Updownarrow\\[-1mm]	
		(\Lambda-\Lambda^{-1})(\vB) &= \frac{\ri}{2}(2\CI -(\Lambda+\Lambda^{-1}))(\vR)\\[-0mm]
		&\Updownarrow\\[-1mm]		
		\sum_{\nu \in B}(\e^{\frac{\beta \nu}{2}}-\e^{\frac{-\beta \nu}{2}})\vB_{\nu}&=\frac{\ri}{2}\sum_{\nu \in B}(2-\e^{-\frac{\beta \nu}{2}}-\e^{\frac{\beta \nu}{2}})\vR_{\nu}\label{eq:tanhStep}\\
		&\Updownarrow\\[-1mm]	
		\sum_{\nu \in B\setminus\{0\}}\vB_{\nu}&=\frac{\ri}{2} \sum_{\nu \in B\setminus\{0\}} \tanh\left(\frac{\beta \nu}{4}\right) \vR_{\nu}.\tag*{($\e^{\frac{\beta \nu}{2}}-\e^{\frac{-\beta \nu}{2}}\ne 0 \iff \nu \ne 0$)}
	\end{align}
Thus, allowing the term $\vQ\ne 0$ merely amounts to dropping constraints on the $\vB_0$ component (which is exactly the set of operators that commute with the Hamiltonian).

This relaxed version of detailed balance comes with the conceptual price of making the \Lword{} ``non-self-adjoint'' under similarity transformation. Due to the anti-self-adjoint component $\CA(\vrho,\CL)^{\dagger}$, the spectral theory of convergence (at first glance) seems to break down since the right eigenvectors are not orthogonal in the Hilbert-Schmidt norm. Fortunately, the following observation comes to the rescue.
\begin{prop}\label{prop:stationarityHA}
If a \Lword{} $\CL$ is $\vrho$-detailed-balanced with unitary drift, then 
\begin{align}
	\CA(\vrho,\CL)[\sqrt{\vrho}] = \CH(\vrho,\CL)[\sqrt{\vrho}]=0.
\end{align}    
\end{prop}

In other words, the specific form of detailed balance implies that the anti-self-adjoint component preserves the eigenvector $\sqrt{\vrho}$; the particular eigenvector $\sqrt{\vrho}$, which we care about, is orthogonal to other eigenvectors. Therefore, the second eigenvalue of $\CH$ corresponds to the contraction of the Hilbert-Schmidt norm of the other eigenvector, controlling the mixing time.
\begin{restatable}[Mixing time from spectral gap (adapted from{~\cite[Proposition II.2]{chen2023QThermalStatePrep}})]{prop}{mixingtimegapDB}\label{prop:mixing_time_from_gap}
	If a \Lword{} $\CL$ satisfies $ \vrho$-DBU, then 
    \begin{align}
        t_{mix}(\CL) \le \frac{ \ln(2\norm{\vrho^{-1/2}})}{\lambda_{gap}(\CH)} \quad \text{where} \quad \CH:= \frac{\CD(\vrho,\CL)^{\dagger} + \CD(\vrho,\CL)}{2},
    \end{align}
    and the mixing time $t_{mix}$ is the smallest time for which
    \begin{align}
\lnormp{\e^{\CL t_{mix}}[\vrho_1-\vrho_2]}{1} \le \frac{1}{2} \normp{\vrho_1-\vrho_2}{1} \quad \text{for any states}\quad \vrho_1, \vrho_2.
\end{align}
\end{restatable}
In other words, adding a coherent term only seems to help with the convergence.
\begin{proof}
Write $\vR = \vrho_1-\vrho_2$, then
     \begin{align}
         \lnormp{ \e^{\CL t}[\vR]}{1} 
         &= \lnormp{\vrho^{1/4}\e^{\CD t}[\vrho^{-1/4}\vR\vrho^{-1/4}]\vrho^{1/4}}{1} 	\\
         &\le \lnormp{\vrho^{1/4}}{4}\cdot \lnormp{\e^{\CD t}[\vrho^{-1/4}\vR\vrho^{-1/4}]}{2}\cdot\lnormp{\vrho^{1/4}}{4}\\
         &\le \e^{-\lambda_{gap}(\CH) t} \lnormp{\vrho^{-1/4}\vR\vrho^{-1/4}}{2}\\
         &\le \e^{-\lambda_{gap}(\CH)t}\norm{\vrho^{-1/4}}^2 \normp{\vR}{2} \\
         &\le \e^{-\lambda_{gap}(\CH)t}\norm{\vrho^{-1/4}}^2 \normp{\vR}{1} \\
         &= \e^{-\lambda_{gap}(\CL^{\dagger})t}\norm{\vrho^{-1/2}} \normp{\vR}{1}. 
     \end{align}
     The first inequality uses Hölder's inequality. The second inequality uses the orthogonality to the leading eigenvector such that $\tr[\sqrt{\vrho}\cdot\vrho^{-1/4}\vR\vrho^{-1/4}]=\tr[\vR]=0$. Take the logarithm to conclude the proof.
\end{proof}

However, if we wish to obtain a parent Hamiltonian corresponding to this Lindbladian, the anti-Hermitian component forces us to consider the \textit{symmetrized discriminant}
\begin{align}
    \vec{\CH_{\beta}} = \frac{ \vec{\CD} + \vec{\CD^{\dagger}}}{2}.
\end{align}
Unfortunately, the presence of the anti-Hermitian component breaks the qualitative analogy between mixing times and discriminant gaps; fast mixing can hold despite a small symmetrized discriminant gap, marking a limitation of the parent Hamiltonian approach. Incorporating a coherent term, the general conversion bound is as follows~\cite{chen2023QThermalStatePrep}; unfortunately, some bounds are only meaningful if the anti-Hermitian part $\CA$ is small enough. 
\begin{prop}[Spectral gap from mixing time{~\cite[Proposition E.5]{chen2023QThermalStatePrep}}]\label{prop:mixing_to_gap}
    For any \Lword{} $\CL$, let $-\lambda_{\mathrm{Re}(gap)}(\CL)$ be the second largest real part in its spectrum (counted by algebraic multiplicity), then
    \begin{align}\label{eq:gapBound}
        \lambda_{gap}(\CH)+2\nrm{\CA}_{2-2}\ge\nrm{\CA}_{2-2}-\lambda_{2}(\CH)\ge\lambda_{\mathrm{Re}(gap)}(\CL)\ge \frac{\ln(2)}{t_{mix}(\CL)}.
    \end{align}
    Moreover, if $\lambda_{\mathrm{Re}(gap)}(\CL)\geq 2\nrm{\CA}_{2-2}$, then there is unique eigenvalue $\lambda_{1}(\CH)\geq -\nrm{\CA}_{2-2}$ and
    \begin{align}\label{eq:gapRelations}
        \lambda_{gap}(\CH)+2\nrm{\CA}_{2-2}\ge\nrm{\CA}_{2-2}-\lambda_{2}(\CH)\ge\lambda_{\mathrm{Re}(gap)}(\CL).
    \end{align}    
\end{prop}

\begin{restatable}[Spectral gap from mixing time]{cor}{spectralgapmixingDB}\label{corr:gapBound}
	If a \Lword{} $\CL$ satisfies $ \vrho$-DBU and $\frac{\ln(2)}{2 t_{mix}(\CL)}\geq\nrm{\CA}_{2-2}$, then\footnote{The bound cannot be qualitatively strengthened without additional structural understanding or assumptions. E.g., consider $\vH\!:=\vZ-(1+\epsilon)\vI$, $\vA\!:=\ri \vX$; the eigenvalues of $\vH+\vA$ are $-1-\epsilon$, while $\lambda_1(\vH)=-\epsilon$, $\nrm{\vA}=1$, and the relaxation (mixing) \nolinebreak time \nolinebreak is \nolinebreak $\bigO{1}$.

} 
    \begin{align}
        \lambda_{gap}(\CH)\ge \frac{\ln(2)}{2t_{mix}(\CL)}.
    \end{align}
\end{restatable}
\begin{proof}
    Observe that $\frac{\ln(2)}{2 t_{mix}(\CL)}\geq\nrm{\CA}_{2-2}$ implies $\lambda_{\mathrm{Re}(gap)}(\CL)\geq 2\nrm{\CA}_{2-2}$ due to \eqref{eq:gapBound}, which in turn by \eqref{eq:gapRelations} also implies that $\lambda_{1}(\CH)=0$ (since $0$ is an eigenvalue due to \autoref{prop:stationarityHA}). Thus combining \eqref{eq:gapBound}-\eqref{eq:gapRelations} yields
    \begin{align}
        \nrm{\CA}_{2-2}+\lambda_{\mathrm{gap}}(\CH)\ge\lambda_{\mathrm{Re}(gap)}(\CL)\ge \frac{\ln(2)}{t_{mix}(\CL)}. \tag*{\qedhere}
    \end{align}    
\end{proof}

The proof of the above proceeds by showing that $\lambda_{\mathrm{Re}(gap)}(\CL)\ge \frac{\ln(2)}{t_{mix}(\CL)}$, however without any prior knowledge on $\nrm{\CA}_{2-2}$ it does not seem to be possible to lower bound $\lambda_{gap}(\CH)$ in general.
Indeed, if the \Lword{} $\CL$ satisfies $ \vrho$-detailed balance with unitary drift, then $\CL$ and $\CD(\vrho,\CL)=\CH + \CA$ are related by a similarity transform and are thus co-spectral. Due to \autoref{prop:stationarityHA} we know that $V:=\text{Span}(\sqrt{\vrho})^\perp$ is an invariant subspace of both $\CH$ and $\CA$, and thus also of $\CL$. Since $\CH_{|V}\preceq - \lambda_{gap}(\CH) \CI$ and $\CA_{|V}$ is anti-Hermitian we get that every eigenvalue of $\CL^{\dagger}_{|V}$ has real part at most $-\lambda_{gap}(\CH)$, thus $\lambda_{\mathrm{Re}(gap)}(\CL)\geq\lambda_{gap}(\CH)$, which is an inequality in the ``wrong'' direction.

\subsection{$s$-detailed balance}
Our KMS detailed balance condition (\autoref{defn:sDB}) is a special case ($s=1/2$) of a larger family of $s$-detailed balance condition.
\begin{defn}[$s$-detailed balance condition]\label{defn:sDB}
For a normalized, full-rank state $ \vrho\succ 0$ and an scalar $0\le s\le 1$, we say that an super-operator $\CL$ satisfies $(s,\vrho)$-detailed balance (or $s$-DB in short) if
\begin{align}
     \CL^{\dagger}[ \cdot] =\vrho^{s-1}\CL[ \vrho^{1-s}\cdot  \vrho^{s}]  \vrho^{-s}.
\end{align}
\end{defn}
Since $\vrho$ is an operator, different choices of $0\le s\le 1$ yield different detailed balance conditions. Nevertheless, they all prescribe the same fixed point.
\begin{prop}[Fixed point]\label{prop:sDB_fixedpoint}
	If a \Lword{} $\CL$ is $(s,\vrho)$-detailed-balanced, then $\CL[\vrho] =0$.
\end{prop}
The case $s=0$ corresponds to the so-called Gelfand-Naimark-Segal (GNS) inner product. One naturally wonders if our constructions also apply here. Unfortunately, the case $s=1/2$ appears to be a special point in the interval $0\le s\le 1$. Let us revisit the energy domain representation of detailed balance (\autoref{prop:DB_energy_domain}) for the transition part. Consider a super-operator parameterized by a Hamiltonian $\vH$, $\beta$, and a set of operators including its adjoints $\{\vA^a\colon a\in A\}=\{\vA^{a\dagg}\colon a\in A\}$:
\begin{align}
    \CT = \sum_{a\in A} \sum_{\nu_1,\nu_2\in B}\alpha_{\nu_1,\nu_2} \vA^a_{\nu_1}(\cdot)(\vA^a_{\nu_2})^{\dagger}.
\end{align}
Then, the $s$-detailed balance condition 
\begin{align}
    \CT^{\dagger}[ \cdot] =\vrho^{s-1}\CT[ \vrho^{1-s}\cdot  \vrho^{s}]  \vrho^{-s} \quad \text{demands that}\quad \quad \alpha_{\nu_1,\nu_2}= \alpha_{-\nu_2,-\nu_1} \exp(-\beta (1-s)\nu_1-\beta s \nu_2 ) \quad \text{for each}\quad \nu_1,\nu_2 \in B.
\end{align}
However, this appears to be a strong condition; if we apply a change of variable $(\nu_1,\nu_2)\rightarrow (-\nu_2,-\nu_1)$,
\begin{align}
    \alpha_{-\nu_2,-\nu_1}= \alpha_{\nu_1,\nu_2} \exp(\beta (1-s)\nu_2+\beta s \nu_1 ) \quad \text{for each}\quad \nu_1,\nu_2 \in B.
\end{align}
Multiply the two to see that (see~\cite[Lemma 2.5]{carlen2017gradient} for an abstract argument)
\begin{align}
    \alpha_{\nu_1,\nu_2}\alpha_{-\nu_2,-\nu_1}= \alpha_{\nu_1,\nu_2}\alpha_{-\nu_2,-\nu_1} \exp(\beta (1-2s)(\nu_2-\nu_1) ) \quad \text{for each}\quad \nu_1,\nu_2 \in B.
\end{align}
Therefore, if $s \ne \frac{1}{2}$ and $\beta \ne 0$, we must have that
\begin{align}
    \alpha_{\nu_1,\nu_2} = 0 \quad \text{if}\quad \nu_1\ne\nu_2 \label{eq:nu1_nu2}
\end{align}
which contradicts our construction, and currently, we do not know how to algorithmically ensure~\eqref{eq:nu1_nu2}. The only existing Lindbladian~\cite{carlen2017gradient} that satisfies~\eqref{eq:nu1_nu2} is the Davies' generator~\cite{davies74}, which requires resolving the level spacing using a (exponentially) long Hamiltonian simulation time.
\section{Mathematica code for symbolically verifying our calculations}\label{apx:Mathematica}
This appendix displays the Mathematica code that symbolically verifies the heavy algebra calculations and Fourier transforms, as well as
\mmaDefineMathReplacement[≤]{<=}{\leq}
\mmaDefineMathReplacement[≥]{>=}{\geq}
\mmaDefineMathReplacement[≠]{!=}{\neq}
\mmaDefineMathReplacement[→]{->}{\to}[2]
\mmaDefineMathReplacement[⧴]{:>}{:\hspace{-.2em}\to}[2]
\mmaDefineMathReplacement{∈}{\in}
\mmaDefineMathReplacement{∉}{\notin}
\mmaDefineMathReplacement{ℝ}{\mathbb{R}}
\mmaDefineMathReplacement{∞}{\infty}
\mmaDefineMathReplacement{±}{\pm}
\mmaDefineMathReplacement{α}{\alpha}
\mmaDefineMathReplacement{β}{\beta}
\mmaDefineMathReplacement{γ}{\gamma}
\mmaDefineMathReplacement{σ}{\sigma}
\mmaDefineMathReplacement{ω}{\omega} 
\phantom{text}\\[-3mm]
\begin{itemize}[leftmargin=-1mm]
	\item Setting up proper parameters for the Fourier Transforms and globally applicable assumptions:
\begin{mmaCell}[functionlocal=x,mathreplacements=bold]{Code}
		SetOptions[FourierTransform, FourierParameters -> {0, -1}];
		SetOptions[InverseFourierTransform, FourierParameters -> {0, -1}];
		$Assumptions = \mmaUnd{β} > 0 && \mmaUnd{σ} > 0 && \mmaUnd{ω} ∈ ℝ && t ∈ ℝ;
\end{mmaCell}

	\item Verifying \eqref{eq:convolutionFilterFinite} by taking $c\leftarrow0$ and \eqref{eq:convolutionFilter} by additionally setting $s\leftarrow\infty$, finally \eqref{eq:convolutionFilterAPX} by resetting $c\leftarrow\beta\sigma_E$:
\begin{mmaCell}[functionlocal=x,mathreplacements=bold]{Code}
		Exp[-(\mmaUnd{ω} + x)^2/(2 (2x/\mmaUnd{β} - \mmaUnd{σ}^2))] Sqrt[2Pi(2x/\mmaUnd{β} - \mmaUnd{σ}^2)];
		Integrate[
\end{mmaCell}
\begin{mmaCell}{Output}
		1/2 E^(-(1/4)β(βσ^2 + 2ω + Abs[βσ^2 + 2ω])) 
			*(-Erf[c/2 - (β Abs[βσ^2 + 2ω])/(4c)] + Erf[s/2 - (β Abs[βσ^2 + 2ω])/(4s)]  
			+ E^(1/2 β Abs[βσ^2 + 2ω])*(	- Erf[c/2 + (β Abs[βσ^2 + 2ω])/(4c)] 
							 + Erf[s/2 + (β Abs[βσ^2 + 2ω])/(4s)])  )
\end{mmaCell}

\item Verifying the integral in the proof of \autoref{cor:transitionMetropolis}:
\begin{mmaCell}[functionlocal=x,mathreplacements=bold]{Code}
		Integrate[Exp[-2 x^2]/(x^2 + 1/16), {x, -Infinity, Infinity}] / Sqrt[32 Pi Exp[1/4]]
\end{mmaCell}
\begin{mmaCell}{Output}
		Sqrt[Pi/2] Erfc[1/Sqrt[8]]
		0.773389
\end{mmaCell}

	\item Verifying the inverse Fourier Transform copmutations at the end of the proof of \autoref{prop:B_timedomain}:
\begin{mmaCell}[functionlocal=w,mathreplacements=bold]{Code}
		InverseFourierTransform[Exp[-(2x + \mmaFnc{ω})^2/(16x/\mmaUnd{β})], \mmaFnc{ω}, t, Assumptions -> x > 0 && \mmaUnd{σ} > 0]
		InverseFourierTransform[Exp[-\mmaFnc{ω}^2/(8\mmaUnd{σ}^2)] Sinh[-\mmaUnd{β}\mmaFnc{ω}/4]/(2 I), \mmaFnc{ω}, t, 
								Assumptions -> \mmaUnd{β} > 0 && \mmaUnd{σ} > 0 && x > 0];
		Exp[-2 t^2 \mmaUnd{σ}^2] \mmaUnd{σ};
		InverseFourierTransform[1/Cosh[-\mmaUnd{β}\mmaFnc{ω}/4]/(2Pi), \mmaFnc{ω}, t] // FullSimplify
\end{mmaCell}
\begin{mmaCell}{Output}
		2 Sqrt[2] E^(-2tx (I + (2t)/β)))/Sqrt[β/x] 
\end{mmaCell}
\begin{mmaCell}{Output}
		-E^(-2t^2σ^2 + (β^2σ^2)/8) σ Sin[tβσ^2]
\end{mmaCell}
\begin{mmaCell}{Output}
		(Sqrt[2/Pi] Sech[(2Pi t)/β])/β
\end{mmaCell}

	\item Verifying the computations in the proof of \autoref{prop:MetropolisLikeFilterf}:
\begin{mmaCell}[functionlocal=x,mathreplacements=bold]{Code}
		Integrate[Exp[-4t^2x/\mmaUnd{β} - 2I\mmaUnd{t}\mmaFnc{x}]/Sqrt[2Pi], {x, 1/(2 \mmaUnd{β}), s^2/\mmaUnd{β}}]	
		Integrate[Sqrt[\mmaUnd{β}/(16Pi x)] Exp[-(\mmaUnd{ω} + 2 x)^2/(16x/\mmaUnd{β})], {x, \mmaUnd{β}\mmaUnd{σ}^2 /2, s^2/\mmaUnd{β}}, 
						Assumptions -> $Assumptions && s > \mmaUnd{β}\mmaUnd{σ}/Sqrt[2]];
		InverseFourierTransform[
		((
\end{mmaCell}
\begin{mmaCell}{Output}
		((E^(-((t (2t + Iβ))/β^2)) - E^(-((2s^2t (2t + Iβ))/β^2))) β)/(2 Sqrt[2Pi] t (2t + Iβ))
\end{mmaCell}
\begin{mmaCell}{Output}
		1/4 (-Erf[(βσ^2 + ω)/(2 Sqrt[2] σ)] + Erf[s/2 + (βω)/(4s)] 
			+ E^(-((βω)/2)) (Erf[s/2 - (βω)/(4s)] - Erf[(βσ^2 - ω)/(2 Sqrt[2] σ)])) 
\end{mmaCell}
\begin{mmaCell}{Output}
		1/4 (E^(-((β ω)/2)) Erfc[(βσ^2 - ω)/(2 Sqrt[2] σ)] + Erfc[(βσ^2 + ω)/(2 Sqrt[2] σ)])
\end{mmaCell}
\begin{mmaCell}{Output}
		(E^(-t(2t + Iβ)σ^2)β)/(2 Sqrt[2Pi] t(2t + Iβ)) + 1/2 Sqrt[Pi/2] DiracDelta[t]
\end{mmaCell}

\item Verifying \eqref{eq:tanhStep}:
\begin{mmaCell}[functionlocal=x,mathreplacements=bold]{Code}
		I/2 (2 - Exp[-\mmaUnd{β}\mmaUnd{ω}/2] - Exp[\mmaUnd{β}\mmaUnd{ω}/2])/(Exp[-\mmaUnd{β}\mmaUnd{ω}/2] - Exp[\mmaUnd{β}\mmaUnd{ω}/2]) // FullSimplify
\end{mmaCell}
\begin{mmaCell}{Output}
		I/2 Tanh[βω/4]
\end{mmaCell}

\end{itemize}
\end{document}